\documentclass[11pt,letterpaper]{article}
\usepackage[margin=1in]{geometry}
\usepackage[utf8]{inputenc}
\usepackage{amsmath,amssymb,amsthm,mathrsfs}
\usepackage{physics}
\usepackage{complexity}
\usepackage{float}
\usepackage[ruled,linesnumbered,noend]{algorithm2e}
\usepackage{setspace}
\usepackage[noend]{algpseudocode}
\usepackage[colorlinks=true]{hyperref}
\usepackage[dvipsnames]{xcolor}
\definecolor{darkred}  {rgb}{0.5,0,0}
\definecolor{darkblue} {rgb}{0,0,0.5}
\definecolor{darkgreen}{rgb}{0,0.5,0}
\hypersetup{
  urlcolor   = blue,         
  linkcolor  = darkblue,     
  citecolor  = darkgreen,    
  filecolor   = darkred       
}
\usepackage[nameinlink, capitalize]{cleveref}
\usepackage{tcolorbox}
\usepackage{subcaption}
\usepackage{tikz}
\usetikzlibrary{decorations.pathreplacing}
\pgfmathsetmacro\MathAxis{height("$\vcenter{}$")}

\renewcommand{\polylog}{\mathrm{polylog}}
\renewcommand{\poly}{\mathrm{poly}}
\usepackage[backend=biber,style=alphabetic,url=false,isbn=false,maxnames=10,minnames=3,maxalphanames=4,minalphanames=3]{biblatex}
\newcommand{\cnot}{\mathsf{CNOT}}

\renewcommand{\epsilon}{\varepsilon}

\newcommand{\under}[2]{\underbrace{#1}_{\substack{#2}}}
\newcommand{\error}{\text{-}\mathsf{error}}
\newcommand{\adj}{\text{-}\mathsf{adj}\text{-}\mathsf{error}}

\newcommand{\Enc}{\mathsf{Enc}}

\newcommand{\Dec}{\mathsf{Dec}}

\newcommand{\rec}{\mathsf{Rec}}
\newcommand{\ec}{\mathsf{EC}}
\newcommand{\meas}{\mathsf{Meas}}
\newcommand{\ho}{\mathsf{Hook}}

\newtheorem{properties}{Properties}
\newtheorem{theorem}{Theorem}[section]
\newtheorem{lemma}[theorem]{Lemma}

\newtheorem{definition}{Definition}[section]
\newtheorem{fact}{Fact}[section]

\newtheorem{corollary}[theorem]{Corollary}
\newcommand{\suppress}[1]{}
\SetKwInput{KwInput}{Input}
\SetKwInput{KwOutput}{Output}
\usepackage{algpseudocode}

\begin{document}

\title{A Constant Rate Quantum Computer on a Line}

\date{\today}

 \author{
Craig Gidney\thanks{Google Quantum AI, California, USA. \href{mailto:craig.gidney@gmail.com}{craig.gidney@gmail.com}}
\and 
Thiago Bergamaschi\thanks{UC Berkeley and Google Quantum AI, California, USA. \href{mailto:thiagob@berkeley.edu}{thiagob@berkeley.edu}}
}

\maketitle

\begin{abstract}
We prove by construction that the Bravyi-Poulin-Terhal bound on the spatial density of
stabilizer codes does not generalize to stabilizer circuits. To do so, we construct a fault tolerant quantum computer with a coding rate
above 5$\%$ and quasi-polylog time overhead, out of a line of qubits with nearest-neighbor connectivity, and prove it has a
threshold. The construction is based on modifications to the tower of Hamming codes of Yamasaki and Koashi (Nature Physics, 2024), with operators measured using a variant of Shor’s measurement gadget. 

%We prove by construction that the Bravyi-Poulin-Terhal bound on the spatial density of stabilizer codes doesn’t generalize to stabilizer circuits, contradicting a recent proof claiming the contrary. We construct a fault tolerant quantum computer with a coding rate above 5$\%$, and quasi-polylog time overhead, out of a line of qubits with nearest-neighbor connectivity, and prove it has a threshold. The construction is based on modifications to the tower of Hamming codes of Yamasaki and Koashi (Nature Physics, 2024), with operators measured using a variant of Shor’s measurement gadget. 
\end{abstract}

{
  \hypersetup{linkcolor=blue}
  \setcounter{tocdepth}{1}
  \tableofcontents
  
}

\newpage
\section{Introduction}

A ``stabilizer code'' is a simple way of describing a quantum error correction strategy~\cite{gottesman1997stabilizerformalism}.
A stabilizer code specifies a list of stabilizers (a set of commuting Pauli product operators) as well as a list of logical qubits (each an anticommuting pair of Pauli product operators).
The intent is that, by measuring the stabilizers, you can catch errors and thereby protect quantum information encoded into the logical qubits.

Stabilizer codes can't describe many forms of quantum fault tolerance. The stabilizers of a stabilizer code must all commute \cite{bacon2006operator}, so stabilizer codes can't represent the Bacon-Shor code~\cite{bacon2006operator}. The observables of a stabilizer code can't change, so stabilizer codes can't represent the Honeycomb code~\cite{hastings2021dynamically}. Stabilizer codes don't include information about \emph{how} to measure the stabilizers, so they can't describe concepts like ``measurement qubits''~\cite{fowler2012surfacecodereview}, ``flag qubits''~\cite{chao2018flags}, ``hook errors''~\cite{fowler2012surfacecodereview}, or measuring expensive operators less often~\cite{gidney2023lessbacon,gidney2024yoked}. Individually, these limitations can be fixed by extending the definition of a stabilizer code. But attempting to simultaneously include all the extensions necessary to fix these limitations, would destroy the simplicity that makes stabilizer codes so useful as a language for quantum error correction.

A ``stabilizer circuit'' is a quantum circuit built out of operations that can be efficiently analyzed using the stabilizer formalism~\cite{gottesman1997stabilizerformalism,aaronson2004chp}.
Stabilizer circuits can use Hadamard gates ($H$), controlled-not gates ($C_X$), phase gates ($S$), measurement gates ($M_Z$), reset gates ($R_Z$), and classical feedback.
This allows stabilizer circuits to implement arbitrary Clifford operations and arbitrary Pauli product measurements, but isn't sufficient for universal quantum computation without some additional ingredient (such as the ability to produce magic states~\cite{bravyi2005distillation}). Although they are strictly more complicated than stabilizer codes, stabilizer circuits can represent a wider variety of fault tolerant constructions \cite{gidney2021stim,mcewenmidoutsurfaces2023,delfosse2023spacetimecode,bombin2024zxunify}.

%Stabilizer circuits can also be used as a language for quantum fault tolerance~\cite{gidney2021stim,mcewenmidoutsurfaces2023,delfosse2023spacetimecode,bombin2024zxunify}.
%A circuit can be annotated with ``detectors'' (sets of measurements with a known parity under noiseless execution), and these detectors can be used to protect logical observables.
%Stabilizer circuits are more complicated than stabilizer codes, but can represent a wider variety of fault tolerant constructions.
%For example, they can describe flag qubits, and can describe checking inconvenient stabilizers less often.

Bounds on quantum error correction are often proved first for stabilizer codes.
For example, in a seminal result, Bravyi, Poulin, and Terhal proved a bound on the rate and distance of stabilizer codes whose stabilizers are local on a two dimensional grid.
In particular, they proved that 2D $[[n,k,d]]$ stabilizer codes must satisfy $kd^2 = O(n)$~\cite{bravyi2010stabilizerbound2}.\footnote{Where $n$ is the number of physical qubits, $k$ is the number of logical qubits, and $d$ is the code distance.}
In 2010, Bravyi extended this work to gauge codes~\cite{bravyi2011gaugebound}.
Crucially, he showed that the \cite{bravyi2010stabilizerbound2} bound \emph{didn't} apply to gauge codes.
Instead, gauge codes are restricted by the looser bound of $kd = O(n)$.
This naturally raises the question: does this bound continue to loosen as more possibilities are considered?
In particular, How do these bounds on spatial density generalize from 2D local codes to 2D local \emph{circuits}? Quantum circuits can use a series of local gates to measure non-local stabilizers, so this isn't a trivial question.

\begin{comment}
Recently, Baspin, Fawzi, and Shayeghi attempted to extend Bravyi et al's work to stabilizer circuits~\cite{baspin2023circuitbound}.
They claim to prove that any stabilizer circuit in 2D with logical error rate $\delta$ must have a spatial overhead (an inverse rate) of at least $\Omega(\sqrt{\log \delta^{-1}})$.
Conceptually, their argument is based on the fact that high rate quantum error correcting codes must be highly entangled.
Any error-correction circuit must create and maintain this long-range entanglement, but a geometrically-local noisy quantum circuit is limited in how quickly and reliably it can create it. 
\end{comment}

Unfortunately, the true limits of the space \& time overheads of noisy stabilizer circuits still remain much less well understood. Recently, Baspin, Fawzi, and Shayeghi extended Bravyi et al's work to certain families of stabilizer circuits~\cite{baspin2023circuitbound}. They proved that any stabilizer circuit in 2D of logical error rate $\delta$ - with a shallow decoding circuit (cf. \cref{section:related}) - must have a spatial overhead (an inverse rate) of at least $\Omega(\sqrt{\log \delta^{-1}})$.
Conceptually, their argument is based on the fact that high-rate quantum error correcting codes must be highly entangled.
Any error-correction circuit must create and maintain this long-range entanglement, but a geometrically-local, noisy quantum circuit is limited in how quickly and reliably it can do so.

In this paper, we prove that the \cite{bravyi2010stabilizerbound2} bound  cannot fully generalize to stabilizer circuits. To do so, we construct a constant-rate quantum memory, whose operations can be fully realized using nearest-neighbor gates in 1 dimension (on a line). Our construction is built on the concatenated ``tower of Hamming codes" of Yamasaki and Koashi \cite{Yamasaki_2024}, and leverages multi-scale error correction to bypass the no-go result of \cite{baspin2023circuitbound}. By further combining our memory with a magic state distillation protocol, we show how to achieve fault-tolerant quantum computation, with a constant space overhead, in 1D. 

%\footnote{Instead, we reason that their results can better be interpreted as a time-space tradeoff for the syndrome extraction circuit (cf. \cref{app:baspin}).}

\subsection{Our Contributions}

A technical statement of our contributions follows. We refer the reader to \cref{section:preliminaries} for formal definitions. Ultimately, our main result is the following theorem.

\begin{theorem}[A Constant Rate Quantum Memory]\label{theorem:main}
    There exists an infinite family of quantum memories $M_n$, such that 
    \begin{enumerate}
    \item $M_n$ uses $n$ physical qubits to implement more than $n/20$ logical qubits.
    \item $M_n$ is implemented by a stabilizer circuit using nearest-neighbor gates on a line, while subjected to local stochastic noise.
    \item Below some local stochastic noise threshold $p$ independent of $n$, the per-circuit-cycle logical error rate of $M_n$ is $$\exp(-\exp(\Omega(\log^{1/3} n))).$$
    \end{enumerate}
\end{theorem}

We assume that $M_n$ is decoded by an efficient (polynomial-time bounded) classical control system, that uses long range communication and perfect operations. Only the quantum part of the memory is local and noisy. However, we emphasize that no classical \textit{feedback} is required for the purposes of simply maintaining the logical quantum information as in \cref{theorem:main}. We simply use the classical control system to store and update syndrome measurements.%\footnote{Since the entire memory is operated via a stabilizer circuit, all Pauli corrections can be deferred to the end.}

To operate on our memory $M_n$, we design fault-tolerant gadgets to perform arbitrary Pauli-product measurements. This enables us to measure stabilizers, logical information, and more generally to perform arbitrary Clifford gates. To achieve a universal set of gates, we combine our memory with a magic state distillation protocol \cite{bravyi2012}. Ultimately, we prove the following theorem on the fault-tolerant simulation of 1D quantum circuits, with a constant space overhead and quasi-$\polylog$ time overhead. 

\begin{theorem}[Fault-Tolerant Computation]\label{theorem:main-ft}
    Let $C$ be a quantum circuit on $m$ qubits configured on a line, which can be implemented using $d$ alternating layers of nearest neighbor, two qubit gates. Then, for any desired target accuracy $\epsilon$ bounded by
    \begin{equation}
        \epsilon \geq d\cdot \exp(-\exp(O(\log^{1/3} m))),
    \end{equation}

    \noindent there exists a fault-tolerant simulation $C_\epsilon$ of $C$, satisfying
    
    \begin{enumerate}

    \item $C_\epsilon$ uses $\leq 20\cdot m$ physical qubits to implement $m$ logical qubits. 

    \item $C_\epsilon$ is implemented by a depth $d\cdot \exp(O(\log^3\log\frac{m\cdot d}{\epsilon}))$ stabilizer circuit using nearest neighbor gates on a line, while subjected to local stochastic noise.
    
    \item Below some threshold noise rate $p$ independent of $m$, the logical error rate of $C_\epsilon$ is $\epsilon$.
    \end{enumerate}
\end{theorem}

Here, we operate in a computational model where the classical control system can perform feedback operations on the memory, and that quasi-$\polylog$ time classical computation is \textit{instantaneous}; albeit this latter assumption can readily be removed using idling gates \cite{Yamasaki_2024}. Precise definitions of the model of computation, correctness, and the noise model, are made in \cref{section:preliminaries}.

\begin{figure}
    \begin{subfigure}[b]{0.55\textwidth}
\centering
    \includegraphics[width = 1\linewidth]{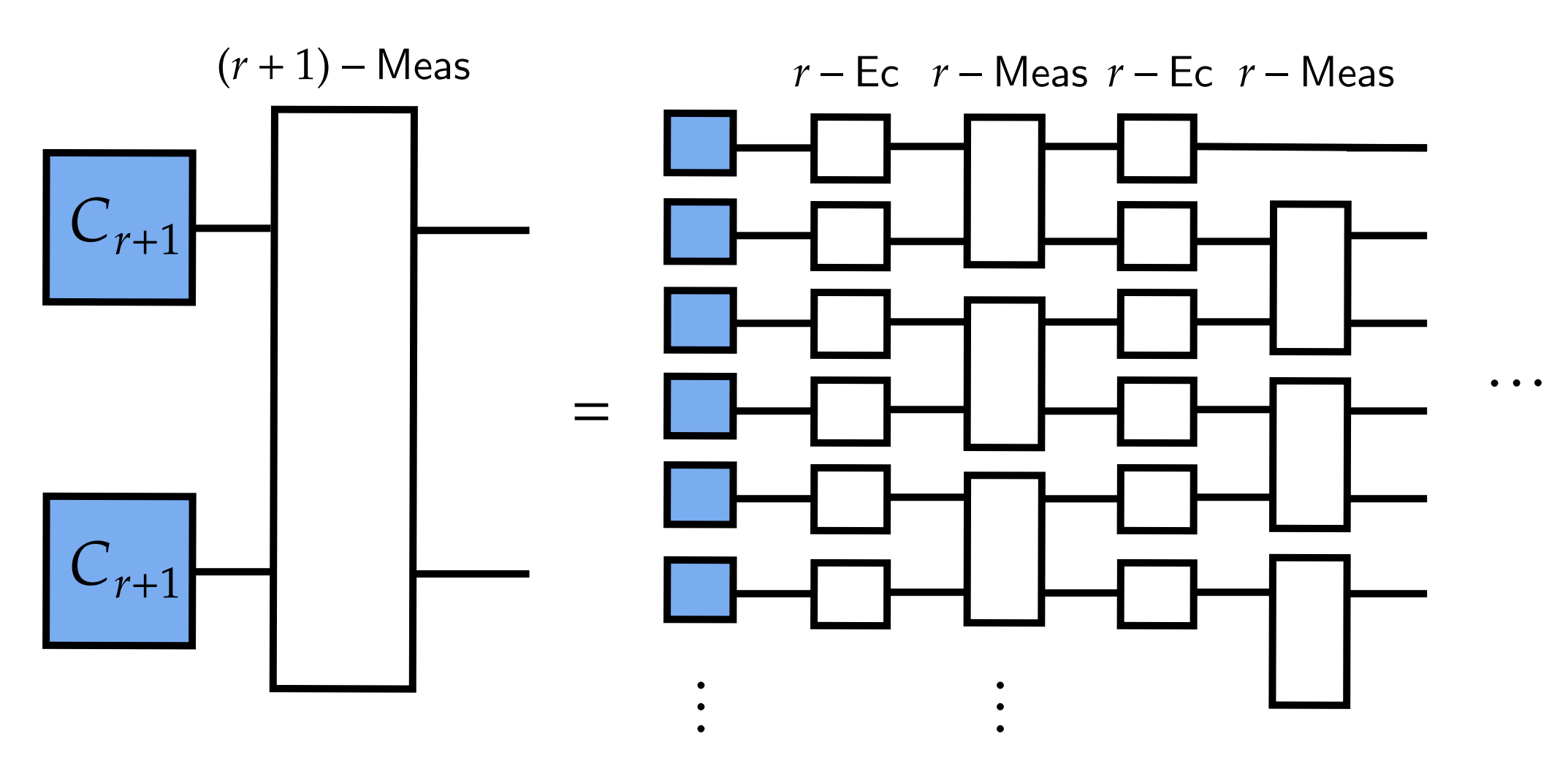}
    \caption{\footnotesize The Concatenated Simulations Framework}
\end{subfigure}
\begin{subfigure}[b]{0.45\textwidth}
\centering
    \includegraphics[width = 1\linewidth]{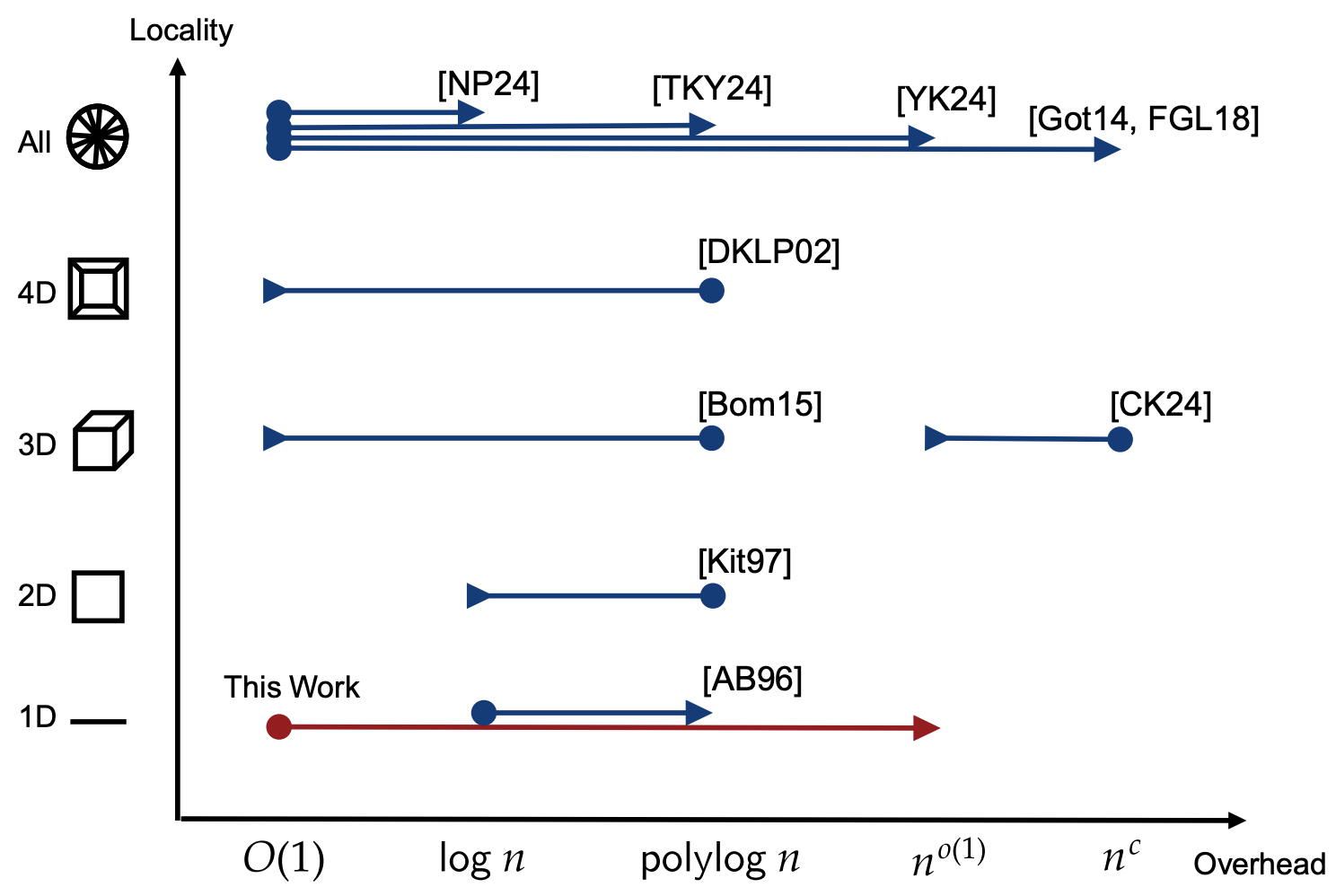}
    \caption{\footnotesize Time/Space Overhead of various FT Schemes}
\end{subfigure}
\caption{ (a) Measurements on the $(r+1)$st level of the code, are implemented recursively by alternating layers of $r$th level measurements and error correction rounds. (b) A comparison of the locality (1, 2, 3, 4D and all-to-all), the space overhead (circles), and the time overhead (triangles) of various quantum fault-tolerance schemes. }
\label{fig:related_work}
\end{figure}

\subsection{Related Work}
\label{section:related}

Our construction follows the \textit{concatenated simulations} framework originally introduced by \cite{von1956probabilistic,Gcs1983}, its adaptation to quantum fault-tolerance by \cite{Aharonov1996FaulttolerantQC, aliferis2005quantum}, and, in particular, the recent developments by \cite{Yamasaki_2024}. In this framework, a small (constant-sized) error-correcting code is repeatedly concatenated with itself, with the intent to perform a fault-tolerant simulation of some computation. At any level $r$ of concatenation, all the physical (2-qubit) gates in the circuit are replaced by constant-sized ``fault-tolerant gadgets" to define level $(r+1)$. In this sense, the $(r+1)$ level is simulating a fault-tolerant execution of the level $r$, and each level of the hierarchy becomes a more and more reliable simulation of the original computation. 

We remark that from these original works it is already well known that (quantum) concatenated codes can implement fault-tolerant quantum memories even in 1D \cite{Aharonov1996FaulttolerantQC, Gottesman_2000, Svore_2005,svore2006}.
However, their rate is inverse poly-logarithmic; not constant. \\

\noindent \textbf{The Yamasaki and Koashi quantum fault-tolerance scheme.} Most relevant to our work are the results of Yamasaki and Koashi \cite{Yamasaki_2024}. They devised a quantum fault-tolerance scheme with constant space- and quasi-polylog time-overhead, via a certain \textit{interleaved concatenation} of quantum Hamming codes of increasing rate. As we discuss below, our construction builds on their ``tower of Hamming codes" (and inherits its parameters), however, we perform gates, measurements and logic on the code using very different means. 

The scheme in \cite{Yamasaki_2024} is implemented using non-local connections.\footnote{Yamasaki \& Koashi (Nature Physics, 2024) \cite{Yamasaki_2024} claim in passing that one could embed their scheme in 2D using \cite{Aharonov1996FaulttolerantQC, Gottesman_2000}. However, doing so with constant space overhead has remained open. The authors cite \cite{baspin2023circuitbound} to claim that ``the constant space overhead would not be achievable on a single fully 2D chip"; see below.} Recently, using techniques from fault-tolerant routing, Choe and König \cite{Choe2024HowTF} showed how to embed the \cite{Yamasaki_2024} scheme in 3D with just a constant factor increase in depth, however, with a $\poly(n)$ blowup to the qubit count. In this work, we essentially achieve the best-of-both-worlds in that we accomplish a constant space overhead, while still a quasi-polylog time-overhead, in 1D.\\

\noindent \textbf{Other quantum fault-tolerance schemes}. In \cref{fig:related_work}b above, we plot the locality, and time-space overheads of various quantum fault-tolerance schemes. \cite{Kitaev1997QuantumEC, Bomb_n_2015, Dennis_2002} discussed schemes for quantum fault-tolerance with polylog space overhead based on topological quantum error-correction. Notably, \cite{Bomb_n_2015, Dennis_2002} achieve constant time overhead by leveraging single-shot decoders. In concurrent work, \cite{tamiya2024polylogtimeconstantspaceoverheadfaulttolerantquantum} developed a constant-space polylog-time scheme, and notably \cite{nguyen2024} devised a constant-space logarithmic-time overhead scheme, using recent developments in locally testable codes \cite{dinur2024expansionhigherdimensionalcubicalcomplexes}. They are the first to improve on the time overhead of \cite{Yamasaki_2024}'s scheme while maintaining constant space; however, it remains to be seen whether their ideas can be implemented in low dimensions. 

\cite{pattison2023hierarchicalmemoriessimulatingquantum} devised a 2D (a bilayer) quantum memory with inverse polylog rate, based on concatenating a good quantum LDPC code with a surface code. They similarly avoid \cite{baspin2023circuitbound}'s no-go result, by leveraging (deep) syndrome extraction circuits based on qubit routing. \cite{balasubramanian2024localautomaton2dtoric} devised a local decoder for the 2D toric code, by embedding a concatenated classical automaton into the decoder (akin to \cite{Cirelson1978ReliableSO,Gcs1983}). Notably, they do not require a noiseless classical computer operating on the memory. Repeating their scheme in parallel gives rise to a memory with inverse polylog rate.\\

\noindent \textbf{Obstructions in implementing quantum error-correction in low dimensions.} Since the seminal results of \cite{bravyi2009stabilizerbound, bravyi2010stabilizerbound2,Haah2020ADB} it is well known that quantum error-correcting codes, when implemented in low dimensions, suffer from fundamental limitations. This has led to a fruitful line of work refining their bounds \cite{Baspin2021QuantifyingNH,Hong2023LongrangeenhancedSC,Fu2024ErrorCI,Dai2024LocalityVQ}, and searching for matching constructions \cite{Portnoy2023LocalQC,Lin2023GeometricallyLQ,Williamson2023LayerC}. A related line of work lies in how locality limits \textit{circuits} that implement quantum error-correcting codes. \cite{delfosse2021boundsstabilizermeasurementcircuits} established lower bounds on the depth of syndrome measurement circuits for qLDPC codes in 2D. \\

\noindent \textbf{Do our results contradict \cite{baspin2023circuitbound}?} Baspin, Fawzi, and Shayeghi extended \cite{bravyi2010stabilizerbound2} work to certain families of noisy stabilizer circuits~\cite{baspin2023circuitbound}. In more detail (Theorem 28), they proved that stabilizer circuits in 2D encoding $k$ logical qubits via $n$ physical qubits, with decoders of depth $\Delta$ of logical error rate $\delta$, are limited to the following tradeoff: $k\cdot \sqrt{\log \delta^{-1}} \leq O(n\cdot \Delta)$. Their result can be interpreted as a time-space trade-off for the decoding channels of quantum memories when implemented in low-dimensions, with implications to their syndrome measurement circuits and to certain classes of quantum fault-tolerance schemes, akin to \cite{delfosse2021boundsstabilizermeasurementcircuits}. As our concatenated codes have deep syndrome measurement circuits, we avoid their no-go result.

\subsection{Techniques}

We dedicate this section to an overview of our memory construction, an outline of the correctness proof, and the basic idea behind our magic state distillation scheme. \\

\noindent \textbf{The Tower of Hamming Codes.} To achieve a constant space overhead, \cite{Yamasaki_2024} revisited the concatenation techniques of \cite{Aharonov1996FaulttolerantQC, aliferis2005quantum} under the \textit{interleaved} concatenated of quantum Hamming codes (of distance $3$). The \textit{interleaved concatenation} of an ``outer" $[[n_1, k_1]]$, and ``inner" $[[n_2, k_2]]$ stabilizer code creates $k_2$ copies of the former, and $n_1$ copies of the latter, and for $i\in [n_1]$ routes the $i$th physical qubit of each copy of the former into the $i$th copy of the latter (See \cref{fig:basicops}a).\footnote{While slightly non-standard, this operation enables the concatenation of \textit{any} two codes without sacrificing rate, which is not true under the standard definition \cite{Forney}.}

Concatenating a Hamming code with itself again and again would create codes with coding rates closer and closer to 0 as the amount of concatenation increased. However, if you concatenate Hamming codes without using the same Hamming code twice, then the coding rate of the concatenation is bounded away from 0.  For example, $H_m \otimes H_{m-1}\otimes \cdots H_5\otimes H_4$ converges to a coding rate of roughly 20$\%$ as $m$ increases:
\begin{equation}
    \lim_{m\rightarrow \infty} \text{Rate}\bigg(\bigotimes_{i=4}^m H_i\bigg) = \lim_{m\rightarrow \infty}\prod_{i=4}^m\frac{2^i-2i-1}{2^i-1}\approx 19.7\%.
\end{equation}

This interleaved concatenation scheme loses much of the modular and ``self-similar" structure of recursive concatenation. What is more, Svore, DiVicenzo and Terhal \cite{svore2006} noted that it is impossible to directly implement a concatenated distance 3 code in a one-dimensional architecture (see below)\footnote{``A 1D architecture necessitates swapping data qubits inside a [code] block; which may generate two-qubit 
errors on [adjacent] qubits due to one failed SWAP.  For a distance-3 code, such errors cannot be corrected" \cite{svore2006}. }. To implement operations in 1D, we will further have to introduce a series of modifications to the tower of Hamming codes of \cite{Yamasaki_2024}. \\

\noindent \textbf{Our Modifications to the Tower.} 
First and foremost, at each level of concatenation, a logical qubit will be reserved.
These reserved logical qubits will later be used to store cat states, and mediate long range measurements.
The second, and key modification is motivated by the following issue: at each level of concatenation, we will be performing logical two qubit operations, that cross between adjacent code blocks, potentially causing their failures to correlate. To mitigate this problem, we design codes that are resilient not just to individual data errors but to simultaneous \textit{adjacent} data errors. For this purpose, very roughly speaking, at each level of concatenation we will interleave the code with a copy of itself (See \cref{section:interleaved_memory}, \cref{fig:tower_of_hamming}). As we discuss below and extensively in \cref{section:threshold}, this will later prevent faulty operations that act on adjacent code-blocks (in 1D) from simultaneously breaking two of the underlying data qubits. \\

%at each level of concatenation, we will be performing logical two qubit operations, that cross between adjacent code blocks. 

\noindent \textbf{Fault-tolerant Operations, in 1D.}
Broadly speaking, all operations on the memory are performed via some form of Pauli-product measurement gadget.
For this purpose, at each level of concatenation $r\geq 1$, we will introduce three types of gadgets:
an ``error-correction gadget'' $r$-$\ec$,
a ``fault tolerant measurement gadget'' $r$-$\meas$,
and a ``Hookless measurement gadget'' $r$-$\ho$.

$r$-$\ec$ is used to measure code stabilizers of an $r$-level code-block.
Its goal is to prevent lower-level errors from different parts of the computation from combining into higher-level errors.
$r$-$\meas$ is used to perform measurements of logical observables on one or two adjacent $r$-level code-blocks.
This is the operation the ``user'' of the code would use to implement logic, which should be more reliable than the fault tolerant measurements from the level below.
It is the operation that will be used by the level above, to implement its functionality.

The key building block to construct these gadgets, will be the hookless measurement $r$-$\ho$. $r$-$\ho$ is a non-fault-tolerant measurement whose implementation lacks ``hook errors''.
That is to say, it might output the wrong measurement but it won't damage the code in a way that local operations would not. 
$r$-$\ho$ will be implemented using a variation of Shor's \cite{shor1996faulttolerance} measurement gadget, which is in turn implemented recursively using alterning rounds of $(r-1)$-$\meas$ and $(r-1)$-$\ec$ (See \cref{fig:related_work}a).
In Shor's gadget, cat states $\ket{0^t}+\ket{1^t}$ are produced and used to mediate non-local measurements.
Cat states are a form of long-range entanglement, which nevertheless can be prepared using local measurements (and Pauli feedback, \textit{or}, tracking the Pauli correction).

We emphasize that how we perform logic on the memory is a key distinction from our work to that of \cite{Yamasaki_2024}. In implementing all our operations with Pauli-product measurements, we do not need to rely on post-selection (even for state-preparation). The memory can simply passively perform syndrome measurements.\\

\noindent \textbf{Correctness, $r\adj$s and Error-Propagation Properties.} To prove correctness in the presence of noise, we need to develop the mechanism through which errors arise and propagate throughout the circuit. We roughly follow the ``extended Rectangles" approach of \cite{aliferis2005quantum} on concatenated distance 3 codes however, because the codes we are concatenating are block codes, we have to be more cautious about the structure of errors \cite{Yamasaki_2024}.
Further, in the tower of concatenated codes, the $r$-level code $C_{r}$ will be built out of instances of $C_{r-1}$ in a side-by-side layout.
Because of the 1D constraint, it is difficult to interact adjacent $C_{r-1}$ codes without risking simultaneously destroying the entirety of both codes.
For this purpose, we introduce the key concept of an $r\adj$ (or, adjacent pair of $r\error$s).

\begin{figure}[t]
    \begin{subfigure}[b]{0.35\textwidth}
\centering
    \includegraphics[width = .5\linewidth]{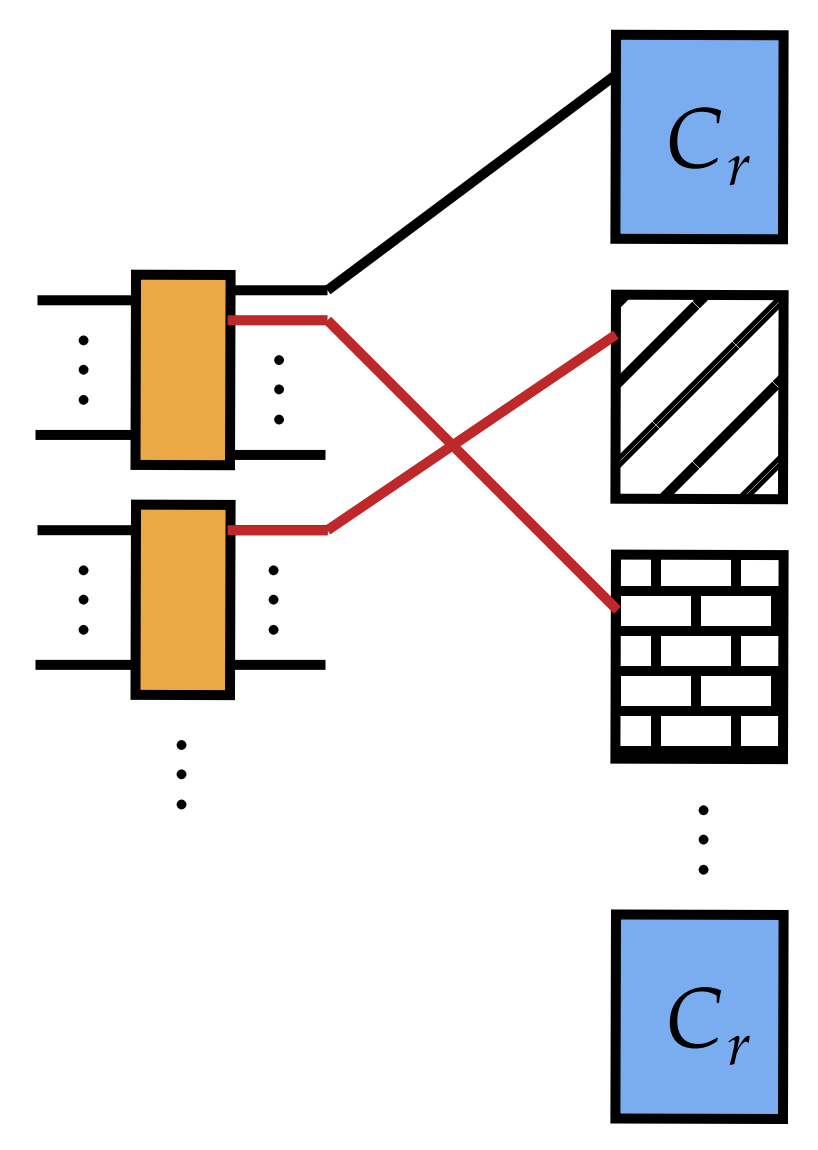}
    \caption{\footnotesize An $(r+1)\adj$}
\end{subfigure}
\begin{subfigure}[b]{0.65\textwidth}
\centering
    \includegraphics[width = 1\linewidth]{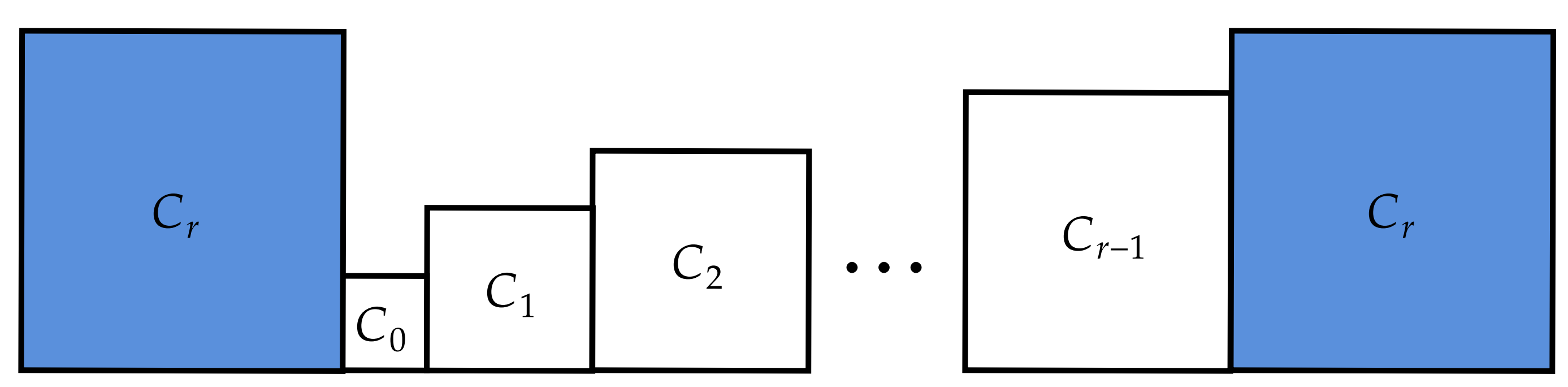}
    \caption{\footnotesize The ladder of code-blocks}
\end{subfigure}
\caption{(a) An $(r+1)$ code-block, comprised of the interleaved concatenation of Hamming codes (orange) and $r$ code-blocks (blue). In white, two adjacent $r$-blocks are corrupted (an $(r+1)\adj$), which correspond to single qubit faults on the underlying Hamming code. (b) $T\ket{+}$ states are teleported into $r$-blocks by injection into the physical level ($C_0$), and teleported sequentially up the ladder ($C_0\rightarrow C_1\rightarrow \cdots $).}
\label{fig:intro_ladder}
\end{figure}

Recursively, we say that $C_{r}$ contains an $r\error$ if one of its $C_{r-1}$ blocks contains at least two, \textit{non-adjacent}, $(r-1)\error$s.\footnote{A $0$-error is simply a Pauli error on a physical qubit of the code.} Two $r\error$s are said to be adjacent if they lie on nearest
neighbor $C_{r-1}$ blocks; we give this pattern of errors a special name: an $r\adj$.
The key feature of $r\adj$s is that they are correctable by the level $r$ code, even though other patterns of two $r\error$s aren't (See \cref{fig:intro_ladder}a). 

The correctness of the entire scheme then hinges on how these errors are created, propagated, and corrected within the circuit. Arguably, the main technical part of this paper lies in describing a minimal set of properties on the measurement gadgets ($r$-$\ec$, $r$-$\meas$) at level $r$, which, if satisfied, ensures that such $r\adj$'s don't proliferate; i.e., guarantees the \textit{sparseness of faults} during the execution of the circuit \cite{Gcs1983, Aharonov1996FaulttolerantQC, aliferis2005quantum}. Roughly speaking, these conditions quantify the behavior of each gadget under either faulty inputs (but fault-free execution) \textit{or} faulty-execution (but perfect inputs); see \cref{section:threshold} for details. \\

\noindent \textbf{Magic State Distillation, via $r$-block Teleportation.} Given the ability to implement arbitrary Pauli product measurements (and initialize ancilla qubits), one can implement arbitrary Clifford operations \cite{horsman2012latticesurgery}. To turn our memory into a computer, it suffices (via \textit{gate injection}) to design a protocol which enables us to produce logical $T\ket{+}$ magic states encoded into an $r$-block. 

For this purpose, we perform a minor modification to our memory layout.  In between the top, $r$-level code-blocks which store data-qubits, we place a ``ladder” of code blocks $(C_0, C_1, \cdots, C_{r-1})$ adjacent to each other (See Fig. 2b). Our goal, roughly speaking, is to teleport noisy $T\ket{+}$ states which are injected into $C_0$, all the way up the ladder into a logical qubit of $C_r$. Although we defer details to \cref{section:the_computer}, roughly speaking, the
protocol is based on establishing a logical EPR pair between an $i$-block and an adjacent $(i+1)$-block, which enables the teleportation of logical qubits between the blocks. Once sufficiently many noisy $T\ket{+}$ states are teleported into the top-level $r$-block, we run a magic state distillation protocol \cite{bravyi2012} within the $r$-block to acquire a high-fidelity $T\ket{+}$ state.

Arguably, the conceptual challenge in establishing the teleportation protocol is that our framework for Pauli-product measurements only enables us to perform logical operations between code-blocks at the same concatenation level. In \cref{section:the_computer}
 we show how to adequately modify the measurement scheme to allow measurements between code-blocks at different levels of concatenation, and prove that the entire protocol teleports $T\ket{+}$ states into the top-level $r$-block with just a constant noise rate - sufficient for distillation.

\subsection{Organization}

In \cref{section:preliminaries}, we discuss terminology and basic operations on stabilizer codes. In \cref{section:interleaved_memory}, we present the code construction, and compute its basic properties. In \cref{section:fault-tolerance}, we discuss how to implement fault-tolerant stabilizer measurements, and logical measurements. 

In \cref{section:threshold}, we introduce a sufficient set of error-propagation properties for these fault-tolerant operations to ensure their correctness, and prove a threshold theorem for our construction. In \cref{section:propagation_proofs}, we inductively prove that the fault-tolerant operations described in \cref{section:fault-tolerance} admit said properties, with the base case presented in \cref{section:base_case}. In \cref{section:proof_of_main}, we put all our results together and prove the main result of \cref{theorem:main}.

In \cref{section:the_computer}, we present our state distillation scheme and prove \cref{theorem:main-ft}. We discuss contributions in \cref{section:contributions}, and conclude in \cref{section:conclusion}.

\section{Preliminaries}
\label{section:preliminaries}

\subsection{Terminology}

An $[[n, k, d]]$ \textbf{stabilizer code} encodes $k$ logical qubits into $n$ physical qubits, and represents the $+1$ eigenspace of a commuting set of $n$ Pauli operators $\in \{\mathbb{I}, X, Y, Z\}^{\otimes n}$, its \textbf{stabilizers}. The \textbf{distance} $d$ of the stabilizer code is the minimum number of Pauli errors needed to flip one or more of the code’s logical observables, without flipping any of the code’s stabilizers.

A \textbf{stabilizer circuit} is a quantum circuit consisting of clifford operations, Hadamard, Phase, and controlled-not gates
\begin{equation}
    H = \frac{1}{\sqrt{2}}\begin{bmatrix}
        1 & 1\\
        1 & -1
    \end{bmatrix}, \quad S = \begin{bmatrix}
        1 & 0\\
        0 & i
    \end{bmatrix}, \quad \cnot = \begin{bmatrix}
        1 & 0 & 0 & 0\\
        0 & 1 & 0 & 0\\
        0 & 0 & 0 & 1\\
        0 & 0 & 1 & 0\\
    \end{bmatrix}
\end{equation}

\noindent as well as computational-basis measurements $M_Z$, reset gates $R_Z$, and classical feedback. The \textbf{fault distance} of a stabilizer circuit, given a noise model, is the minimum number of
errors from the model needed to flip a logical observable. 

In this work we phrase our proofs in the \textbf{local stochastic noise} model, where at each layer of computation an $n$ qubit channel $\mathcal{D}_p$ is applied, which randomly picks a (possibly correlated) subset of qubits, but is allowed to apply an adversarial channel to said qubits. Formally, the randomness over the choice of subset satisfies:

\begin{equation}
    \forall S\subset [n]:\quad \mathbb{P}[S\subset \text{Supp(Error)}] \leq p^{|S|}
\end{equation}

A special case which is helpful to build intuition is the \textbf{depolarizing noise} model (of rate $p\in (0, 1)$), defined by the single-qubit quantum channel which acts on a quantum state $\rho$ via a random Pauli operator:
\begin{equation}
    \mathcal{N}_p(\rho) = (1-p)\rho + \frac{p}{3}\big(X\rho X + Y\rho Y + Z\rho Z),
\end{equation}

\subsection{Model of Computation}

We make the following assumptions on the computational model.\\

\noindent \textbf{Classical operations are non-local and noiseless.} As raised previously, we remark that in designing a quantum memory based on stabilizer circuits, all Pauli feedback operations can be deferred to post-processing after the final measurement. Thereby, we do not need to assume classical operations are instantaneous for the purposes of \cref{theorem:main} (nor any feedback operations at all). However, to design a fault-tolerant quantum computer, we assume instantaneous (quasi-polylog time) classical operations. \\

\noindent \textbf{Quantum operations are noisy, nearest-neighbor, and parallelizable.} At the physical level, we assume operations proceed in layers, where each layer is comprised of arbitrary nearest-neighbor unitary gates and single-qubit measurements. After each layer, we subject all the qubits to depolarizing noise. \\

We remark that this entails measurement outcomes are subject to noise, but once the outcome is recorded, it is classically stored without faults. \\

\noindent \textbf{Correctness.} Here we formally define a model of correctness for our quantum memory and quantum computer. Let $\mathcal{X}$ denote the classical memory register and $\mathcal{Q}$ denote the $n$ qubit quantum register. We model the execution of the memory via alternating layers of 2 qubit gates and measurements, expressed via a sequence of separable channels $\{W_{i, i\pm 1}^t\}_{t\in [T], i\in [n]}$, up to some time $T$. Each $W_{i, i\pm 1}^t$ acts on nearest neighbor qubits $i$, $i\pm 1$, in addition to the classical register $\mathcal{X}$. 

This circuit is said to define a quantum memory with per-circuit-cycle error $\delta$, if there exists idealized (noiseless) decoding/encoding channels $\Enc, \Dec$ acting on $\mathcal{X}, \mathcal{Q}$ such that

\begin{equation}
    \Dec\circ  \under{\prod_t^T \mathcal{D}_p \circ  \bigg(\bigotimes_{\text{even }i} W_{i, i\pm 1}^t \bigg) \circ \mathcal{D}_p \circ \bigg( \bigotimes_{\text{odd }i}  W_{i, i\pm 1}^t \bigg)}{\text{the noisy circuit}} \circ \Enc(\psi) \approx_{\delta \cdot T} \psi,
\end{equation}

\noindent for all message-states $\psi$ on $\mathcal{Q}$ (possibly entangled with some reference system $\mathcal{R}$), and where the distance is measured in trace distance.\footnote{Equivalently, the effective channel is $\delta\cdot T$ close to the identity channel in diamond distance.} This idealized model of correctness is akin to the correctness model of certain fault-tolerance proofs (like \cite{aliferis2005quantum}), but also arises in the "self-correcting quantum memory" literature \cite{Alicki2008OnTS}. 

Consider next a generic quantum computation $C$, consisting of a series of gates acting initially on the $\ket{0}^{\otimes k}$ product state, and concluded with single qubit measurements. We say $C$ is simulated to logical error $\delta$ if there exists an analogous sequence of separable channels (on adjacent qubits and $\mathcal{X}$) acting initially on $\ket{0}^{\otimes m}$, which concludes with single-qubit measurements, such that the classical measurement information can be efficiently post-processed into a sample $\delta$ close to a sample from the measurement outcome distribution of $C$. The setting of fault-tolerant computation is naturally strictly harder than that of a memory; in that one is expected to perform fault tolerant state preparation, logical measurements, and arbitrary gates.

\subsection{Basic Transformations on Stabilizer Codes}

\begin{definition}
    [$\mathsf{Reserve}_1$] Reserving a logical qubit of a code $C$ produces a code $\mathsf{Reserve}_1(C)$ with one fewer logical qubit.
\end{definition}

\noindent The intent being that the lost logical qubit will be used to support (non-local) measurements (\cref{fig:basicops},b), by storing cat state qubits.  

\begin{definition}
    [Interleaved Concatenation]\label{def:interleaved_concatenation} Let $A, B$ be $[[n_A, k_A, d_A]]$ and $[[n_B, k_B, d_B]]$ stabilizer codes respectively. The \emph{interleaved concatenation} $A\otimes B$ is the $[[n_A\cdot n_B, k_A\cdot k_B, d_A\cdot d_B]]$ stabilizer code, defined on $k_B$ copies of $A$ and $n_A$ copies of $B$, where each physical qubit $i\in [n_A]$ of each copy of $A$, is encoded into the $i$th copy of $B$ (See  \cref{fig:basicops}a).
\end{definition}

The distance follows from the observation that a logical error is imparted to a copy of the ``outer code" $A$ only if there is a logical error on at least $d_A$ of the ``inner" blocks $B$. The \textbf{inner stabilizers} of $A\otimes B$ are the stabilizers of the various copies  of $B$. The \textbf{outer stabilizers} of $A\otimes B$ are the stabilizers of the copies of $A$, represented by the physical qubits of the copies of $B$.

\begin{figure}
    \begin{subfigure}[b]{0.5\textwidth}
\centering
    \includegraphics[width = 0.5\linewidth]{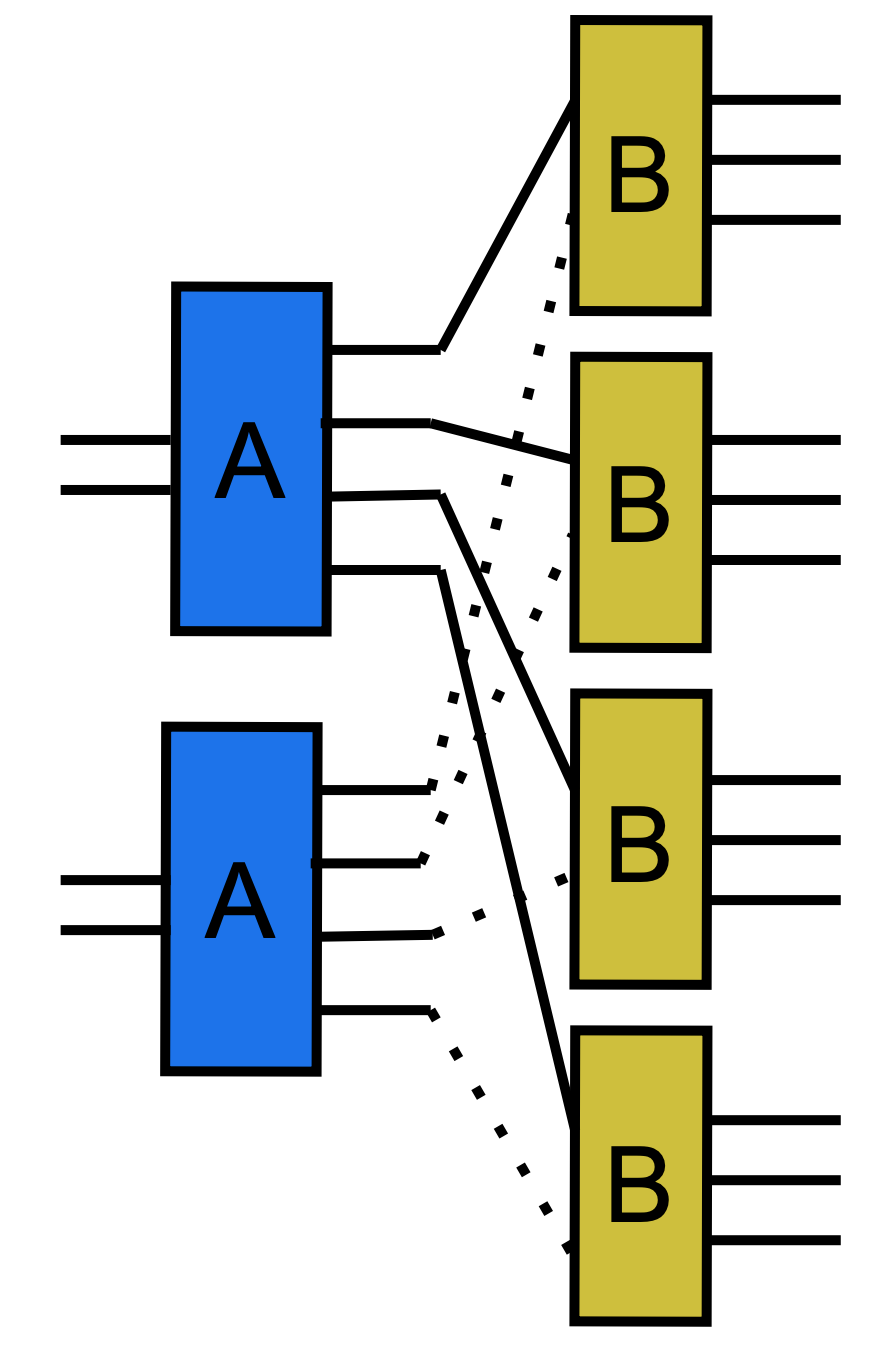}
    \caption{The Interleaved Concatenation $[[4, 2]]\otimes [[3, 2]]$}
\end{subfigure}
\begin{subfigure}[b]{0.5\textwidth}
\centering
    \includegraphics[width = 0.4\linewidth]{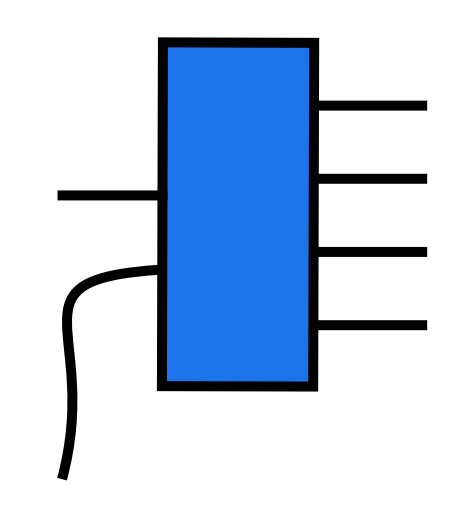}
    \caption{The $\mathsf{Reserve}_1$ Operation}
\end{subfigure}
\caption{Two Basic Operations on Stabilizer Codes}
\label{fig:basicops}
\end{figure}

Of particular interest to us will be the operation $(2\otimes B)$, which, in an abuse of notation, represents placing two copies of $B$ side-by-side, and interleaving their logical (input) qubits. That is, the $2\cdot k_B$ logicals of $(2\otimes B)$ are numbered, and the odd ones are encoded into the first copy of $B$, and the even ones into the second copy of $B$. As we discuss, performing this interleaving will be crucial to ensure robustness against errors acting on adjacent copies of the concatenated code (See \cref{fig:tower_of_hamming}).

\subsection{Quantum Hamming Codes}

Quantum Hamming codes \cite{steane1996} are high-rate CSS codes. 

\begin{definition}
    The quantum Hamming code $H_m$ is a $[[2^m -1, 2^m - 2m - 1, 3]]$ CSS code. 
\end{definition}

Their $k$’th $X (Z)$ stabilizer includes the term $X_i (Z_i)$ if and only if the $k$’th binary bit of the integer $i$ is 1. As a result, listing whether or not each $X (Z)$ stabilizer is flipped by a single $Z (X)$ data error produces the binary representation of the position of the error. 

\suppress{

As Hamming codes get larger, their coding rate converges to $100\%$:

\begin{equation}
    \text{Rate}(H_m) = \frac{2^m-2m-1}{2^m-1}, \quad \lim_{m\rightarrow \infty} \text{Rate}(H_m) = 1
\end{equation}

Concatenating a Hamming code with itself again and again would create codes with coding
rates closer and closer to 0 as the amount of concatenation increased. However, if you concatenate Hamming codes without using the same Hamming code twice, then the coding rate of the concatenation is bounded away from 0. For example, $H_m \otimes H_{m-1}\otimes \cdots H_5\otimes H_4$ converges to a coding rate of roughly 20$\%$ as m increases:

\begin{equation}
    \lim_{m\rightarrow \infty} \text{Rate}\bigg(\otimes_{i=4}^m H_i\bigg) = \lim_{m\rightarrow \infty}\prod_{i=4}^m\frac{2^i-2i-1}{2^i-1}\approx 19.7\%.
\end{equation}

}

\section{The Tower of Quantum Hamming Codes}
\label{section:interleaved_memory}

A constant rate quantum memory can be built by concatenating larger and larger Hamming
codes \cite{Yamasaki_2024}. However, in order to build such a fault tolerant circuit in 1D, instead of a non-local code, additional overheads are needed. In this paper, we make the following additions to the tower of interleaved-concatenated Hamming codes (\cref{def:interleaved_concatenation}), to enable implementing it with a 1D local circuit.

\begin{enumerate}
    \item At the physical level (the bottom), each data qubit will be accompanied by two helper qubits (a “measurement qubit” and an “entangling qubit”). These helper qubits will be used to create, verify, and consume cat states.\footnote{We use the notation $[[3, 1, 1]]$ to indicate a data qubit is placed together with $2$ ancilla qubits.}

\item At each level of concatenation, the concatenated code will further be interleaved with a copy of itself. This will later prevent operations on logical qubits from \textit{adjacent} underlying codes (in 1D) from simultaneously breaking two underlying data qubits.

\item Additionally, a logical qubit will be reserved at each level of concatenation, for
storing cat states.
\end{enumerate}

In \cref{section:fault-tolerance}, we make the role of each of these additions precise. Formally, the stabilizer code family $(C_0, C_1, C_2,\cdots )$ that we use to store logical information is defined as follows:
\begin{gather}
\label{equation:code_definition}
    C_0 = H_4\otimes [[3, 1, 1]]\\
    C_{m+1} = H_{m+5} \otimes \bigg(2\otimes \mathsf{Reserve}_1(C_m)\bigg)
\end{gather}

\noindent We refer the reader to \cref{fig:tower_of_hamming} for a diagram of the recursive definition, and back to \cref{fig:basicops} for definitions of the basic operations. To conclude this section, we present simple calculations of the basic static properties of the code; namely block-length, rate, and distance. However, in a first pass we recommend the reader to skim these statements and proceed to \cref{section:fault-tolerance}
 on the implementation of gadgets on the code.

\begin{figure}[h]
    \begin{subfigure}[b]{0.5\textwidth}
\centering
    \includegraphics[width = 0.5\linewidth]{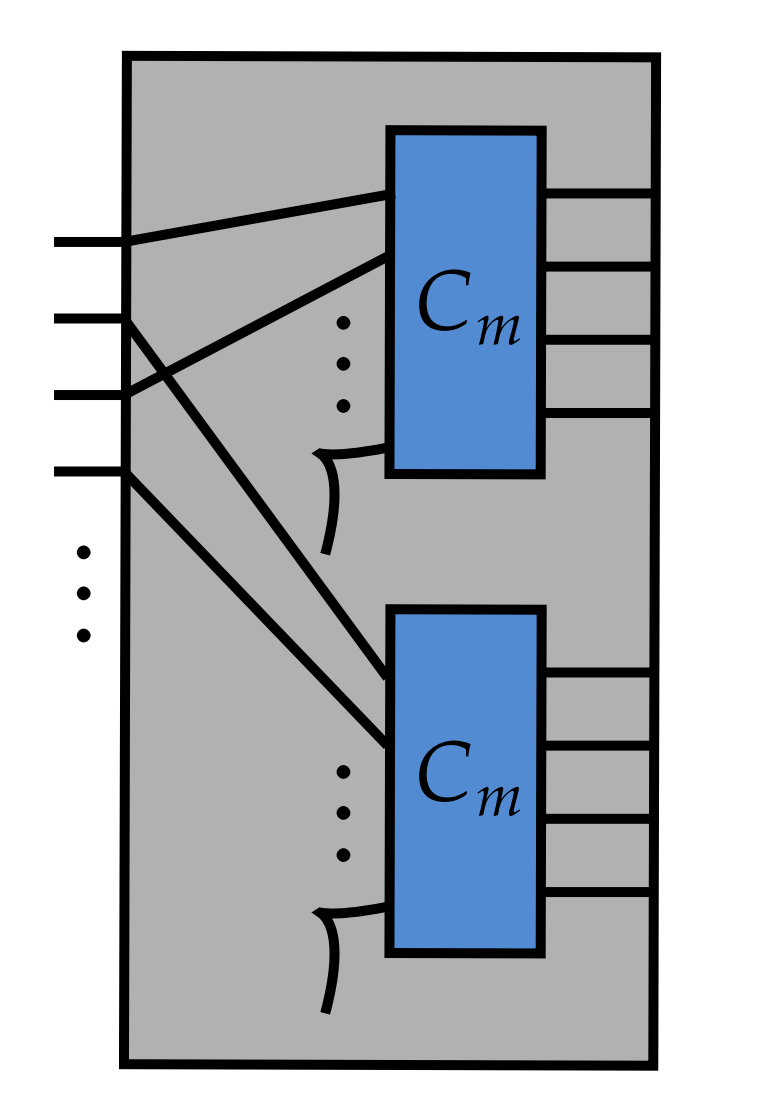}
    \caption{$2\otimes \mathsf{Reserve}_1(C_m)$}
\end{subfigure}
\begin{subfigure}[b]{0.5\textwidth}
\centering
    \includegraphics[width = 0.65\linewidth]{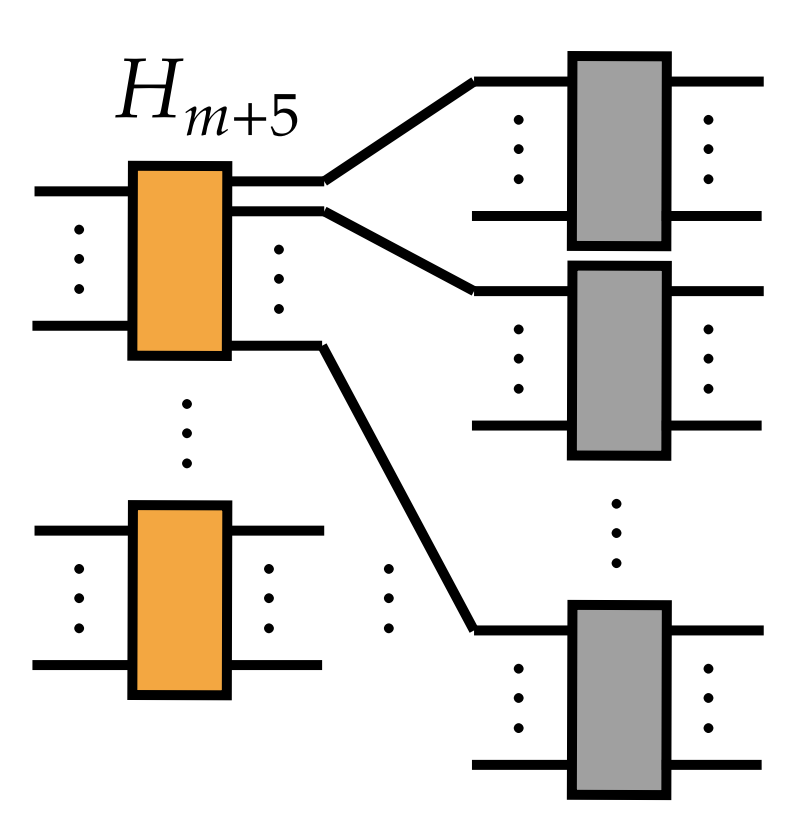}
    \caption{The recursive definition of $C_{m+1}$.}
\end{subfigure}
\caption{The modified tower of interleaved concatenated Hamming codes (\cref{equation:code_definition}). }
\label{fig:tower_of_hamming}
\end{figure}

\subsection{Static Properties of the Code Construction}

Here we analyze the block-length $n_m$, rate $r_m$, distance $d_m$, and number of stabilizers of the family of concatenated, interleaved Hamming codes defined above.

\begin{theorem}\label{theorem:code_parameters}
    The interleaved concatenated Hamming code $C_m$ of \cref{equation:code_definition} is a $[[n = 2^{m^2/2 + O(m)}, >  n/20,  3^m = 2^{\Theta(\sqrt{\log n})}]]$ CSS code.
\end{theorem}

We divide the proof into two lemmas.

\begin{lemma}
    The interleaved concatenated Hamming codes of \cref{equation:code_definition} is a family of $[[n_m, r_m\cdot n_m]]$ of stabilizer codes of blocklength $n_m \approx 42\cdot 2^{(m^2+11m)/2}$ and rate $\lim_{m\rightarrow \infty} r_m  > 1/20$.
\end{lemma}

\begin{proof}
    The blocklength $n_m$ of the concatenated code satisfies the recursion
    \begin{equation}
        n_{m+1}  = 2\cdot n_m \times \bigg(\text{Size of Hamming Code } H_{m+5}\bigg) = 2 \big(2^{m+5}-1\big)\cdot n_m.
    \end{equation}
    Which, under the appropriate base case, solves to $n_m \approx 42\cdot 2^{(m^2+11m)/2}$. In turn, the number of logical qubits $k_m$ and the rate $r_m$ of the concatenated code satisfy the recursions
    \begin{gather}
        k_{m+1} = (2^{m+5}-2(m+5)-1)\cdot 2\cdot (k_m-1) \Rightarrow r_{m+1} \approx r_m\cdot \frac{2^{m+5}-2(m+5)-1}{2^{m+5}-1} \\ \Rightarrow \lim_{m\rightarrow \infty} r_m \approx 6\%.
    \end{gather}\end{proof}

Next, we quantify the code distance of the code. By ``code distance" we simply mean the smallest weight of the Pauli error which would damage the data qubits (but, regardless of the reserved qubits); it could be less than the fault-distance of the circuit which implements the code. Nevertheless, we compute it for completeness.

\begin{lemma}
    The distance of $C_m$ is $3^m$, and it is defined on $O(n_m\cdot m\cdot 2^{-m})$ outer stabilizers.
\end{lemma}

\begin{proof}
    The operations $\mathsf{Reserve}_1$ and the interleaving $2\otimes$ do not modify the distance $d_m$ of the code. Therefore, we obtain the recursion
    \begin{equation}
        d_{m+1} = d_m\cdot 3\Rightarrow d_m = 3^m
    \end{equation}

    We remark that the number of outer stabilizers, i.e. the number of total Hamming code stabilizers at level $m$, is simply the number of copies of the Hamming code ($k_{m-1}$) times the number of stabilizers of a single copy of the Hamming code $2(m+5)$. Thus, upper bounded by
    \begin{equation}
        2(m+5)\cdot k_{m-1} \approx 12\% \cdot (m+5)\cdot n_{m-1} = n_m\cdot \frac{.06 (m+5)}{2^{m+5}-1}
    \end{equation}
\end{proof}

\section{A Fault-tolerant Quantum Memory in 1D}
\label{section:fault-tolerance}

In this section, we describe how to implement the error-correction gadget $r$-$\ec$ and the fault-tolerant measurement $r$-$\meas$. Integral will be the definition of a non-fault-tolerant ``Hookless" measurement $r$-$\ho$, described below. In the subsequent sections, we discuss their fault-tolerance.

\subsection{Overview}

When measuring an operator, there are generally two classes of errors to worry about: \textit{Wrong-result} errors (the result of the measurement is reported incorrectly) and \textit{data-damage} errors (flipping the qubits touched by the measurement circuit). A measurement process can easily introduce and spread errors in a way that reduces the fault-tolerance of the circuit to below the distance of the code implemented by the circuit. Especially when the connectivity of the circuit is restricted. Error mechanisms that reduce the data fault distance in this way are known as “hook errors” \cite{fowler2012bridge}.

In this manner, key in the construction will be to implement a measurement process without bad hook errors, referred to as “hookless measurements” at level $r$, or $r$-$\ho$. $r$-$\ho$ will be non-fault-tolerant (i.e. won't successfully return the measurement outcome).
However, it will have high data-fault distance.

For simplicity, and ease of explanation, we begin by assuming we have a black-box measurement functionality $r$-$\ho$. In \cref{section:rec_definition}, we show how to build the error correction gadget $r$-$\ec$ from $r$-$\ho$; in \cref{section:rmeas_definition}, we show how to build the fault-tolerant measurement gadgets $r$-$\meas$ from $r$-$\ho$ and $r$-$\ec$. Then, in \cref{section:hookless_definition}, we describe how to implement $r$-$\ho$ recursively from $(r-1)$-$\ec$, $(r-1)$-$\meas$.

\subsection{$r$-$\ec$, the Error Correction Gadget}
\label{section:rec_definition}

$r$-$\ec$, the error correction gadget at level $r$, consists simply of repeat Hookless measurements of the outer stabilizers of $C_r$. That is, we measure all the $s_r$ stabilizers of the underlying Hamming codes $H_{r+5}$ within $C_r$, and repeated said measurements 3 times (see \cref{fig:rec}).

\begin{figure}[t]
    \centering
    \includegraphics[width=0.6\linewidth]{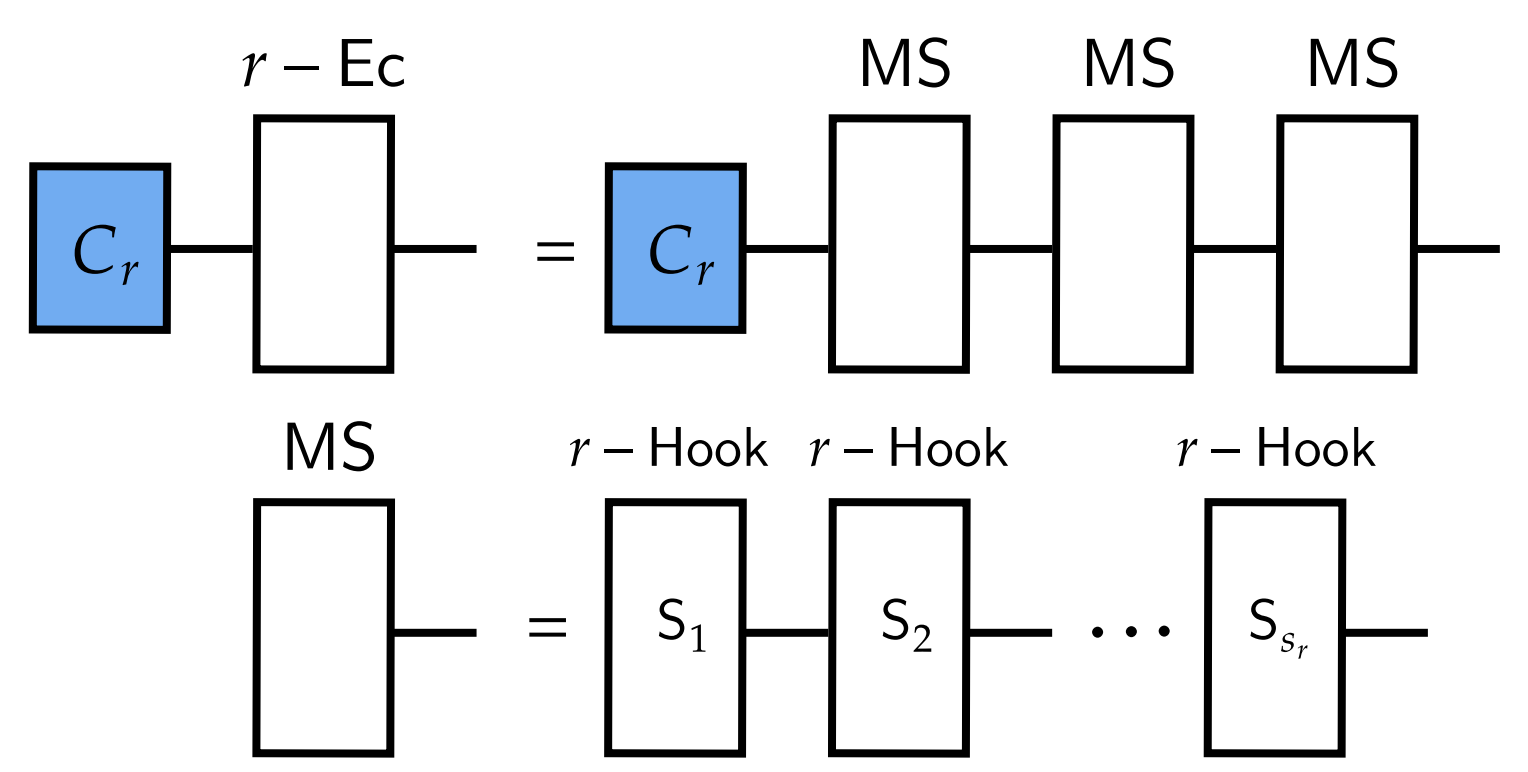}
    \caption{The error correction gadget $r$-$\ec$, consists of three repetitions of Hookless measurements $r$-$\ho$ of all outer stabilizers.}
    \label{fig:rec}
\end{figure}

We claim that 3 repetitions are enough to recover any effective single-qubit errors, or single-faults, occurring on the input $r$-block to the gadget or during the execution of $r$-$\ec$; a claim we make precise and prove only in \cref{section:propagation_proofs}. Here, we emphasize that our goal is to only correct one error/fault. Thereby, the intuition is that comparing the stabilizers between each round of Hookless measurements, acts as a repetition code (of distance 3), and thereby serves to identify the region wherein the faulty measurement must lie (again, assuming there is only one).

\subsection{$r$-$\meas$, the Fault-tolerant Measurement Gadget}
\label{section:rmeas_definition}

At level $r$, suppose we are given some logical observable $O$ supported on the underlying Hamming codes ($H_{r+5}$) at that level (possibly on two $C_r$ blocks). The fault-tolerant measurement gadget $r$-$\meas$ for $O$ consists of alternating 3 rounds of Hookless measurements $r$-$\ho$ of $O$ with error-correction rounds $r$-$\ec$ (see \cref{fig:rmeas}).

\begin{figure}[b]
    \centering
    \includegraphics[width=0.8\linewidth]{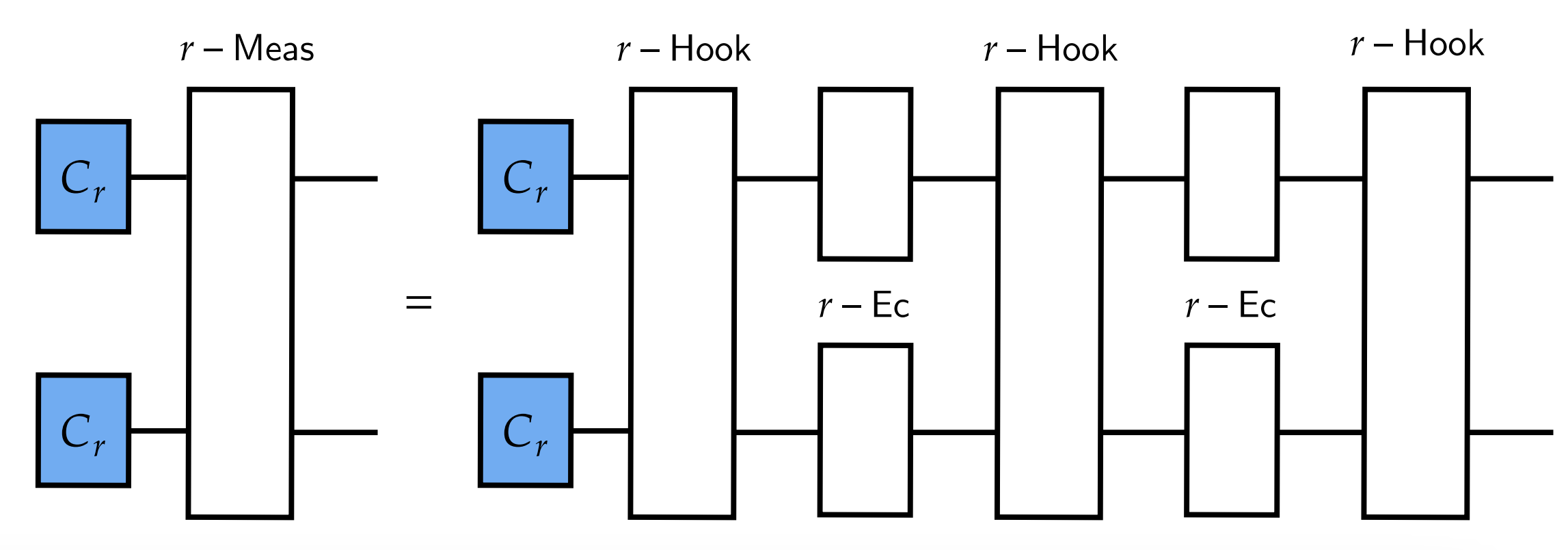}
    \caption{The fault-tolerant measurement gadget $r$-$\meas$. }
    \label{fig:rmeas}
\end{figure}

Broadly speaking, the intuition behind $r$-$\meas$ is that the alternation with rounds of error-correction $r$-$\ec$ serves to isolate faults between the rounds of Hookless measurements $r$-$\ho$. Indeed, rechecking the stabilizers after each hookless measurement prevents a fault distance 2 error mechanism where the same data error occurring before and after a series of measurements simultaneously flips all the measurement results in between. Again, this is made precise in \cref{section:propagation_proofs}.

\subsection{$r$-$\ho$, the Hookless Measurement}
\label{section:hookless_definition}

We are now in a position to describe the implementation of $r$-$\ho$. Broadly speaking, Hookless measurements are supported using variants of Shor's gadget: a means to perform multi-qubit measurements based on cat states. We present a brief recollection of this scheme in \cref{section:shors}.

As a minor technicality, we require two distinct implementations of $r$-$\ho$. One that uses unitary two qubits gates (for working with physical qubits) and another that uses dissipative two qubit gates (for working with encoded qubits). In \cref{section:hookless-recursive}, we present a recursive, dissipative implementation of $r$-$\ho$ by appealing to the measurement and error-correction gadgets at lower levels, $(r-1)$-$\meas$, $(r-1)$-$\ec$. In \cref{section:hook_base_case}, we describe how to implement the base case $0$-$\ho$. For this purpose, we require a unitary implementation, see below.

\subsubsection{A Recap of Shor's Gadget}
\label{section:shors}

Suppose one would like to measure an operator $M = P_{q_1}\otimes\cdots \otimes P_{q_t}$ defined on qubits $(q_1, \cdots q_t)$. To construct a hookless measurement of this operator, we use a variation of Shor’s measurement gadget \cite{shor1996faulttolerance}.

\begin{algorithm}[h]
    \setstretch{1.35}
    \caption{Shor's Measurement Gadget}
    \label{alg:shors-gadget}
    \KwInput{An Observable $M = P_1\otimes P_2\cdots P_t$ on registers $(q_1, \cdots, q_t)$.}

    \KwOutput{A bit $v$ corresponding to the measurement outcome.}

    \begin{algorithmic}[1]

    \State Prepare a (hardened) cat state $\ket{0^t}+\ket{1^t}$ on registers $(c_1, \cdots, c_t)$.

    \State Pair the cat state qubits with the measurement qubits $(c_k, q_k)_{k\in [t]}$ and for each $k\in [t]$,\linebreak measure $P_{q_k}\otimes X_{c_k}$. If $v$ is the measurement outcome, the resulting state is
    \begin{equation}
        \ket{0^t}+(-1)^v\ket{1^t}.
    \end{equation}

    \State Measure all the qubits of the cat state in the $Z$ basis, except for qubits $q_k$ where $P_{q_k} = I$.
    (These qubits were measured in the X basis by the previous step.)

    \end{algorithmic}

\end{algorithm}

Shor's gadget is based on preparing a cat state.
A cat state is naturally not fault-tolerant: a single fault in the preparation could spoil the measurement.
Even worse, some faults during the creation of the state can fail to synchronize its ends, resulting in states like $|000\dots111\rangle + |111\dots000\rangle$.
These states effectively have huge ranges of X errors, which propagate into the data qubits to produce high weight hook errors.
Shor \textit{hardens} the cat state against these hook errors by checking that random pairs of qubits from the cat state are in the +1 eigenstate of the $ZZ$ operator.
Unfortunately, comparing random pairs of qubits isn’t ideal when restricted to 1d-local connectivity.
We instead prepare hardened cat states by running a repetition code, i.e. by performing 3 repetitions of ``nearest-neighbor" logical $Z_{c_k}\otimes Z_{c_{k+1}}$ measurements (see \cref{fig:hook_dissipative} and \cref{fig:hook_unitary}).

\subsubsection{A recursive implementation of Hookless measurements}
\label{section:hookless-recursive}

Let us now understand how to implement each Hookless measurement in an (or within two) $C_r$ block(s), using operations on their $C_{r-1}$ sub-blocks. We refer the reader to \cref{fig:hook_recursive} for a diagram, explained below:

At the logical level, cat states are stored using the encoded reserved qubits in \cref{equation:code_definition}. We begin by preparing cat states encoded within their $C_{r-1}$ sub-blocks. This is performed via 3 repetitions of logical $Z\otimes Z$ measurements on the reserve logical qubits of nearest-neighbor $C_{r-1}$ codes, implemented recursively using the fault-tolerant measurement $(r-1)$-$\meas$ of \cref{section:rmeas_definition}. Following the concatenated simulations framework, these calls to $(r-1)$-$\meas$ alternate with error-correction rounds $(r-1)$-$\ec$. Finally, after the cat state is prepared, parity $(r-1)$-$\meas$ are performed within each $C_{r-1}$ block, acting between the cat qubit within that block, and the Hamming code qubits.

\begin{figure}[t]
    \centering
    \includegraphics[width=1.0\linewidth]{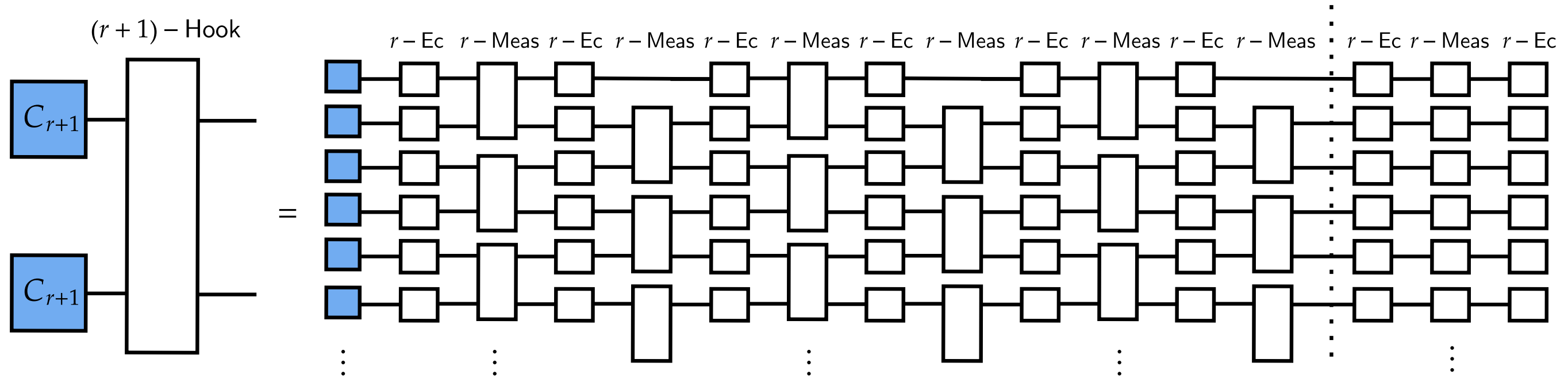}
    \caption{A recursive implementation of the Hookless measurement, $r$-$\ho$. The dashed line to the RHS indicates the end of ZZ parity measurements/cat state preparation.}
    \label{fig:hook_recursive}
\end{figure}

Modelling noise, and the propagation of faults, in this circuit is one of the key technical challenges of this work.  We include \cref{fig:hook_dissipative} to illustrate the cat state preparation: the reserved qubits and the data qubits are drawn as separate lines; they look independent. However, each reserved qubit is indeed \textit{within} the same code as the data qubit above it in the diagram. Thus, a fault during a parity measurement between two reserved (logical) qubits, could in principle destroy not only those two (logical) qubits, but also simultaneously destroy the other logical qubits in the code block.

This is precisely why the designed definition of $C_m$ (\cref{equation:code_definition}) includes the $2\otimes$ term, that interleaves each code with a copy of itself. Data qubits from adjacent underlying codes can’t be part of the same overlying code, because a single parity measurement between those two codes can simultaneously break both data qubits; which would reduce the fault distance of the circuit below the code distance. 

\begin{figure}[h]
    \centering
    \includegraphics[width=0.6\linewidth]{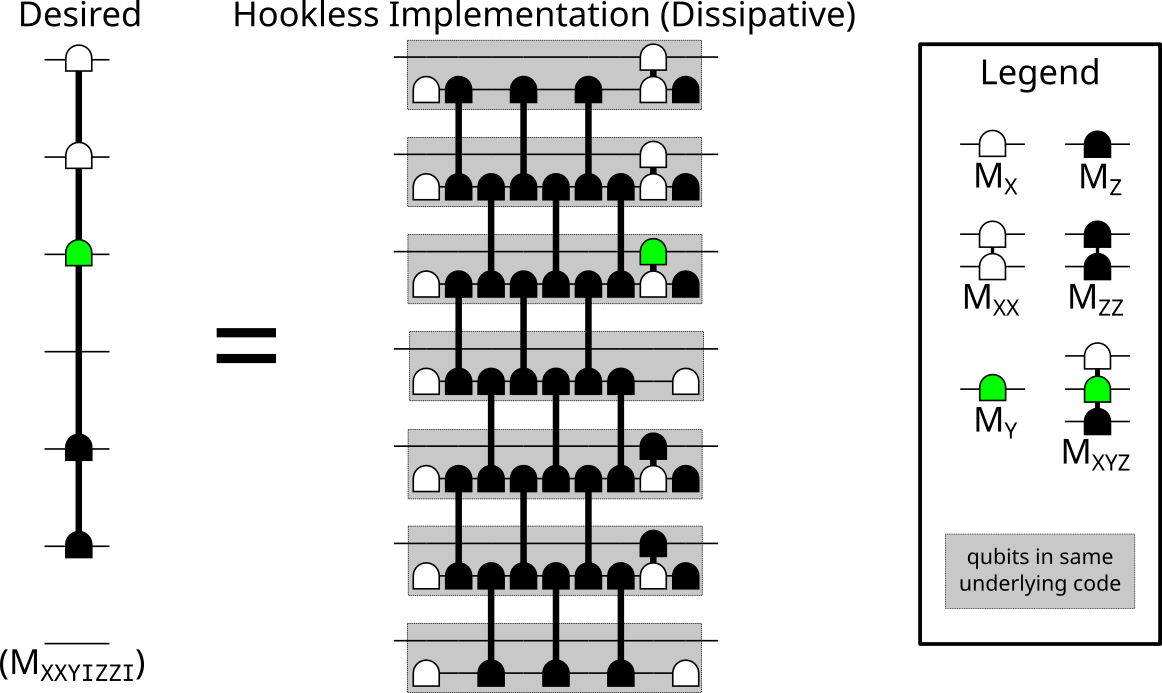}
    \caption{The recursive implementation of the Hookless measurement, $r$-$\ho$. The circuit is shortened and all $r$-$\ec$ blocks are omitted for illustrative purposes. Each gray block represents an instance of $C_r$. Within it, there are data qubits and a single cat state qubit.}
    \label{fig:hook_dissipative}
\end{figure}

\subsubsection{The base case: Hookless Measurements at the Physical Level}
\label{section:hook_base_case}

It remains to discuss the base case. At the physical level, there are three roles that qubits play: data, measurement, and entangling. Cat states are stored using the measurement qubits and created with assistance from the entangling qubits. The main challenge to implement the (hardened) cat state creation in 1D is that it shouldn’t involve the data qubits, but it has to cross over data qubits. If next-nearest-neighbor connectivity was allowed, the data qubits could simply be bypassed. 

We recommend the reader content with next-nearest-neighbor connectivity to skip ahead; to achieve nearest-neighbor connectivity, we require a painstaking construction and analysis of the base case circuits. To implement the next-nearest-neighbor connectivity, $\cnot$ gates crossing over data qubits are decomposed into four $\cnot$ gates touching the data qubit, see \cref{fig:hook_unitary}. This decomposition isn’t just expensive, it also means there are error mechanisms that can simultaneously damage cat state qubits and data qubits.  
Nevertheless, in \cref{section:base_case} we show that this circuit still has a ``data"-fault distance of $3$.

\begin{figure}[h]
    \centering
    \includegraphics[width=0.6\linewidth]{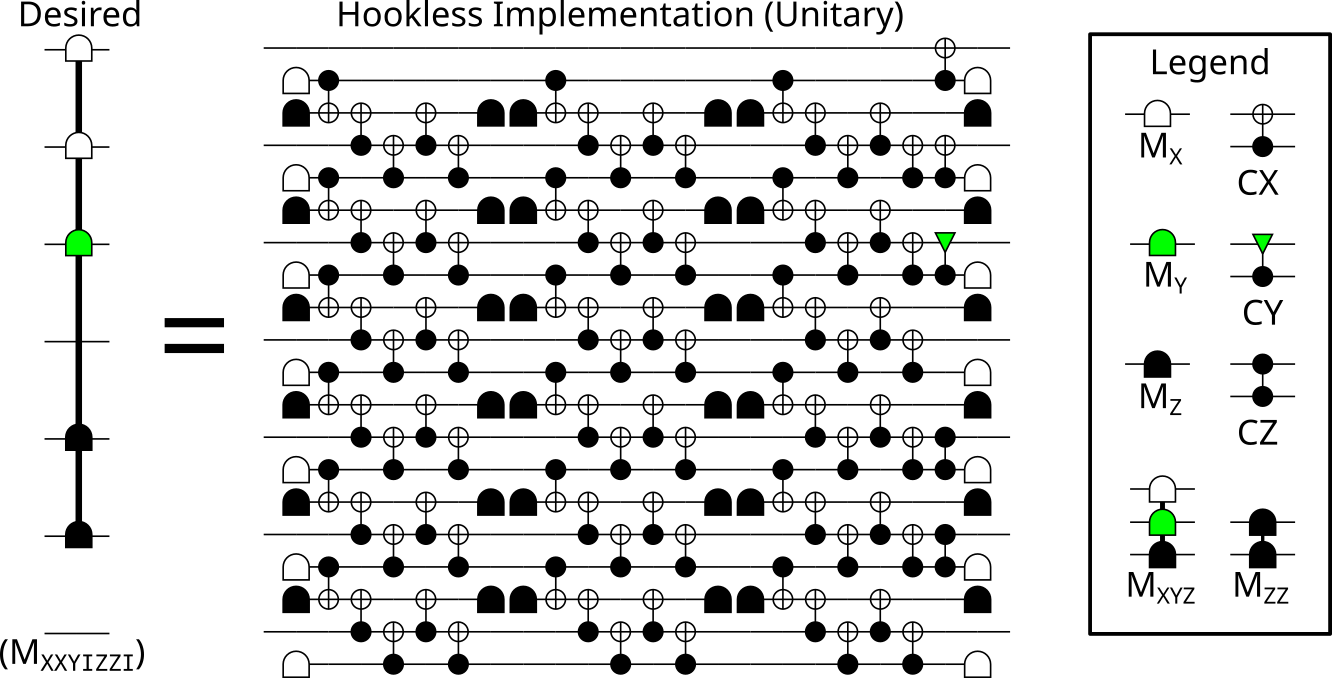}
    \caption{The unitary implementation of $0$-$\ho$. Three rounds of the repetition code circuit are performed, with next-nearest-neighbor $\cnot$ gates decomposed into $4$ adjacent $\cnot$ gates.}
    \label{fig:hook_unitary}
\end{figure}

\subsection{Time Overhead}

Let $T_r$ be the number of physical circuit layers (i.e. the circuit depth) it takes to perform a single fault tolerant measurement within the code $C_r$. 

\begin{lemma}\label{lemma:time-overhead}
    The time overhead of $C_r$ scales asymptotically as $T_r =  2^{O(r^3)} $.
\end{lemma}

\begin{proof}
    Recall the definition of interleaved concatenation in \cref{def:interleaved_concatenation} and \cref{fig:basicops}. Performing a fault-tolerant measurement of some logical observable for $C_r$, entails measuring all the outer-stabilizers of the copies of $H_{r+5}$, interpreted as logical observables of the ``inner" copies of $C_{r-1}$ (three times). Since the many outer-stabilizers may share overlapping support, they must be performed sequentially. Between each outer-stabilizer measurement, we run a level $r-1$ error-correction routine. The circuit depth $T_r$ to implement these measurements then satisfies the recursion 
    \begin{equation}
        T_r = O(n_r\cdot T_{r-1}) = \exp\bigg[O\bigg(\sum^r_i i^2\bigg)\bigg]  = 2^{O(r^3)},
    \end{equation}

    \noindent where $n_r = 2^{\Theta(r^2)}$ is the block-length of $C_r$.
\end{proof}

We remark that if $b = 2^{\Theta(r^2)}$ is the block-length of $C_r$, then the time-overhead of $T_r$ is $2^{O(\log^{3/2} b)}$, super-polynomial in the block-length! The sequentiality to the stabilizer measurements and its effect to the runtime is another key distinction of our results to \cite{Yamasaki_2024}. Fortunately, we discuss how to decrease the time-overhead in \cref{section:proof_of_main}.

\section{The Threshold Dance}
\label{section:threshold}

In this section, we study the fault-tolerance of the measurement and error-correction gadgets of the quantum memory, $r$-$\meas$ and $r$-$\ec$. We begin in \cref{section:rerrors}, by presenting a recursive definitions of what it means for a code-block to contain an error. Subsequently, in \cref{section:recs_correctness}, we introduce the notion of an $r$-$\rec$(-tangle), a key concept to understand the correctness of the recursive simulation in the presence of faults. In \cref{section:recs}, we discuss the minimal (or sufficient) properties on the $r$-$\meas$ and $r$-$\ec$ gadgets, to ensure correctness. Informally, these properties quantify the structure of how faults create errors on the code-block, and constrain their propagation throughout the circuit.

Finally, in \cref{section:percolation}, we prove a threshold theorem for our construction, under the assumption that it satisfies the desired error-propagation properties. We defer a proof of these properties to the next section, \cref{section:propagation_proofs}.

\subsection{$r$-Errors and Decodability}
\label{section:rerrors}

We refer to an instance of the code $C_r$ within the quantum memory as an $r$-block. For $r\geq 1$, $r$-blocks are comprised of multiple instances of $(r-1)$-blocks, laying side-by-side (see \cref{fig:rblockserrors}, a). A $0$-block is a Hamming codes $H_4$, whose physical qubits are interleaved with 2 ancilla qubits, see \cref{equation:code_definition}. We are now in a position to define an $r$-error on an $r$-block.

\begin{definition} [$r$-errors]
    For $r\geq 1$, an \emph{$r\error$} on an $r$-block corresponds to a $(r-1)$-block with at least $2$ non-adjacent $(r-1)\error$s. Two $r\error$s are said to be \emph{adjacent} if they lie on nearest neighbor $(r-1)$-blocks; in which case we refer to the pair as an $r\adj$. A $0\error$ corresponds to a Pauli operator on a single qubit of $C_0$; the notion of adjacency is the same.
\end{definition}

The crux of this definition is that adjacent $r$-errors roughly correspond to \textit{single-qubit} errors on the underlying Hamming codes, due to the alternating/interleaving odd/even structure of the concatenation. In this next lemma, we inductively show these patterns of errors are decodable. 

\begin{figure}[h]
    \begin{subfigure}[b]{0.5\textwidth}
\centering
    \includegraphics[width = 0.7\linewidth]{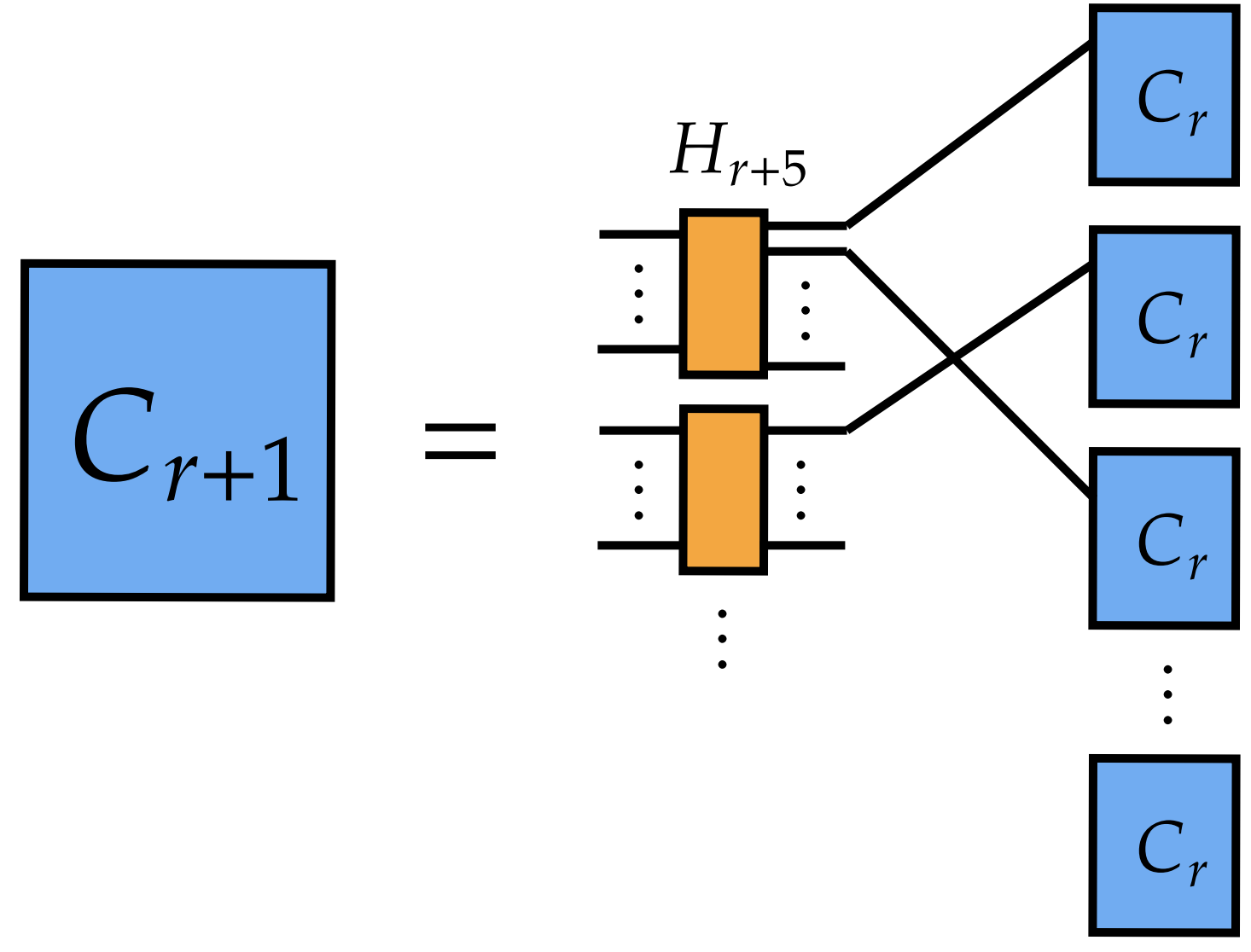}
    \caption{An $(r+1)$-block.}
\end{subfigure}
\begin{subfigure}[b]{0.5\textwidth}
\centering
    \includegraphics[width = 0.7\linewidth]{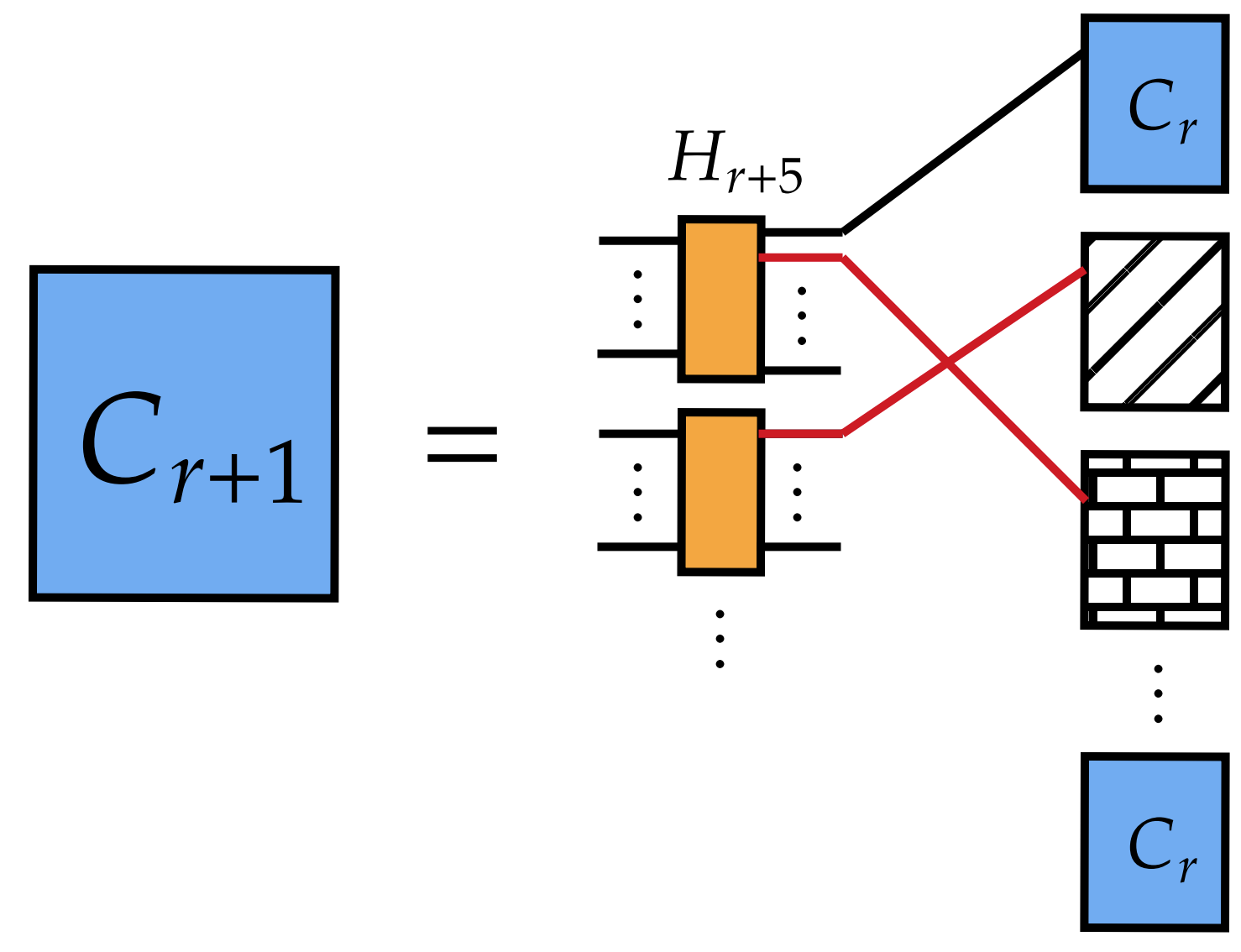}
    \caption{An $(r+1)$-block with adjacent $(r+1)$-errors.}
\end{subfigure}
\caption{In (a), we depict an $(r+1)$-block as a sequence of side-by-side $r$-blocks. Reserved qubits are omitted. In (b), we depict an adjacent pair of $(r+1)\error$s. The dashed/brickwork patterns indicate the errors lie on different underlying Hamming codes. }
\label{fig:rblockserrors}
\end{figure}

\begin{lemma}
    [Adjacent $r\error$s are decodable]\label{lemma:1error->decodable} Suppose an $r$-block $C_r$ contains at most one $r\adj$. Then, there exists a decoder $r$-$\Dec$, whose noiseless execution results in $C_r$ containing no $k\error$s for any $k\leq r$.
\end{lemma}

\begin{proof}
    We present a proof by induction, and defer the base case $C_0$ to the bottom. Assume for $k\leq r$, $C_k$ is decodable from patterns of single adjacent pairs of $k\error$s as defined above. We prove the same holds for $C_{r+1}$. Our decoder for $C_{r+1}$ first runs the decoder for $C_{r}$ on each $r$-block, and then decodes the outer Hamming codes. 
    
    By design, at most one a pair of $r$-blocks $C_r$ contains an $r\adj$. By the inductive hypothesis, all other $r$-blocks can be decoded such that each such $r$-block contains no $k\error$. In turn, the $(r+1)\error$s on the two adjacent $r$-blocks, map to physical errors on the underlying Hamming codes. However, by the structure of the interleaved concatenation, there is only a single-qubit error on each Hamming code. Since the Hamming code is distance 3, it can decode $1$ physical error. 

    For the base case $C_0$, we note that the data-qubits in $C_0$ are encoded into a Hamming code with non-adjacent physical qubits.
\end{proof}

\subsection{$r$-Rectangles and Correctness}
\label{section:recs_correctness}

Let us now begin to quantify the presence, and propagation, of faults in the circuit. For this purpose, we follow \cite{aliferis2005quantum}, and introduce the notion of an $r$-$\rec$ at each level of the recursion. Informally, an $r$-$\rec$ represents a operation on an $r$-block, and its surrounding error-correction gadgets. At each level $r\geq 0$, an $r$-$\rec$ consists of an $r$-$\meas$ and its $\leq 4$ adjacent $r$-$\ec$s.

 Roughly speaking, we say that an $r$-$\rec$ is \textit{correct} if it doesn't amplify the number of errors on the input state. 

\begin{definition}
    [Correctness] An $r$-$\rec$ is \emph{correct} if on input $r$-blocks with at most one $r\adj$ each, their output is a set of $r$-blocks at most one $r\adj$ each. 
\end{definition}

Note that the location of the $r$-errors may move around inside the block (see \cref{fig:correctrecs})

\begin{figure}[h]
    \centering
    \includegraphics[width=0.8\linewidth]{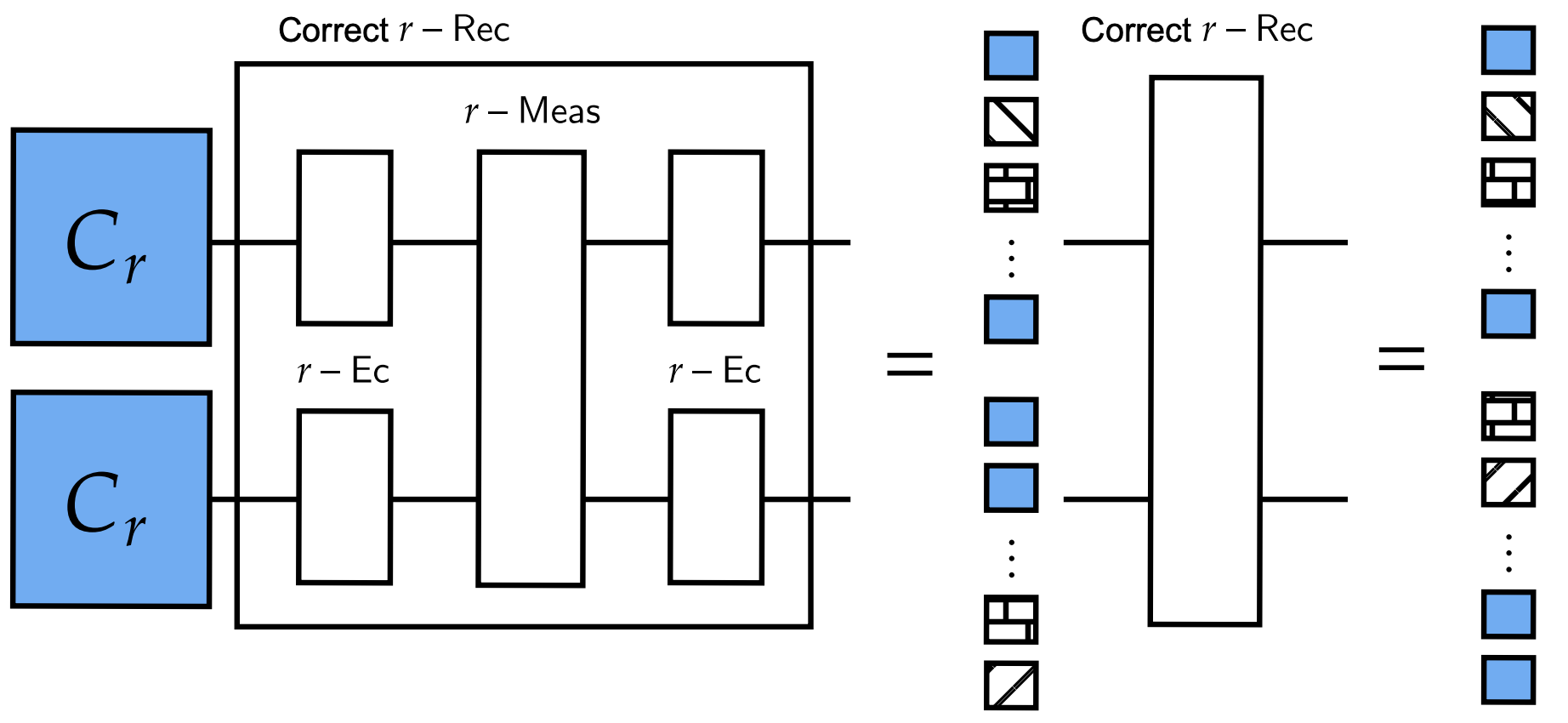}
    \caption{An $r$-$\rec$, consisting of two input $r$-blocks, a $r$-$\meas$, and surrounding $r$-$\ec$s, is pictured on the LHS. If each input $r$-block has at most $1$ $r\adj$, then the output to a \textit{correct} $r$-$\rec$ does as well.}
    \label{fig:correctrecs}
\end{figure}

We emphasize that the notion of correctness of $r$-$\rec$ is apriori independent of the possible faulty execution of its components. Shortly, we will show that an $r$-$\rec$ ``with few faults" is correct. For now, we simply prove that correct rectangles lead to decodable output states.

\begin{lemma}
    [Correct $\Rightarrow$ Decodable]\label{lemma:good-decodable} If every $r$-$\rec$ in the circuit is correct, then the logical information in the memory is preserved. 
\end{lemma}

\begin{proof}
    By assumption, the input to the quantum memory is a quantum state $\psi$ logically encoded into the concatenated stabilizer code $C_r$, with no faults on any code-blocks. If all $r$-$\rec$s in the circuit are correct, then after the first layer of $r$-$\rec$s, each $r$-block contains at most one $r\adj$. Inductively, this holds for all layers of $r$-$\rec$s, including the output layer. By \cref{lemma:1error->decodable}, the output is then decodable, and a noiseless encoding of $\psi$ is recovered. 
\end{proof}

\subsection{Good/Bad Rectangles and Sufficient Properties for Correctness}
\label{section:recs}

We are now in a position to quantify the effects of faults during the execution, and to study how they propagate. For this purpose, introduce the notion of Good/Bad rectangles \cite{aliferis2005quantum}. Informally, an $r$-$\rec$ is good if it only contains one bad $(r-1)$-$\rec$; we expect the $r$-blocks to be able to handle the presence of at most one lower level fault.  

\begin{definition}
    [Good/Bad $r$-$\rec$s] An $r$-$\rec$ is \emph{bad} if it contains at least 2 independent bad $(r-1)$-$\rec$s; if it is not bad, it is \emph{good}. Two bad $r$-$\rec$s are said to be \emph{independent} if they remain bad after removing any shared $r$-$\ec$ gadgets. A $0$-$\rec$ is \emph{bad} if it contains at least two faults.
\end{definition}

It will become relevant to quantify when rectangles independently fault and become bad; for this purpose the notion of \textit{independent} bad $r$-$\rec$s is introduced. See \cref{fig:badrecs} for a depiction.

\begin{figure}[h]
    \begin{subfigure}[b]{0.5\textwidth}
\centering
    \includegraphics[width = 0.8\linewidth]{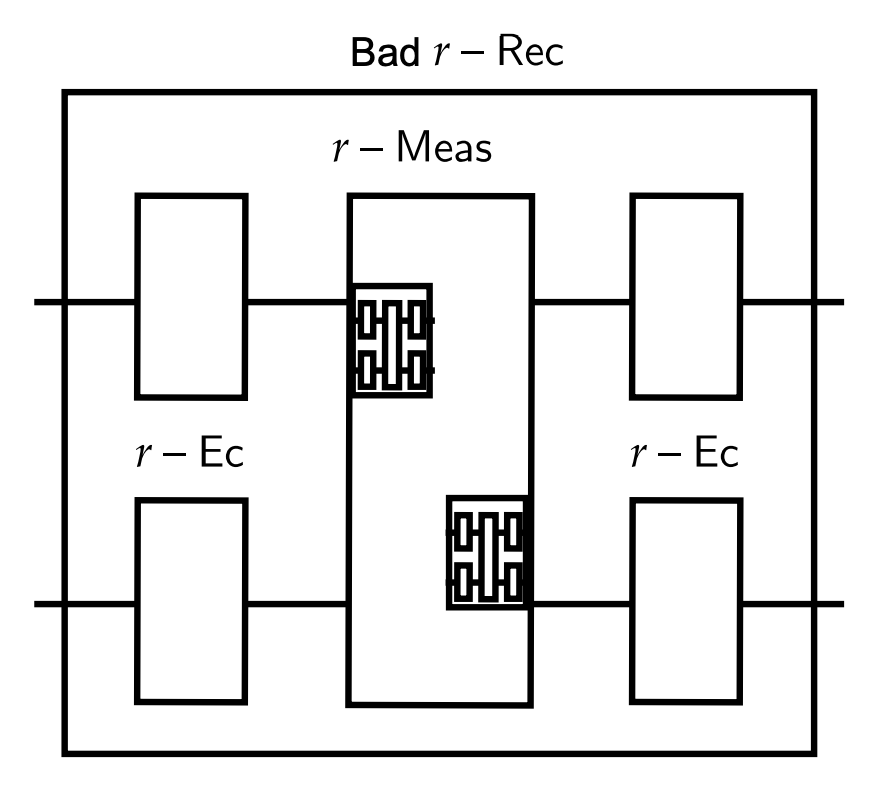}
    \caption{A bad $r$-$\rec$.}
\end{subfigure}
\begin{subfigure}[b]{0.5\textwidth}
\centering
    \includegraphics[width = .9\linewidth]{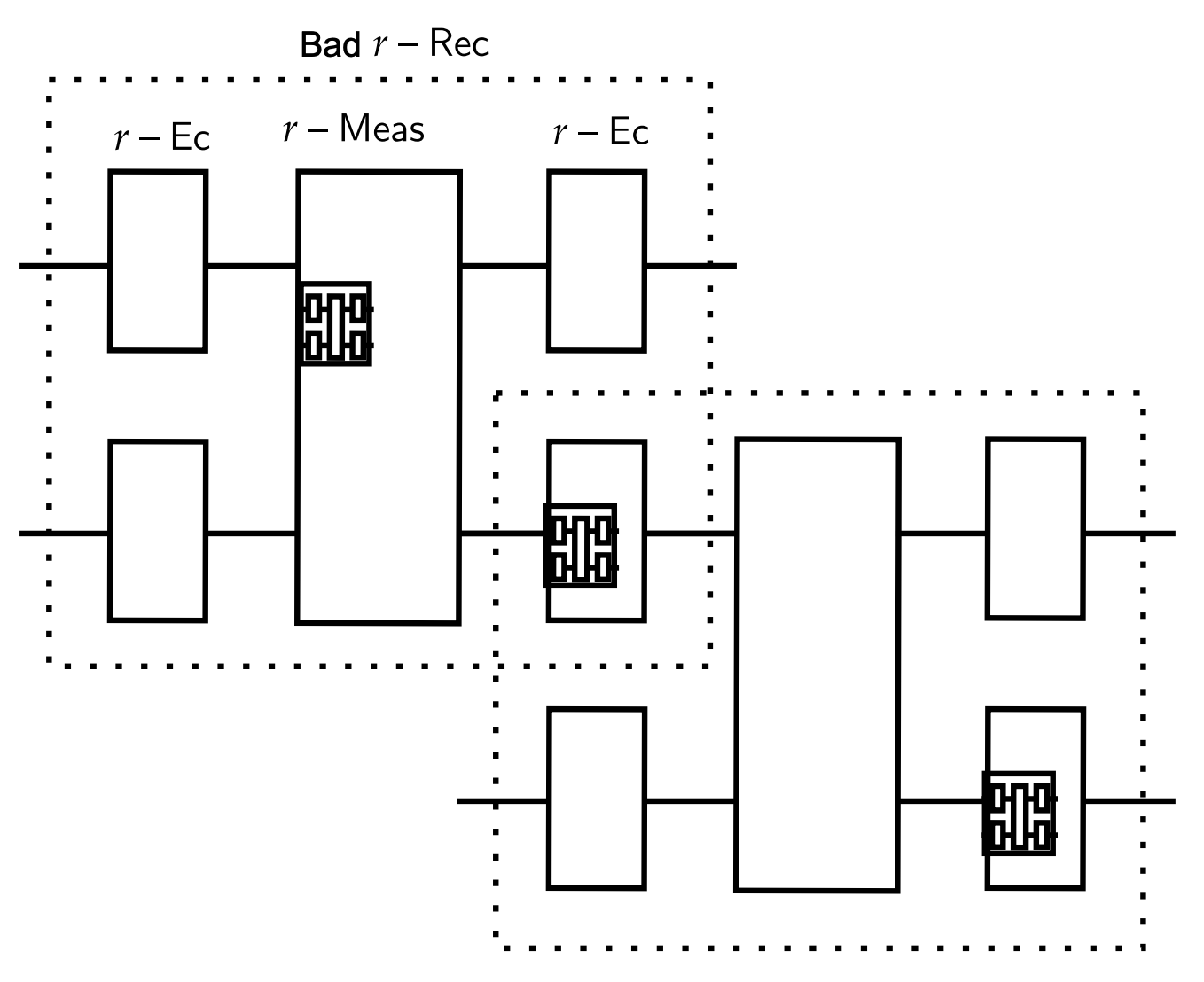}
    \caption{Two non-independent bad $r$-$\rec$s.}
\end{subfigure}
\caption{Depicted are bad $r$-$\rec$s. The smaller subsquares are bad $(r-1)$-$\rec$s; note that each bad $r$-$\rec$ has two (independent) bad $(r-1)$-$\rec$s. }
\label{fig:badrecs}
\end{figure}

We are now in a position to impose a set of basic conditions on these gadgets, which limit how they propagate and spread errors through the circuit.  

\begin{properties}\label{properties:prop1}
    For each $r\geq 1$, we assume the following properties on the Error Correction Gadget $r$-$\ec$ and the Measurement Gadget $r$-$\meas$, regarding their execution in the presence of noise. 

    \begin{enumerate}
        \item $r$-$\ec$, \emph{Faulty-Inputs}. If the input $r$-block to $r$-$\ec$ has at most one $r\adj$, and $r$-$\ec$ has no bad $(r-1)$-$\rec$, then the output $r$-block has no $r\error$s (\cref{fig:faulty-inputs}a).

        \item $r$-$\ec$, \emph{Faulty-Execution}.  If the input $r$-block to $r$-$\ec$ has no $r\error$s, and $r$-$\ec$ has no non-independent pair of bad $(r-1)$-$\rec$, then the output $r$-block has $\leq 1$ $r\adj$ (\cref{fig:faulty-execution}a).

         \item $r$-$\meas$, \emph{Faulty-Inputs}. If each input $r$-block to $r$-$\meas$ has at most one $r\adj$, and $r$-$\meas$ has no bad $(r-1)$-$\rec$, then the output $r$-blocks have $\leq 1$ $r\adj$ each (\cref{fig:faulty-inputs}b).

        \item $r$-$\meas$, \emph{Faulty-Execution}.  If each input $r$-block to $r$-$\meas$ has no $r\error$s, and $r$-$\ec$ has no non-independent pair of bad $(r-1)$-$\rec$s, then the output $r$-blocks have $\leq 1$ $r\adj$ each (\cref{fig:faulty-execution}b) . 
    \end{enumerate}

    \noindent The extension to the level $r=0$ of recursion is immediate. 
\end{properties}

\begin{figure}[h]
    \begin{subfigure}[b]{0.5\textwidth}
\centering
    \includegraphics[width = 0.7\linewidth]{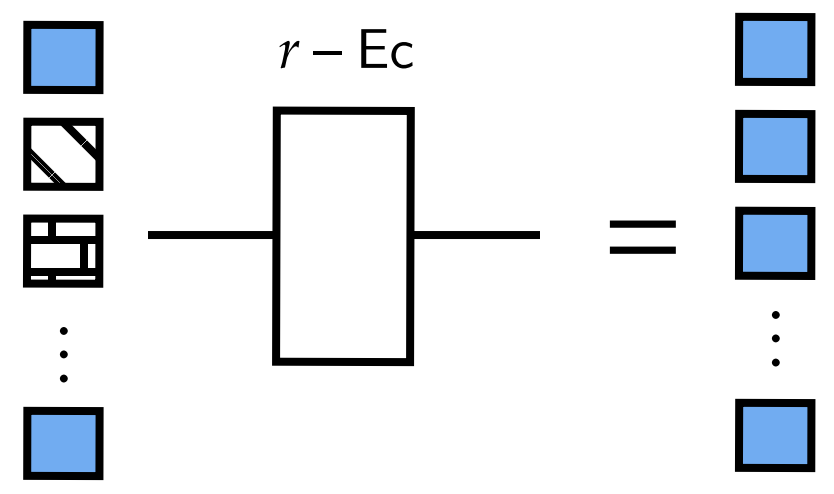}
    \caption{$r$-$\ec$, Faulty-Inputs}
\end{subfigure}
\begin{subfigure}[b]{0.5\textwidth}
\centering
    \includegraphics[width = .7\linewidth]{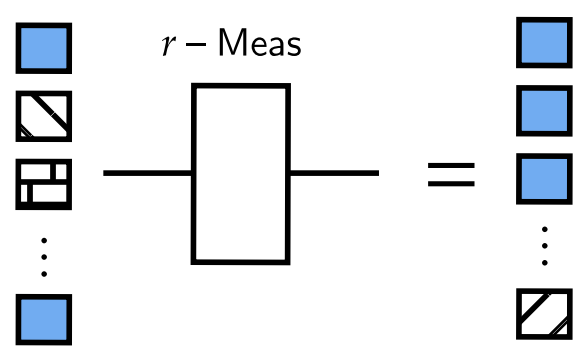}
    \caption{$r$-$\meas$, Faulty-Inputs}
\end{subfigure}
\caption{How $r$-$\ec$ and $r$-$\meas$ propagate errors under faulty inputs (but fault-free execution). }
\label{fig:faulty-inputs}
\end{figure}

\begin{figure}[h]
    \begin{subfigure}[b]{0.5\textwidth}
\centering
    \includegraphics[width = 0.7\linewidth]{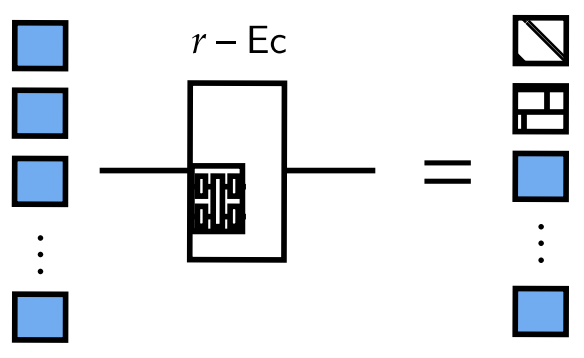}
    \caption{$r$-$\ec$, Faulty-Execution}
\end{subfigure}
\begin{subfigure}[b]{0.5\textwidth}
\centering
    \includegraphics[width = .7\linewidth]{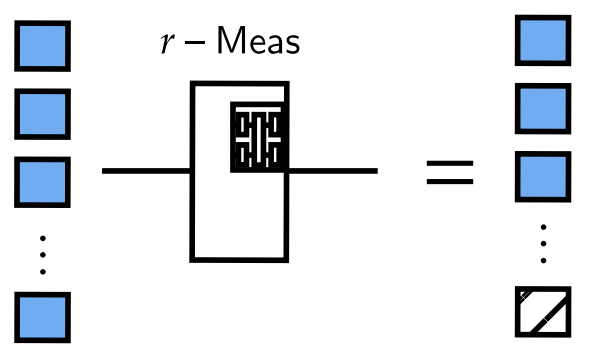}
    \caption{$r$-$\meas$, Faulty-Execution}
\end{subfigure}
\caption{How $r$-$\ec$ and $r$-$\meas$ create errors in the presence of a bad $(r-1)$-$\rec$ (but inputs which are $r$-error free). }
\label{fig:faulty-execution}
\end{figure}

We conclude this section with a crucial lemma, which shows that \cref{properties:prop1}, in combination with the assumption that all $r$-$\rec$s are good, implies all $r$-$\rec$s are correct. In turn, this implies that the logical information is preserved. In the next subsection, we present a percolation argument that proves the existence of a threshold noise rate $p^*$; below which the probability all $r$-$\rec$s are good with high probability (for sufficiently large $r$).

\begin{lemma}
    [Good $\Rightarrow$ Correct]\label{lemma:goodcorrect} Assume \cref{properties:prop1}. If every $r$-$\rec$ in the circuit is good, then every $r$-$\rec$ is also correct.
\end{lemma}

\begin{proof}
    To prove correctness of the $r$-$\rec$s, it suffices to show that every $r$-block contains at most 1 $r\adj$ during the execution of the circuit. Note, by definition, if each $r$-$
    \rec$ is good, then each one may contain at most one pair of non-independent bad $(r-1)$-$\rec$s. Each pair of non-independent $(r-1)$-$\rec$s may lie within (1) the first layer of $r$-$\ec$s, or (2) within a $r$-$\meas$, or (3) within the last layer of $r$-$\ec$s (but not between). 

    We assume the input to the first layer of $r$-$\rec$s is a collection of $r$-blocks with no $r\error$s. Let us fix our attention to a specific $r$-$\rec$ in that first layer, and divide into cases on the location of its bad $(r-1)$-$\rec$s. If 

    \begin{enumerate}
        \item The bad $(r-1)$-$\rec$s lie within the first layer of $r$-$\ec$s (see \cref{fig:goodcorrect12}, a). By ($r$-$\ec$, \emph{Faulty-Execution}), after the $r$-$\ec$ the output $r$-block has $\leq 1$ $r\adj$. Consequently, by ($r$-$\meas$, \emph{Faulty-Inputs}), after $r$-$\meas$, the output $r$-blocks still have $\leq 1$ $r\adj$. Since the $r$-$\rec$ is Good, the last layer of $r$-$\ec$s cannot have any bad $(r-1)$-$\rec$. Thus, by ($r$-$\ec$, \emph{Faulty-Inputs}), the output $r$-blocks to said $r$-$\rec$ have no $r\error$s, and therefore is Correct. 

        \item The bad $(r-1)$-$\rec$s lie within the $r$-$\meas$ (see \cref{fig:goodcorrect12}, b). If the input to the $r$-$\rec$ contains $\leq 1$ $r\adj$, then by ($r$-$\ec$, \emph{Faulty-Inputs}), after the first layer of $r$-$\ec$ gadgets, each $r$-block contains no $r$-errors. By ($r$-$\meas$, \emph{Faulty-Execution}), after the faulty $r$-$\meas$, each output block contains $\leq 1$ $r\adj$. Since the $r$-$\rec$ is Good, the last layer of $r$-$\ec$ gadgets cannot have any bad $(r-1)$-$\rec$. Thus, by ($r$-$\ec$, \emph{Faulty-Inputs}), the output $r$-blocks to said $r$-$\rec$ have no $r\error$, and therefore is Correct.

\begin{figure}[h]
    \begin{subfigure}[b]{0.5\textwidth}
\centering
    \includegraphics[width = 0.9\linewidth]{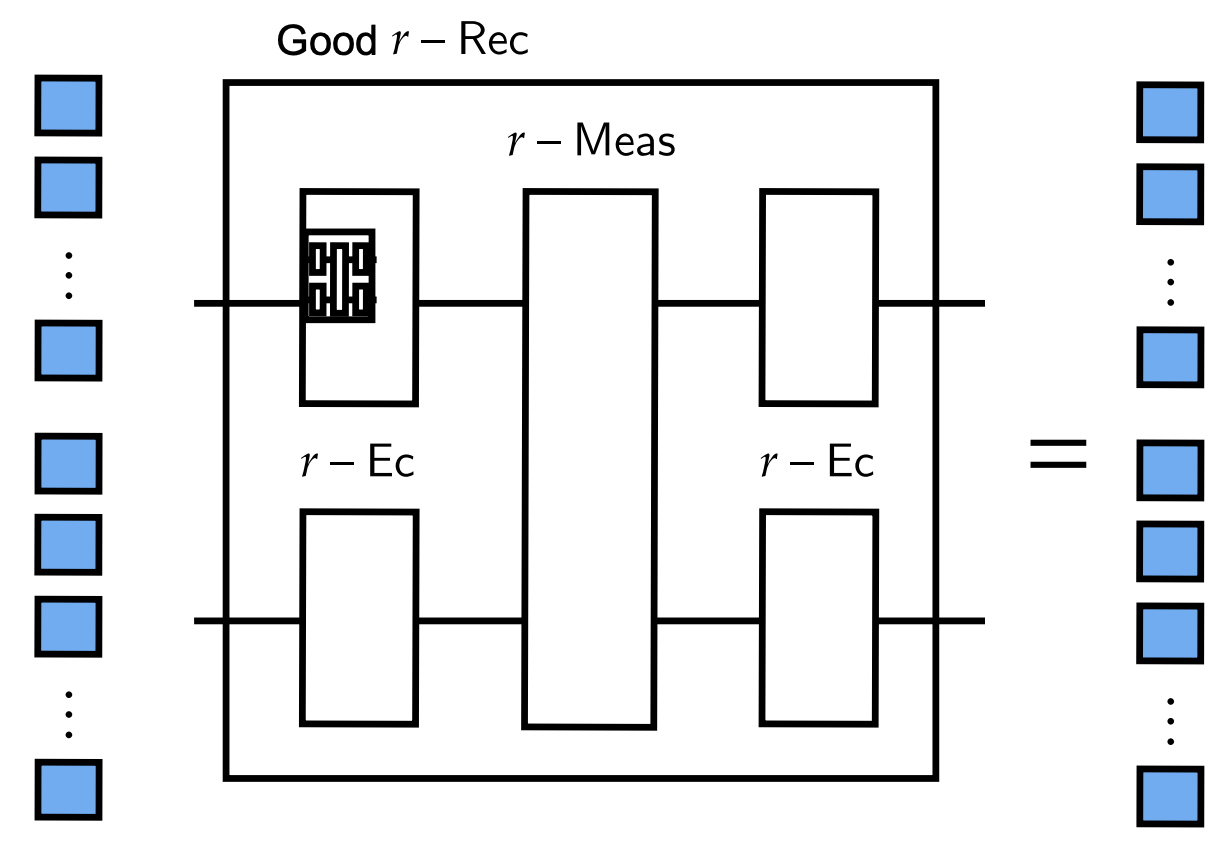}
    \caption{Case 1: Bad $(r-1)$-$\rec$s in the first layer of $r$-$\ec$s.}
\end{subfigure}
\begin{subfigure}[b]{0.5\textwidth}
\centering
    \includegraphics[width = .9\linewidth]{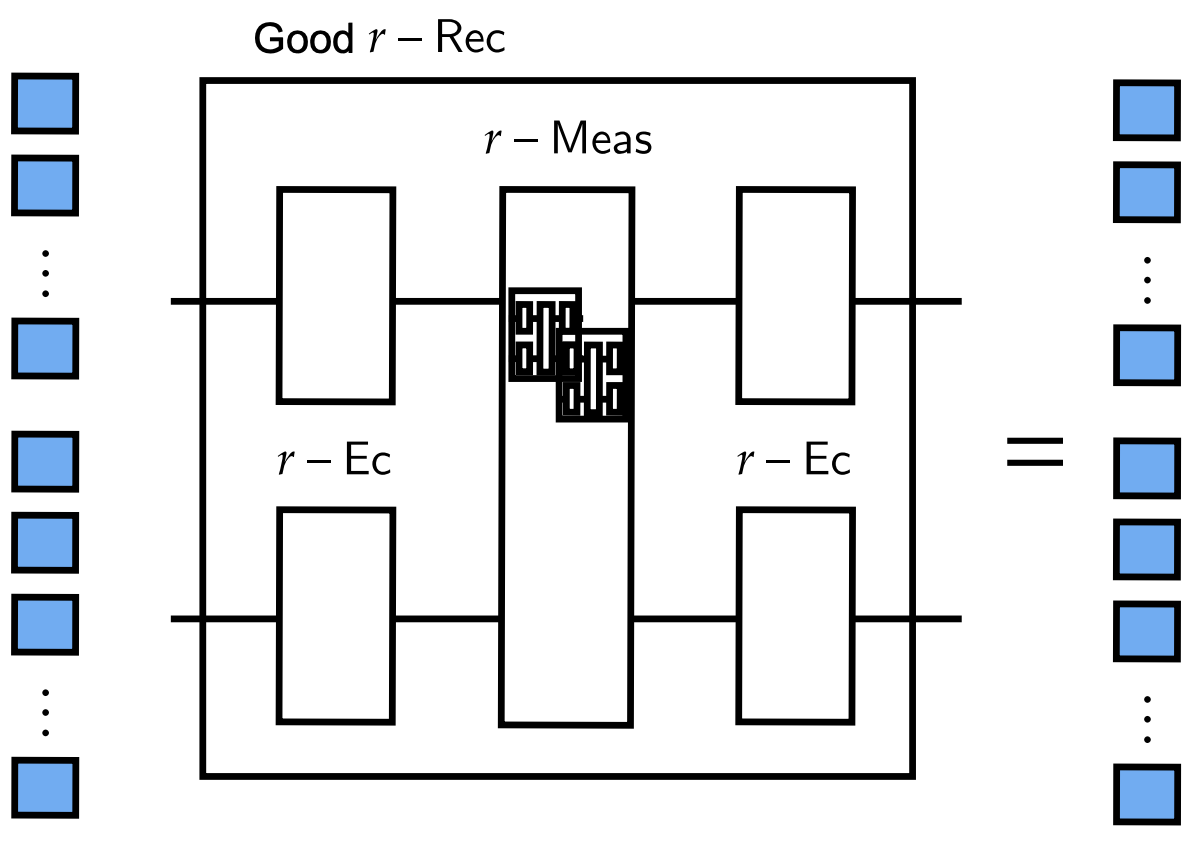}
    \caption{Case 2: Bad $(r-1)$-$\rec$s in the $r$-$\meas$.}
\end{subfigure}
\caption{How a bad $(r-1)$-$\rec$ in a Good $r$-$\rec$ doesn't proliferate errors, Cases 1 and 2. }
\label{fig:goodcorrect12}
\end{figure}

        \item The bad $(r-1)$-$\rec$s lie within the last layer of $r$-$\ec$ gadgets (see \cref{fig:goodcorrect3}, a). Up till that last layer, the $r$-blocks contain no $r$-errors. By ($r$-$\ec$, \emph{Faulty-Execution}), each output $r$-block to that $r$-$\rec$ has $\leq 1$ $r\adj$. Therefore, the $r$-$\rec$ is Correct. 
    \end{enumerate}    
    
    The argument above only addresses the first layer of $r$-$\rec$s, whose inputs have no $r$-errors. However, note that in cases (1, 2), the output $r$-blocks also have no $r$-errors, so similar reasoning could apply to the correctness of the next layer. The non-trivial case lies in $(3)$, which we discuss as follows.

    Consider all the $r$-$\rec$s adjacent to a given $r$-$\rec$ whose bad $(r-1)$-$\rec$s lie within the last layer of $r$-$\ec$ gadgets (i.e., neighboring case 3 above. See \cref{fig:goodcorrect3}, b). Crucially, since by assumption these neighbors are Good, the location of the bad $(r-1)$-$\rec$s in these neighbors is determined to lie in their first layer of $r$-$\ec$s (since they overlap in said locations). Applying ($r$-$\meas$, \emph{Faulty-Inputs}) to these neighboring $r$-$\rec$s, we deduce each of their $r$-blocks has $\leq 1$ $ r\adj$. Since their last layer of $r$-$\ec$ gadgets in the neighboring $r$-$\rec$s contains no bad $(r-1)$-$\rec$, by ($r$-$\ec$, \emph{Faulty-Inputs}) we conclude their output has no $r$-errors, and therefore is Correct. 

    \begin{figure}[h]
    \begin{subfigure}[b]{0.5\textwidth}
\centering
    \includegraphics[width = 0.9\linewidth]{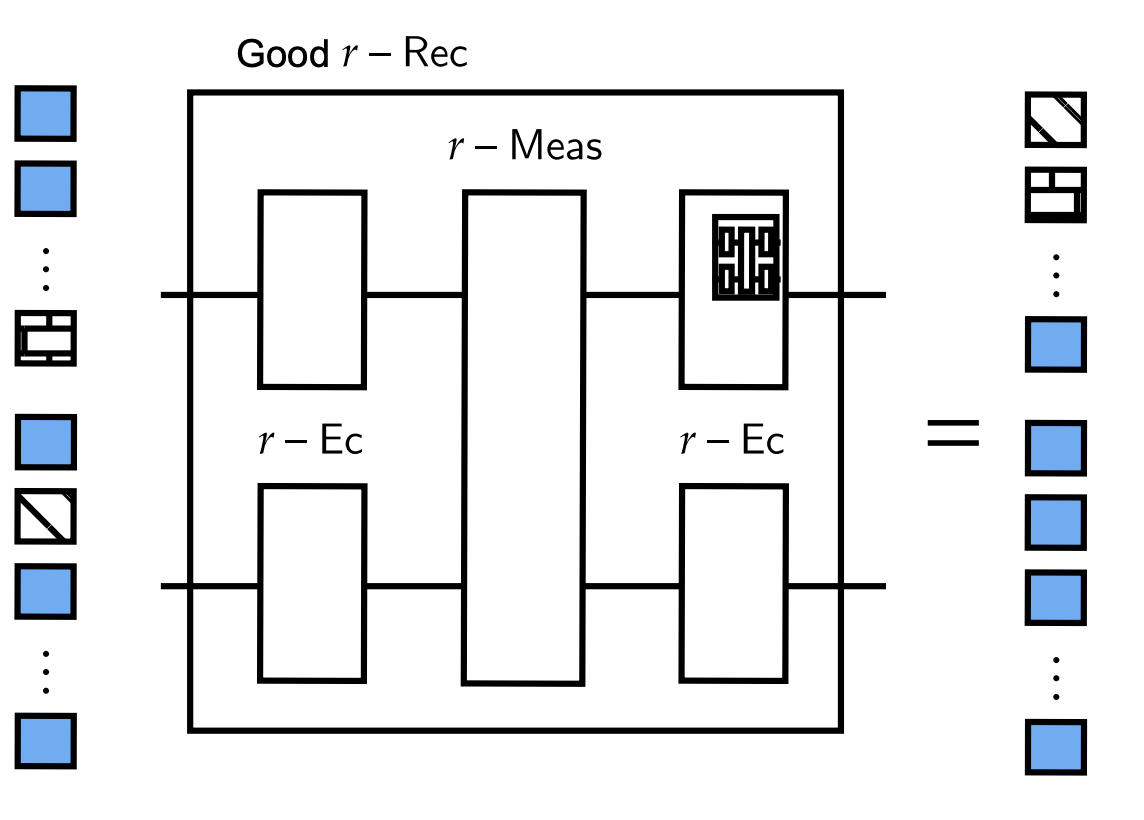}
    \caption{Case 3: Bad $(r-1)$-$\rec$s in the last layer of $r$-$\ec$s.}
\end{subfigure}
\begin{subfigure}[b]{0.5\textwidth}
\centering
    \includegraphics[width = .95\linewidth]{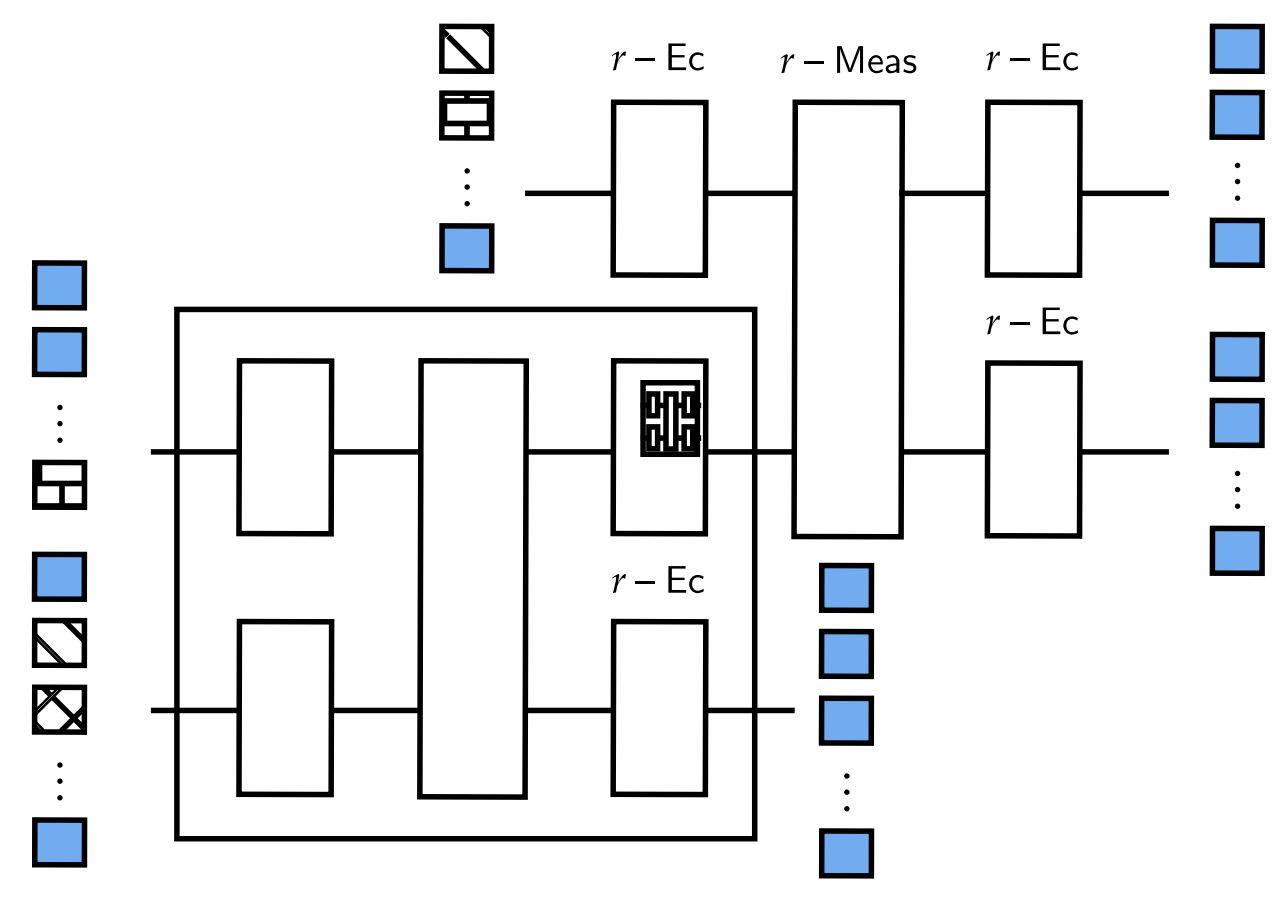}
    \caption{Correctness of $r$-$\rec$s neighboring Case 3.}
\end{subfigure}
\caption{How a bad $(r-1)$-$\rec$ in a Good $r$-$\rec$ doesn't proliferate errors, Case 3. }
\label{fig:goodcorrect3}
\end{figure}
\end{proof}

\subsection{Percolation of Bad $r$-$\rec$s and a Threshold Theorem}
\label{section:percolation}

We dedicate this section to a proof that Bad rectangles are exceedingly rare, as $r$ increases. The proof follows a now-standard percolation argument \cite{Aharonov1996FaulttolerantQC, aliferis2005quantum}.

\begin{lemma}[A Threshold Theorem]\label{lemma:threshold}
    There exists a constant $c\in (0, 1)$ and a threshold noise rate $p^*\in  (0,1)$, such that for all $p<p^*$, and concatenation levels $r\geq 0$, the probability a given $r$-$\rec$ is Bad is $\leq 2^{-2^{c\cdot r}}$.
\end{lemma}

\begin{proof}
    Let $p_r$ be the probability a given $r$-$\rec$ is bad. Recall an $r$-$\rec$ is bad if it contains at least $2$ independent bad $(r-1)$-$\rec$s, and, $(r-1)$-$\rec$s are which are bad and fail independently are independent as random variables. By a union bound, 

    \begin{gather}
        p_r =  \bigg(\text{Choices of $2$ Independent Bad $(r-1)$-$\rec$s}\bigg)\times p_{r-1}^2  = O(V_{r}^2\times p_{r-1}^2)
    \end{gather}

    \noindent Where we assume $V_{r}\times p_{r-1} < 1$. $V_r$ is the number of $(r-1)$-$\rec$s in any $r$-$\rec$ (the ``space-time volume"), which satisfies:

    \begin{equation}
        V_r \leq O\bigg(\frac{T_r}{T_{r-1}} \times \frac{C_r}{C_{r-1}}\bigg) \leq 2^{\poly(r)} 
    \end{equation}

    \noindent We claim, inductively, $p_r  \leq  2^{-2^{c\cdot r}}$ for some constant $c < 1$. Indeed, 

    \begin{equation}
        p_{r+1} = p_r^2\times 2^{O(\poly(r))} \leq 2^{-2\cdot 2^{c\cdot r}}\cdot 2^{O(\poly(r))} = 2^{-2^{c\cdot (r+1)}} \cdot 2^{-(2-2^c)\cdot 2^{cr} + \poly(r)}
    \end{equation}

    \noindent For every $c\in (0, 1)$, there exists a constant $r^*$ such that $\forall r\geq r^*$ we have $(2-2^c)\cdot 2^{cr} \geq  \poly(r)$, thereby satisfying the recursion. It only remains now to satisfy the base case, where $p_{r^*}\leq  2^{-2^{c\cdot r^*}}$; for this purpose, we pick a sufficiently small constant noise rate $p^*$ which implies the base case. For a loose argument, it suffices to pick $p^*$ such that the probability of any single $0$-error in $r^*$-$\rec$ is $p_{r^*}$:

    \begin{equation}
        p_{r^*} \leq p^*\cdot \prod_{k\leq r^*} V_{k}= p^* \cdot 2^{\poly(r^*)} \leq 2^{-2^{c\cdot r^*}}\Rightarrow p^*\equiv 2^{-2^{c\cdot r^*}} \cdot 2^{-\poly(r^*)} .
    \end{equation}

\end{proof}

\section{Proofs for the Error-Propagation Properties}
\label{section:propagation_proofs}

We dedicate this section to a proof that the quantum memory satisfies the desired error-propagation properties required for correctness, at every level $r\geq 1$, assuming they do so at level $0$. The base case of the induction is proved in the subsequent section \cref{section:base_case}. We refer the reader back to \cref{properties:prop1} for a recollection of the desired properties of $r$-$\ec$, $r$-$\meas$. 

Our proof strategy is inductive: we assume that the collection of \cref{properties:prop1} hold at every level $k\leq r$, and show how to combine with the properties of the hamming code at that level (and the measurement circuits) to achieve the properties at level $r$. Instrumental will be to add another assumption to the pile, regarding the behavior of the error correction gadget on arbitrary input states. Informally, we assume that even in the presence of $\leq 1$ fault in the gadget, it takes \textit{any} input back to the code-space (up to a single $r\adj$), similar to \cite{Aharonov1996FaulttolerantQC, aliferis2005quantum}.

\begin{properties}
    [$r$-$\ec$, Arbitrary Inputs]\label{properties:prop2} We assume that on input an arbitrary state, if the error-correction gadget $r$-$\ec$ has at most one non-independent pair of bad $(r-1)$-$\rec$s, then the output $r$-block has $\leq 1$ $r\adj$ (see \cref{fig:arbitrary_inputs}). The output may contain an arbitrary encoded logical state. 
\end{properties}

\begin{figure}[h]
    \centering
    \includegraphics[width=0.5\linewidth]{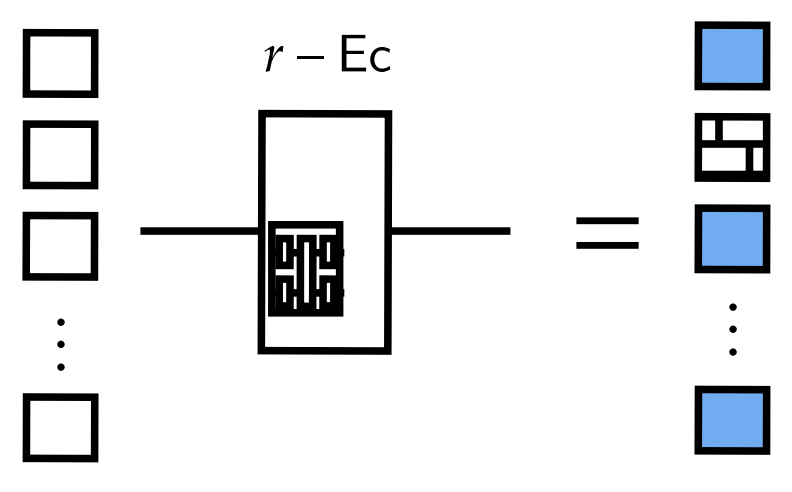}
    \caption{[$r$-$\ec$, Arbitrary Inputs] Any state (white blocks) is converted into a code-state of $C_r$ by $r$-$\ec$ (blue blocks), up to one (adjacent pair of) $r$-error (s).}
    \label{fig:arbitrary_inputs}
\end{figure}

We dedicate the ensuing subsections to an inductive proof of said properties. We begin in \cref{section:faultyinputs-proofs} with a proof of the properties of $r$-$\ec$ and $r$-$\meas$, under faulty inputs. In \cref{section:rho_execution}, we present an interlude and present a key lemma on the faulty execution of the non-fault-tolerant measurement $r$-$\ho$. Then, in \cref{section:faultyexecution-proofs}, we analyze faulty execution of $r$-$\ec$ and $r$-$\meas$. This concludes a proof of \cref{properties:prop1}, conditional on \cref{properties:prop2}. We conclude in \cref{section:rec-inputs} with a proof of [$r$-$\ec$, Arbitrary Inputs] of \cref{properties:prop2}.

\subsection{$r$-$\ec$ and $r$-$\meas$, under Faulty Inputs}
\label{section:faultyinputs-proofs}

Recall that an $(r+1)$-error in an $(r+1)$-block corresponds to an $r$-block contained within that $(r+1)$-block, which is corrupted arbitrarily. To study the behavior of $(r+1)$-$\ec$ and $(r+1)$-$\meas$ under Faulty Inputs, broadly we argue that the these corrupted $r$-blocks can be converted or \textit{simulated} by single-qubit Pauli errors on the underlying Hamming code, which are then encoded into the $r$-blocks. The properties of the error-correction gadgets in combination with the distance 3 guarantees of the Hamming code will then be sufficient to ensure correctness. 

\begin{lemma}
    [$(r+1)$-$\ec$, Faulty Inputs]\label{lemma:ecfaultinputs} Assume  \cref{properties:prop1}, \cref{properties:prop2} hold at every level $k\leq r$. Then, \emph{$(r+1)$-$\ec$, Faulty Inputs} holds at level $r+1$. 
\end{lemma}

\begin{proof}
    Of the input $r$-blocks to the $(r+1)$-$\ec$, at most $2$ adjacent ones contain $(r+1)$-errors. Recall that within $(r+1)$-$\ec$, the first operation to be performed is a layer of $r$-$\ec$s (see \cref{fig:layerofrecs}, LHS), whose $r$-$\rec$s are all Good (by assumption of the Faulty Inputs property). Since each $r$-$\ec$ is in a Good $r$-$\rec$, it contains at most one non-independent pair of bad $(r-1)$-$\rec$s, and thus we can apply [$r$-$\ec$, Arbitrary Inputs] of \cref{properties:prop2}.

    [$r$-$\ec$, Arbitrary Inputs] implies this layer of $r$-$\ec$s converts the $2$ faulty $r$-blocks into arbitrary code-states of $C_r$, possibly up to a single $r\adj$ each. Which, in turn, correspond to logical single-qubit errors on the underlying Hamming codes within $C_{r+1}$, due to the interleaving structure of the code definition in \cref{equation:code_definition} (see \cref{fig:layerofrecs}, RHS); as well as two corrupted cat qubits. 
    
    Since all subsequent $r$-$\rec$s in the $(r+1)$-$\ec$ are also Good (again, by assumption), by the correctness guarantee \cref{lemma:goodcorrect} we are guaranteed the output $r$-blocks contains at most one $r\adj$. That is to say, all the measurements are performed correctly, all the encoded cat states are correctly prepared, and in particular all stabilizer measurements of the underlying Hamming codes are performed correctly. Since the Hamming code is distance $3$, one can recover and correct the logical single-qubit errors on the Hamming codes. After applying a Pauli correction, this results in the desired $(r+1)$-block, with no $(r+1)$-error, proving [$(r+1)$-$\ec$, Faulty Inputs].
\end{proof}

\begin{figure}
    \centering
    \includegraphics[width=1.0\linewidth]{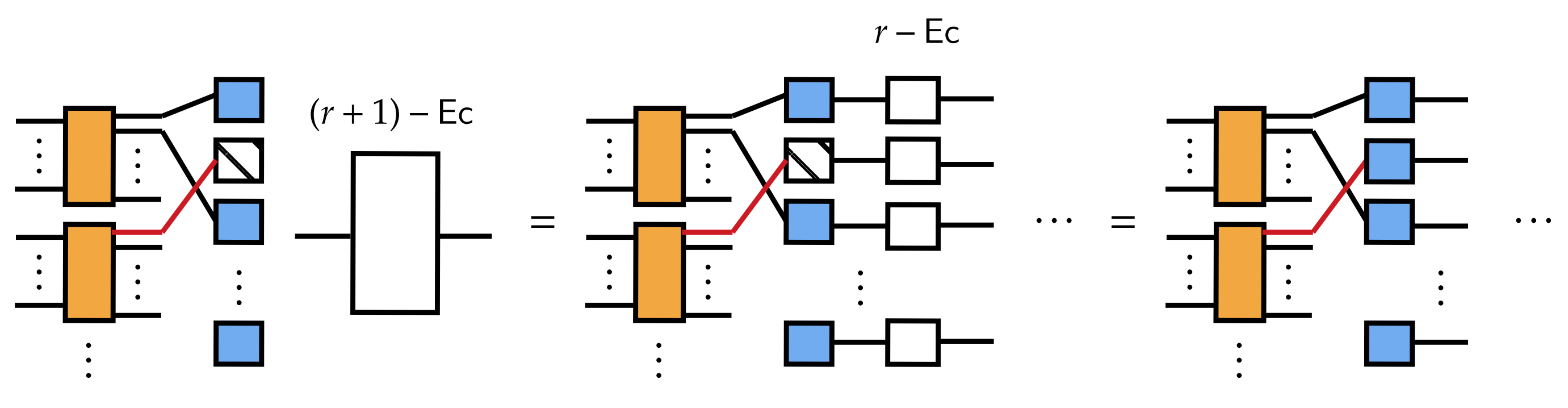}
    \caption{The first operation within a $(r+1)$-$\ec$ is a layer of $r$-$\ec$s. By [$r$-$\ec$, Arbitrary Inputs], this converts arbitrary input $r$-blocks (white stripes) to code-states of $C_r$ with $\leq 1$ $r\adj$ (blue); However, an arbitrary single qubit error is applied to the underlying Hamming codes (red). }
    \label{fig:layerofrecs}
\end{figure}

Let us now turn our attention to the case of $(r+1)$-$\meas$, under faulty inputs. As we discuss, we begin similarly to the above $(r+1)$-$\ec$ case. However, crucial in the fault tolerance of $(r+1)$-$\meas$ will be to ensure that the measurement outcome can be reliably obtained in the presence of faulty-inputs; for that purpose, we appeal to the fact that $(r+1)$-$\meas$ is built by alternating 3 Hookless measurements $(r+1)$-$\ho$ with error correction rounds $(r+1)$-$\ec$. We show that in the presence of faulty inputs (but perfect execution), only the first $(r+1)$-$\ho$ is corrupted, while the redundancy in the next two ensure correctness.

\begin{lemma}
    [$(r+1)$-$\meas$, Faulty Inputs]\label{lemma:meas-faultyinputs} Assume \cref{properties:prop1}, \cref{properties:prop2}  hold at every level $k\leq r$. Then, \emph{$(r+1)$-$\meas$, Faulty Inputs} holds at level $r+1$. 
\end{lemma}

\begin{proof}
    Recall that an $(r+1)$-$\meas$ consists of alternating 3 rounds of Hookless measurements $(r+1)$-$\ho$, with error correction $(r+1)$-$\ec$ rounds. The first layer of operations within $(r+1)$-$\ho$ is a layer of $r$-$\ec$s. Thus, similarly to the proof of \cref{lemma:ecfaultinputs}, by [$r$-$\ec$, Arbitrary Inputs] the adjacent faulty $r$-blocks within the input $(r+1)$-block are converted to arbitrary code-states of $C_r$, each up to a $r\adj$; i.e. there are logical single-qubit errors on the underlying Hamming codes within $C_{r+1}$, but each $C_r$ block is in the code-space up to a $r\adj$ (\cref{fig:layerofrecs}).

    Again, since all subsequent $r$-$\rec$s are Good, \cref{lemma:goodcorrect} ensures that at the output of the $(r+1)$-$\meas$ there remains most one $r\adj$; also that the cat states within the $(r+1)$-$\ho$s are correctly prepared. Next, we show that the measurement information can be reliably recovered. 

    For this purpose, note that the first of three $(r+1)$-$\ho$ is performed on Hamming code code-states with at most one single-qubit error, and therefore the outcome is possibly flipped. Fortunately, \cref{lemma:ecfaultinputs}, property [$(r+1)$-$\ec$, Faulty Inputs] tells us that the subsequent $(r+1)$-$\ec$ corrects the underlying single-qubit errors on the Hamming codes, resulting in an $(r+1)$-block with no $(r+1)$-errors. The next two $(r+1)$-$\ho$ are therefore fault-free, which ensures a correct measurement outcome after taking majority. See \cref{fig:rmeas_proof} for a diagram of the flow of $(r+1)$-errors after each step. This concludes the proof of [$(r+1)$-$\meas$, Faulty Inputs].
\end{proof}

\begin{figure}
    \centering
    \includegraphics[width=1\linewidth]{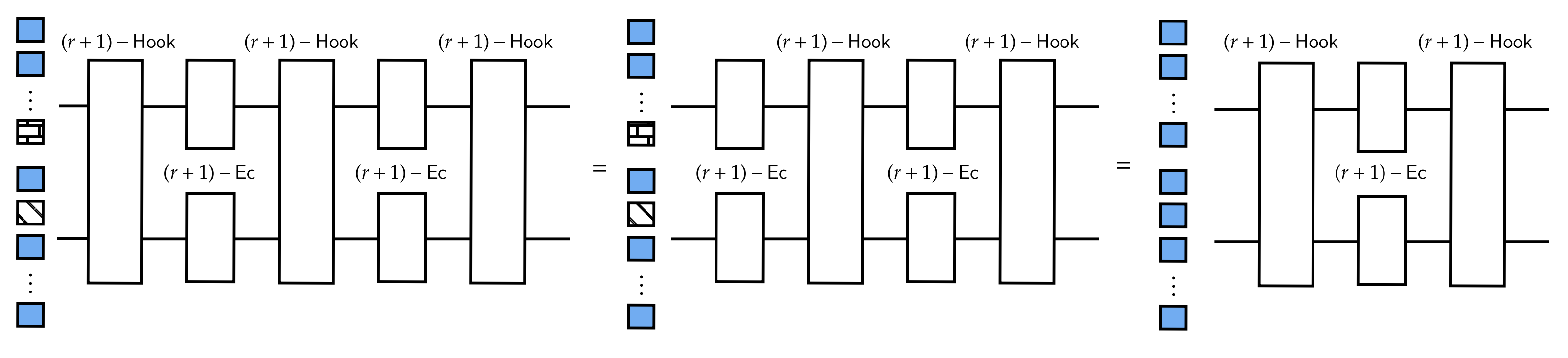}
    \caption{The $(r+1)$-error propagation in the proof of [$(r+1)$-$\meas$, Faulty Inputs]}
    \label{fig:rmeas_proof}
\end{figure}

\subsection{Interlude: $r$-$\ho$, under Faulty Execution}
\label{section:rho_execution}

To study property ($(r+1)$-$\ec$, Faulty Execution) and ($(r+1)$-$\meas$, Faulty Execution), we need to understand the structure of the circuit in the presence of a bad $r$-$\rec$ (or non-independent pair), and, in particular, the effect of such a fault on a Hookless measurement. The crux of the argument will be to understand this bad rectangle as simulating a faulty-measurement on single-qubits of the underlying hamming codes.  

Before doing so, however, we require a short fact on the structure of the cat preparation circuit via repeated ZZ measurements. 

\begin{fact}
    [Cat State preparation via the Repetition Code]
    \label{fact:ghz_error_propagation} Consider the dissipative implementation of Shor's measurement gadget in \cref{fig:hook_dissipative}, in the presence of an error channel applied to 2 adjacent cat-qubits. Then, the output state has errors on at most 2 adjacent data-qubits, but, there is no guarantee on whether the measurement is performed correctly. 
\end{fact}

It is instructive to consider the case of Pauli errors first: Pauli errors on two-qubit Pauli measurements have the effect of possibly flipping the measurement
outcome, but do not propagate to other adjacent qubits. The repeated measurements ensure that the output of the cat state preparation phase is simply a genuine cat-state with 2 adjacent Pauli errors. 

\begin{proof}
    Following the intuition above, we consider the error channel through its Krauss decomposition as a linear combination (superposition) of Pauli errors. After each measurement in the circuit, the relative phases in the Krauss decomposition change but their support does not propagate. In this manner, before any feedback, the state at the end of the cat state preparation circuit has at most two Pauli errors. It remains to show that the inferred feedback operation does not increase the error weight. 

    Here, we simply use basic guarantees of the repetition code circuit. The effect of any Pauli $X$ error, even in superposition, is to flip parity measurement outcomes. Observe that any single Pauli $X$ has the effect of flipping the parity of the checks incident on it; any two adjacent Pauli $X$'s have the effect of flipping the parity of the checks above and below said qubits (but not the shared check, as $XX$ commutes with $ZZ$). Since the circuit is implementing the repetition code, by inspection, can readily identify the location of the $\leq 2$ $X$ errors up to trivial degeneracies. The feedback operation thereby doesn't increase the number of corrupted data-qubits. 
\end{proof}

In the above we remark that the error channels may collapse the cat state into computational basis strings. However, the repetition is ensuring we are able to decode the resulting state back into the span of $\ket{0}^n, \ket{1}^n$ \textit{or}, if the noise occurs in the last layer, simply view the resulting state as a cat state with noise on two qubits. In the former case, the measurement on the data isn't even performed at all.

   \begin{lemma}
    [$(r+1)$-$\ho$, Faulty Execution]\label{lemma:rhook} Assume \cref{properties:prop1}, \cref{properties:prop2} hold at every level $k\leq r$. Suppose the input $(r+1)$-blocks to an $(r+1)$-$\ho$ contain no $(r+1)$-errors, and $(r+1)$-$\ho$ contains at most a non-independent pair of bad $r$-$\rec$s. Then, the output state contains at most one $(r+1)\adj$, and may return an arbitrary measurement outcome. 
\end{lemma}

\begin{proof}
    We refer the reader back to \cref{fig:hook_recursive} for the recursive definition of the Hookless measurement. Following \cite{aliferis2005quantum}, let us first consider the case in which there is only a single bad $r$-$\rec$. We prove that the $\leq$ two adjacent output $r$-blocks of said $r$-$\rec$ may be arbitrarily corrupted, (and thus correspond to $(r+1)$-errors), however, they propagate to the output of the $(r+1)$-$\ho$ without further increasing the number of $(r+1)$-errors. This gives the desired claim. 
    
    To do so, we first claim that at the output of said single bad $r$-$\rec$, the corresponding $r$-blocks have been mapped to an arbitrary encoded state, with at most $1$ $r\adj$; corresponding to single-qubit corruptions on the underlying Hamming code and 2 errors on adjacent cat qubits (this is akin to \cref{fig:layerofrecs}). To see this, recall that the bad $r$-$\rec$ must contain at least $2$ independent bad $(r-1)$-$\rec$s; let us consider their locations. By assumption, each of the two last $r$-$\ec$s in the bad $r$-$\rec$ must contain at most one bad $(r-1)$-$\rec$. Otherwise, the subsequent $r$-$\rec$ would also be bad (and non-independent). This enables us to apply the property [$r$-$\ec$, Arbitrary Inputs] of \cref{properties:prop2} on the output $r$-$\rec$s, and obtain the claimed guarantee on the output $r$-blocks.

    After that bad $r$-$\rec$, we have effectively imparted an arbitrary error map on the 2 cat state qubits encoded into the corresponding $r$-blocks. During the encoded cat state preparation circuit in $(r+1)$-$\ho$, we are effectively simulating $ZZ$ measurements on those cat state qubits, which via \cref{fact:ghz_error_propagation} do not propagate the errors. The resulting output state thereby remains with a single pair of adjacent $(r+1)$-errors. 

    Let us now revisit the case of a non-independent pair of bad $r$-$\rec$s. We claim that of the $\leq$ three output $r$-blocks, only two of them contain $(r+1)$-errors. Indeed, by definition, the former bad $r$-$\rec$ must be Good if the shared $r$-$\ec$ is removed. Therefore, the former must contain at most one other bad $(r-1)$ outside of the shared $r$-$\ec$, this ensures that the non-shared $r$-block contains at most one $r$-error; and no $(r+1)$-error. In turn the $r$-blocks of the latter bad $r$-$\rec$ can be treated analogously to the case of a single bad $r$-$\rec$.
\end{proof}

\subsection{$r$-$\ec$ and $r$-$\meas$, under Faulty Execution}
\label{section:faultyexecution-proofs}

Equipped with the behavior of the Hookless measurement, we are now in a position to prove property ($(r+1)$-$\ec$, Faulty Execution) and ($(r+1)$-$\meas$, Faulty Execution).

\begin{lemma}
    [$(r+1)$-$\ec$, Faulty Execution]\label{lemma:rec-execution} Assume \cref{properties:prop1}, \cref{properties:prop2} hold at every level $k\leq r$. Then, \emph{$(r+1)$-$\ec$, Faulty Execution} holds at level $r+1$. 
\end{lemma}

\begin{proof}

    WLOG, the bad $r$-$\rec$ (or non-independent pair thereof) lies within the Hookless measurement $(r+1)$-$\ho$ of a stabilizer of the underlying Hamming codes. Following  [$(r+1)$-$\ho$, Faulty Execution] of \cref{lemma:rhook}, we are guaranteed at the output of that $(r+1)$-$\ho$, the output state contains at most one $(r+1)\adj$, and may return an arbitrary measurement outcome. 

    Fortunately, all subsequent $r$-$\rec$s are Good. [$r$-$\ec$, Arbitrary Inputs] then implies the subsequent layer of $r$-$\ec$s collapses the adjacent $(r+1)$-errors into single-qubit errors on the underlying Hamming code. Moreover, (1) via the correctness lemma \cref{lemma:goodcorrect} the output state (before any Pauli correction) contains these same these single-qubit errors, and thus also contains at most one $(r+1)\adj$; (2) all subsequent cat state preparation steps are successful after that faulty $(r+1)$-$\ho$, and thus all subsequent stabilizer measurements correctly measure the syndrome of the single-qubit errors.

    It only remains to show that the Pauli correction which is decoded, doesn't increase the number of errors. In contrast to the proof of \cref{lemma:ecfaultinputs}, we are not guaranteed all the stabilizer measurements will agree.\footnote{Since the fault doesn't necessarily lie in the beginning of the $(r+1)$-$\ho$ Hookless measurements.}  Nevertheless, for this purpose, we can divide into cases on the pattern of faults.

    \begin{itemize}
        \item If all the $(r+1)$-$\ho$ measurements corresponding to the same Hamming code stabilizer measurements agree, then either no correction is needed or the correct Pauli correction is inferred, and the resulting state has no $(r+1)$-errors.

        \item If all stabilizer measurements of the first two rounds of stabilizer measurements agree, then no correction is applied (by assumption, $(r+1)$-$\ec$ is fault-less on input), and the output state still has at most 1 $(r+1)\adj$.

        \item If not all stabilizer measurements of the first two rounds of stabilizer measurements agree, then the fault must have occurred during either of these rounds. This ensures the last round is correct, and those outcomes are used to apply the correction. The resulting state has no $(r+1)$-errors.
    \end{itemize}

    In the first and third point above we use that the Hamming code is distance $3$, thus the correct single qubit operations on the Hamming codes are applied.

    %Now let us return to the case of \textit{a pair} of non-independent bad $r$-$\rec$s. As discussed, this pair must lie within some $r$-$\ec$. Analogously, the $r$-block which enters/exits that $r$-$\ec$ may be corrupted arbitrarily resulting in an $(r+1)$-error. If this is the last $r$-$\ec$ in the $(r+1)$-$\ec$, we are done; if not, after the subsequent $r$-$\meas$, it may arbitrarily corrupt both $r$-blocks involved, and simply reduces to the case above. 
\end{proof}

\begin{lemma}
    [$(r+1)$-$\meas$, Faulty Execution]\label{lemma:rmeas-execution} Assume \cref{properties:prop1}, \cref{properties:prop2} hold at every level $k\leq r$. Then, \emph{$(r+1)$-$\meas$, Faulty Execution} holds at level $r+1$. 
\end{lemma}

\begin{proof}
    Recall that $(r+1)$-$\meas$ consists of 3 alternating rounds of Hookless measurements $(r+1)$-$\ho$ and error-correction rounds $(r+1)$-$\ec$. If the bad $r$-$\rec$ (or non-independent pair thereof) is contained in one of the $(r+1)$-$\ec$ rounds, then from [$(r+1)$-$\ec$, Faulty Execution] of \cref{lemma:rec-execution}, after that $(r+1)$-$\ec$, the output state contains at most a single $(r+1)\adj$. From the proof of \cref{lemma:rhook}, the subsequent $(r+1)$-$\ho$ may output an incorrect result, but it only propagates the $(r+1)\adj$. The proceeding $(r+1)$-$\ec$ has no bad $r$-$\rec$s, and thereby  [$(r+1)$-$\ec$, Faulty Inputs] ensures the other two $(r+1)$-$\ho$ will be correct.

    In turn, if the bad $r$-$\rec$s are contained in a $(r+1)$-$\ho$, \cref{lemma:rhook} ensures the the output state of that $(r+1)$-$\ho$ contains at most at most a single $(r+1)\adj$, and reports an arbitrary measurement outcome. If proceeded by a $(r+1)$-$\ec$ with no bad $r$-$\rec$s, then [$(r+1)$-$\ec$, Faulty Inputs] ensures the other two $(r+1)$-$\ho$ are correct. This concludes the proof of [$(r+1)$-$\meas$, Faulty Execution].
\end{proof}

\subsection{$r$-$\ec$ under Arbitrary Inputs}
\label{section:rec-inputs} 

We are now in a position to prove \cref{properties:prop2}. 

\begin{lemma}
    [$(r+1)$-$\ec$, under Arbitrary Inputs]\label{lemma:rec-arbitrary} Assume \cref{properties:prop1}, \cref{properties:prop2} hold at every level $k\leq r$. Then, \emph{$(r+1)$-$\ec$, under Arbitrary Inputs} holds at level $r+1$. 
\end{lemma}

\begin{proof}
    Recall $(r+1)$-$\ec$ consists of 3 repetitions of hookless measurements of all the stabilizers $(r+1)$-$\ho$. The $(r+1)$-$\ec$ may contain a bad pair of non-independent $r$-$\rec$; here we divide into cases on its location and its effect on the outcomes of the $(r+1)$-$\ho$. 

   First, suppose the bad pair of non-independent $r$-$\rec$s lies in the first repetition. Then, we have two consecutive, consistent rounds of $(r+1)$-$\ho$ (of all stabilizers) which consist entirely of Good $r$-$\rec$s, acting on an arbitrary input. Following the proofs above, since each $(r+1)$-$\ho$ first consists of a layer of $r$-$\ec$s, after said layer, the output state consists of a generic state encoded into the copies of $C_r$, up two non-adjacent $r$-errors, by [$r$-$\ec$, Arbitrary Inputs]. Next, the subsequent $r$-$\meas$ are contained in Good $r$-$\rec$s, and therefore the cat states are prepared correctly, and the stabilizer measurements are implemented correctly. This projects the encoded state in the $r$-blocks into an arbitrary (but known) syndrome subspace of the underlying Hamming codes, which can be brought back to the code-space via a Pauli correction. The resulting state has no $(r+1)$-errors.

   Next, suppose the bad pair of non-independent $r$-$\rec$s lies in the third repetition. Similar reasoning to the above ensures that the first two repetitions are consistent and project the state encoded into the $r$-blocks into a known syndrome subspace of $C_{r+1}$, which can be corrected (by computing majority of the syndrome measurements). From \cref{lemma:rhook}, the last repetition of stabilizer measurements $(r+1)$-$\ho$, in the presence of bad $r$-$\rec$s, incurs an additional pair of adjacent $(r+1)$-errors. Fortunately, these errors do not propagate, analogously to \cref{lemma:rec-execution}.

   Finally, suppose the bad pair of non-independent $r$-$\rec$s lies in the middle repetition. Here we must further divide into cases on the measurement outcomes of $(r+1)$-$\ho$. 

   \begin{itemize}
       \item If all the stabilizer measurements in the first two repetitions agree, then the syndrome of the state encoded within the $r$-blocks (before the fault) has been correctly determined. If we could apply this correction before the fault, that would result in a codestate of $C_{r+1}$ with no $(r+1)$-errors. Instead, we must apply it after $(r+1)$-$\ec$, where (again analogously to \cref{lemma:rec-execution}) the $(r+1)$-errors introduced by the fault remain on the same $r$-blocks and we are left with an $(r+1)$-$\ec$ with at most one $(r+1)\adj$.

       \item If all the stabilizer measurements in the first two repetitions do not agree, then we are guaranteed that the third repetition of stabilizer measurements is correct. Those outcomes are applied to infer the underlying syndrome, resulting in an output state with no $(r+1)$-errors.
   \end{itemize}

\end{proof}

\section{Fault-tolerance at Level 0}
\label{section:base_case}

We dedicate this section to showing that the bottom level of the construction, $C_0$, admits the desired error-propagation properties. We refer the reader back to \cref{fig:hook_unitary} for the unitary implementation of the Hookless measurement.

\subsection{Faulty Inputs}

The Hookless measurement circuit, $0$-$\ho$, is comprised of a network of $\cnot$ gates which implement the repeated $ZZ$ measurements. It will be helpful to recap the Pauli propagation properties of $\cnot$ gates.

\begin{fact}
    [$\cnot$ Pauli Propagation]\label{fact:cnot_prop} We use the following circuit identities:
    \begin{gather}
        (Z\otimes \mathbb{I}) \cnot_{1, 2}  = \cnot_{1, 2} (Z\otimes \mathbb{I})\quad ( \mathbb{I}\otimes X) \cnot_{1, 2}  = \cnot_{1, 2} (\mathbb{I} \otimes X) \\
        (\mathbb{I}\otimes Z) \cnot_{1, 2}  = \cnot_{1, 2} (Z\otimes Z)\quad ( X\otimes \mathbb{I}) \cnot_{1, 2}  = \cnot_{1, 2} (X \otimes X) 
    \end{gather}
\end{fact}

Let us begin to analyse the error-propagation properties of $C_0$ with the simplest case - the case of faulty inputs - where a $0$-block $C_0$ has just two adjacent Pauli errors, and they are input into a fault-tolerant error correction or measurement gadget.

\begin{lemma}
    [$0$-$\ec$, $0$-$\meas$, Faulty Inputs]\label{lemma:base-fi}Suppose the physical qubits of an input $0$-block contain at most $2$ adjacent errors. If input to a $0$-$\ec$ or a $0$-$\meas$  which contains no faulty gates, then, the output state contains no errors, and the measurement returned is correct.
\end{lemma}

\begin{proof}
    Note that the first $0$-$\ho$ circuit begins with measurements on the entangling and cat-state qubits, so if the $0$-block contains at most $2$ adjacent errors then after said layer it can only contain a single error on a data qubit. Now, let us recollect the $\cnot$ error propagation properties \cref{fact:cnot_prop} and the structure of the circuit \cref{fig:hook_unitary}. Pauli $Z$ errors on the data qubits propagate down to the cat state qubit during each application of a $\cnot$ gate; Since there are an even number of them, the Pauli $Z$ on the data qubit simply propagates to the end. Similar reasoning shows Pauli $X$ errors also simply propagate to the end. 

    We conclude all syndrome measurements during the $0$-$\ec$ are performed correctly, obtaining the desired property [$0$-$\ec$, Faulty Inputs]. Moreover, while the first $0$-$\ho$ in an $0$-$\meas$ may report the wrong outcome (due to the data-error), the subsequent two are screened by a $0$-$\ec$ and thereby are fault-less. 
\end{proof}

\subsection{Faulty Execution}

The following lemma computes the ``data damage distance" of the unitary implementation of the Hookless measurement, and quantifies the resilience of $0$-$\ho$ against faulty execution. This is analogous to \cref{lemma:rhook} at the physical level.

\begin{lemma}
    [$0$-$\ho$, Faulty Execution]\label{lemma:0ho}
    Suppose the physical qubits of two adjacent $0$-blocks contain no $0\error$s. If input to a $0$-$\ho$ which contains at most one 2-qubit fault, then each output $0$-block of the $0$-$\ho$ has at most $1$ $0\adj$. 
\end{lemma}

In the below, we develop a painstaking case analysis of two-qubit faults of the repetition code circuit of \cref{fig:hook_unitary}, using $\cnot$ gates to implement next-nearest-neighbor gates. We refer the reader back to \cref{fig:hook_unitary} for a diagram of the unitary Hookless measurement.

Our approach first analyzes the effect of a single Pauli error on the cat-state preparation circuit and measurement outcome then, we understand the effect of the Pauli feedback performed to correct the cat state, and the final measurement. The case of two, \textit{adjacent} Pauli errors is then reduced to that of a single Pauli error, due to the structure of the circuit. Finally, we generalize the argument to superpositions of Pauli errors.

\begin{proof}
    As discussed, we begin by dividing into cases on the location and type of a single Pauli error, whether on a data, cat-state, or ``entangling" qubit.  Note that all $\cnot$ gates on cat-state qubits are controlled on said qubits, all $\cnot$ gates on entangling qubits are targeted on said qubits. 

    \begin{enumerate}
        \item If the error occurs on a data qubit. $Z$ errors propagate down to $Z$ errors on the cat-state, where they remain unmodified until the end of the circuit and flip the final measurement outcome. $X$ errors propagate up to $X$ errors on the entangling qubits, and flip at most one parity measurement.

        \item If the error occurs on an entangling qubit. $X$ errors simply propagate and flip the next parity measurement. Depending on their location, Z errors either propagate to 1) $Z$ errors on the cat state qubits immediately above and below, flipping their corresponding measurement outcomes; 2) the data qubit and cat state qubit immediately below.

        \item If the error occurs on a cat-state qubit. $Z$ errors simply propagate to the end of the circuit, and flip the final measurement outcome. $X$ errors propagate both to 1) the entangling qubit immediately below, where they flip all subsequent parity measurements. 2) to the data-qubit above, and subsequently the next entangling qubit immediately above. This applies an error to the data-qubit, and flips all the subsequent parity measurements above and below the original cat-state qubit.
        \end{enumerate}

        In each of three cases, at most one error is imparted to a data-qubit. However, it remains to ensure the effect of the Pauli feedback performed to correct the cat state doesn't increase the error weight on the data qubits. It suffices the understand the effect of $X$ errors above, as those are the ones which flip the $ZZ$ measurements. 

        In cases 1 and 2 above, the $X$ error has the effect of flipping a single ZZ measurement, and has no effect on the cat state. By majority, the relevant Pauli feedback is correctly computed. In case 3 above, an $X$ error flips all subsequent parity measurements on the two neighboring entangling qubits. Regardless of the location of this fault, either the fault is correctly identified and the cat state contains no $X$ error, or the cat state contains at most one $X$ error adjacent to a possibly already faulty data-qubit. In either case, the resulting output state has at most one data error. 

        By \cref{fact:cnot_prop}, two adjacent Pauli errors of the same type $X^{\otimes 2}$, $Z^{\otimes 2}$ before/after a $\cnot$ gate are equivalent to a single Pauli error after/before the same gate. If they are of different types $X\otimes Z$, their propagation is independent and we also reduce to the cases above. Case by case one identifies at most one data-qubit can be effected. 
        
        Finally, generic error channels are first decomposed into their Krauss decompositions, and treated as linear combinations of Pauli errors. If proceeded by a measurement (i.e. if the error occurs in the middle of the circuit, and not the end), then the superposition over errors collapses accordingly. In which case, the output state is a superposition over errors on a single data-qubit as desired. 
\end{proof}

We can now prove the base case properties of $0$-$\ec$, $0$-$\meas$ under faulty execution. 

\begin{lemma}
    [$0$-$\ec$, $0$-$\meas$, Faulty Execution]\label{lemma:base-execution} Suppose the physical qubits of two input $0$-blocks contain no $0\error$s. If input to a $0$-$\ec$ or a $0$-$\meas$ which contains at most one fault, then the output state contains at most 1 $0\adj$, and the measurement returned is correct.
\end{lemma}

\begin{proof}
    Equipped with \cref{lemma:0ho}, the proof is the same as that of \cref{lemma:rec-execution} of [$r$-$\ec$, Faulty Execution] and \cref{lemma:rmeas-execution} [$r$-$\meas$, Faulty Execution]. As a sketch, we divide into cases on the location of the $0$-$\ho$ which contains the fault in execution. [$0$-$\ec$, Faulty Execution] follows by applying \cref{lemma:0ho} and \cref{fact:cnot_prop} to ensure the fault only propagates to a single $0\adj$ at the output of the faulty $0$-$\ho$; careful consideration as before implies repeating the stabilizer measurements thrice is sufficient to infer the error up to degeneracy. [$0$-$\meas$, Faulty Execution] follows by applying a combination of [$0$-$\ec$, Faulty Execution], [$0$-$\ec$, Faulty Inputs] with [$0$-$\ho$, Faulty Execution] and \cref{fact:cnot_prop} analogously to \cref{lemma:rmeas-execution}. 
\end{proof}

\subsection{Arbitrary Inputs}

\begin{lemma}
    [$0$-$\ec$, Arbitrary Inputs] Suppose the physical qubits of an input $0$-block lie in an arbitrary quantum state. If input to a $0$-$\ec$ which contains at most one fault, then the output state corresponds to a code-state of $C_0$ up to a $0\adj$. 
\end{lemma}

\begin{proof}
    We divide into cases on the location of the faulty $0$-$\ho$. If the fault lies in the first sequence of stabilizer measurements, then the subsequent two are faultless and consistent, and by majority correctly project and correct the state into a code-state of $C_0$ with no $0\error$s. If the fault lies in the last sequence of stabilizer measurements, then the first two are faultless and consistent, and by majority by majority correctly project and correct the state into a code-state of $C_0$. By \cref{lemma:0ho}, the faulty $0$-$\ho$ may create a $0\adj$, but they won't propagate. 

    If the fault lies in the middle sequence of stabilizer measurements, similar reasoning to \cref{lemma:rec-execution} applies. If the middle sequence and the first sequence differ, then we are guaranteed the last sequence is fault-free and acts on an arbitrary input; using those measurement outcomes ensures the output state is a code-state of $C_0$ with no $0\adj$. Conversely, if the middle sequence and the first sequence agree, those syndromes can be used to correct the state back to a code-state of $C_0$; however, up to a $0\adj$ created by \cref{lemma:0ho} which propagates to the end. 
\end{proof}

\section{Proof of \cref{theorem:main}}
\label{section:proof_of_main}

In this section we combine the ingredients developed in the previous sections, and prove \cref{theorem:main}. To define our code $M_n$, we will ``chunk" the $n$ physical qubits of $M_n$ into blocks of size $b$, and independently use each block to represent an instance of the code $C_r$ of \cref{section:interleaved_memory}. As we discuss, this subdivision enables us to decrease the time-overhead of the construction without sacrificing rate nor significantly sacrificing logical error rate.

\begin{theorem}
    $\forall n>n_0$ and constant $\delta\in (0, 1/3]$, there exists a quantum memory $M_n$ satisfying:

    \begin{enumerate}

        \item $M_n$ is a $[[n, > n/20, \exp(\Theta(\log^{\delta} n))]]$ stabilizer code.
        \item $M_n$ is implemented by a stabilizer circuit using nearest-neighbor gates on a line of qubits while subjected to uniform depolarizing circuit noise.

        \item Below some depolarizing noise threshold $p$ independent of $n$, the per-circuit-cycle logical error rate of $M_n$ is $$ \exp(-\exp(\Omega(\log^{\delta} n))).$$

         \item Each circuit cycle has circuit depth $T_n = \exp(\Theta(\log^{3\delta}n))$.
    \end{enumerate}
\end{theorem}

\begin{proof}

    $M_n$ will be comprised of adjacent instances of $C_r$, where we pick $r = \log^{\delta} n$. \cref{theorem:code_parameters} on the code parameters tells us the $r$th level of concatenation $C_r$ has blocklength $b = 2^{r^2/2+O(r)}$ and rate $ > 1/20$. The rate of $M_n$ is therefore 

    \begin{equation}
        > \bigg\lfloor\frac{n}{b}\bigg\rfloor \cdot b\cdot \text{Rate}(C_r)\cdot \frac{1}{n} \geq \bigg(1 - \frac{b}{n}\bigg) \cdot \text{Rate}(C_r) > \frac{1}{20}
    \end{equation}

    \noindent for sufficiently large $n$. The time-overhead for $C_r$ is presented in \cref{lemma:time-overhead}. Assuming the classical system is capable of parallel operations, the instances of $C_r$ can also be decoded in parallel.

    The results of \cref{section:base_case}, namely \cref{lemma:base-fi} and \cref{lemma:base-execution}, prove that the level $0$ code, $C_0$, satisfies the desired error-propagation properties \cref{properties:prop1}, \cref{properties:prop2}. Using $C_0$ as a base case, the results of \cref{section:propagation_proofs}, namely \cref{lemma:ecfaultinputs}, \cref{lemma:rec-execution}, \cref{lemma:meas-faultyinputs}, \cref{lemma:rmeas-execution}, and \cref{lemma:rec-arbitrary}, prove $C_r$ satisfies \cref{properties:prop1}, \cref{properties:prop2} at all $r\geq 1$.

    From the percolation argument in \cref{lemma:threshold}, we know that below some threshold noise rate $p^*$, the probability per-circuit-cycle there exists a Bad $r$-$\rec$ within the execution of a single $C_r$ block decays doubly exponentially with the concatenation level $r$. The total number of $r$-$\rec$'s per circuit-cycle is $\leq \lceil n/b \rceil = O(n)$. By a union bound, and from the relationship between $n$ and $r$, we conclude the probability any $C_r$ block has a bad $r$-$\rec$ is:

    \begin{equation}
        n\cdot 2^{-2^{\Theta(r)}} = \exp\bigg( - \exp\bigg( \Theta(\log^\delta n)\bigg)\bigg)
    \end{equation}

    \noindent assuming $\delta$ a constant, and $n$ sufficiently large. 

    Finally, by the correctness \cref{lemma:goodcorrect} and the proof above of \cref{properties:prop1}, if every $r$-$\rec$ in the circuit is Good, then every $r$-$\rec$ is correct; consequently, by \cref{lemma:good-decodable} the logical information is decodable. 
\end{proof}

\section{Fault-tolerant Computation}
\label{section:the_computer}

\subsection{Outline}
\label{subsection:computation-outline}

To begin, we recollect that the fault-tolerant measurement scheme of \cref{section:fault-tolerance} enables us to perform arbitrary Pauli-product measurements; which, in turn, enables us to implement arbitrary stabilizer operations. It remains to show how to perform any non-Clifford operation, to complete a universal set of quantum gates. 

We proceed by showing how to implement a magic state distillation procedure, which via gate-injection, enables us to implement $T$ gates. To do so, we need to perform a minor modification to the layout of our memory, which will decrease its rate by a negligible amount. In between the top, $r$-level code-blocks which store data-qubits, we place a ``ladder" of code blocks $(C_0, C_1, \cdots, C_{r-1})$ adjacent to each other (See \cref{fig:ladder}). 

This ladder of code-blocks will serve to inject noisy $T$ gates from the physical level ($C_0$), all the way into a logical qubit of $C_r$. Although we defer details to \cref{section:tstate-distillation}, roughly speaking the protocol is based on teleporting logical qubits within an $i$-block, into an $(i+1)$-block, all the way up the ladder. Once sufficiently many noisy $T\ket{+}$ states are encoded into the top-level $r$-block, a magic state distillation routine is run (inside $C_r$) to create a high-fidelity $T\ket{+}$ state. 

The only remaining detail is how to perform operations between top-level $r$-blocks, which no longer are adjacent. For this purpose, we develop a basic shuffling/shuttling protocol in \cref{section:shuttling}, which moves an $r$-block over (noisy) $\ket{0}$ qubits until adjacent to the relevant other $r$-block.

\begin{figure}[h]
    \centering
    \includegraphics[width=0.7\linewidth]{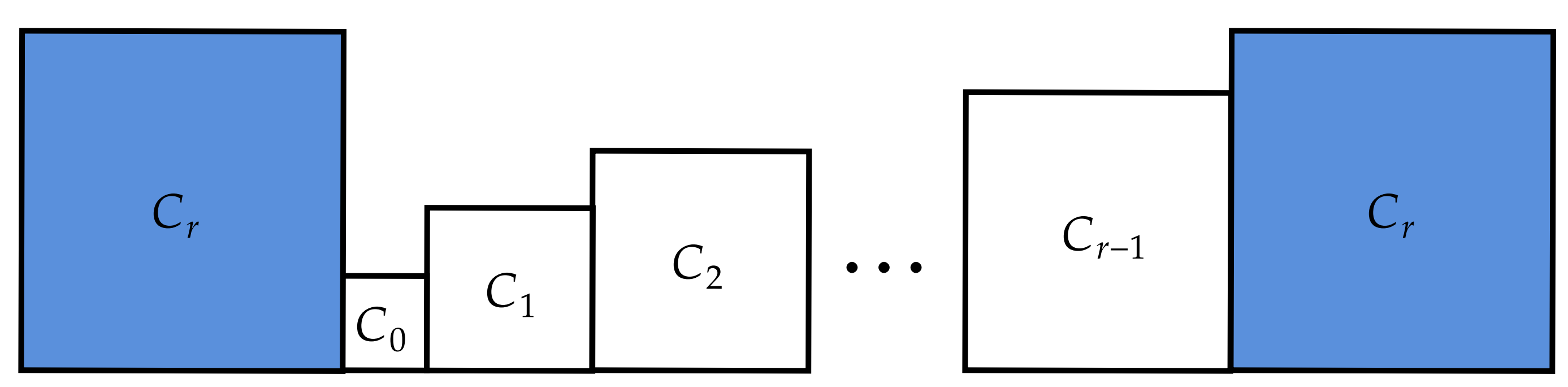}
    \caption{The ladder of code-blocks used to teleport $T$ gates into an $r$-block. Pictured in blue are $r$-blocks used to store data-qubits. In white, the blocks used to create and distill $T$ states. }
    \label{fig:ladder}
\end{figure}

\subsection{Code-block Shuffling}
\label{section:shuttling}

Here we describe a simple scheme to move the top-level code-blocks $C_r$ over a distance of $b$. Simply put, we interleave swap gates with error-correction rounds. We recall that at the physical level, $C_0$ consists of data-qubits, ``entangling" and ``measurement" qubits, the latter two initialized to $\ket{0}$. These helper qubits assist in the shuttling protocol: at each layer of swap gates, we apply swap gates only between the data-qubits and the entangling qubit immediately to their left (or to their right). In doing so, we can move all the data-qubits simultaneously; the effect of the swapping is then to simply swap the role of the entangling and measurement qubits. 

\begin{lemma}
    [Runtime of the Shuttling Protocol]\label{lemma:shuttling_runtime} The runtime to move a $C_r$ block over a linear distance of $b$ qubits, all initialized to $\ket{0}$, is $b\cdot \exp (\Theta(r^3))$.
\end{lemma}

In the context of the ladder construction, the linear distance $b$ is the sum of blocklengths of $C_i$, $i\leq r$, which is $\exp(\Theta((r-1)^2)) = o(|C_r|)$.

\subsection{The $T$ State-Distillation Protocol}
\label{section:tstate-distillation}

\subsubsection{Phase 1: $T$ State Teleportation}

We devise the following algorithm to teleport a single noisy $T\ket{+}$ state into a code-block $C_r$. Roughly speaking, it suffices to show how to teleport a $T\ket{+}$ state encoded into a code-block $C_i$, into a logical qubit of an adjacent $C_{i+1}$ block. To do so, we proceed in 3 steps. 
\begin{enumerate}
    \item \textbf{Cat State Preparation.} Establish a multi-partite cat state, between a logical qubit of $C_i$, and the reserved qubits of each $C_i$ block within $C_{i+1}$. We do so using the repeated $ZZ$ measurement scheme within the protocol for $r$-$\ho$, described in \cref{section:fault-tolerance}. (\cref{fig:tgates_phase1}a)

    \item \textbf{Bell State Preparation.} Prepare a logical EPR pair across the different code-blocks $C_i, C_{i+1}$, where one of two qubits is encoded into the $C_i$ block, and the other is encoded into $C_{i+1}$. To do so, it suffices to perform a logical $XX$ measurement across the two code-blocks. Equipped with the cat states, it suffices to perform logical measurements $i$-$\meas$/$(i+1)$-$\meas$ within the individual code-blocks. 

    \item \textbf{$T$-State Teleportation.} Given a logical EPR pair across the different code-blocks $C_i, C_{i+1}$ and a $T$ state encoded into $C_i$, Pauli measurements and feedforward suffice to teleport the state into $C_{i+1}$. (\cref{fig:tgates_phase1}b)
\end{enumerate}

Applying the scheme above sequentially to each pair of adjacent codes in the ladder $C_0, C_1, \cdots, C_{r}$, teleports a single $T\ket{+}$ state into a logical qubit of the top-most code $C_r$. We claim (and prove shortly) that in the presence of a small constant amount of noise $p\leq p^*$, the resulting fidelity of the logical $T\ket{+}$ state is $1-\Theta(p)$. After the teleportation up the ladder concludes, we repeatedly measure all its qubits in the computational basis; until the next teleportation phase on said block commences.

\begin{figure}[h]
    \begin{subfigure}[b]{0.6\textwidth}
\centering
    \includegraphics[width = 1\linewidth]{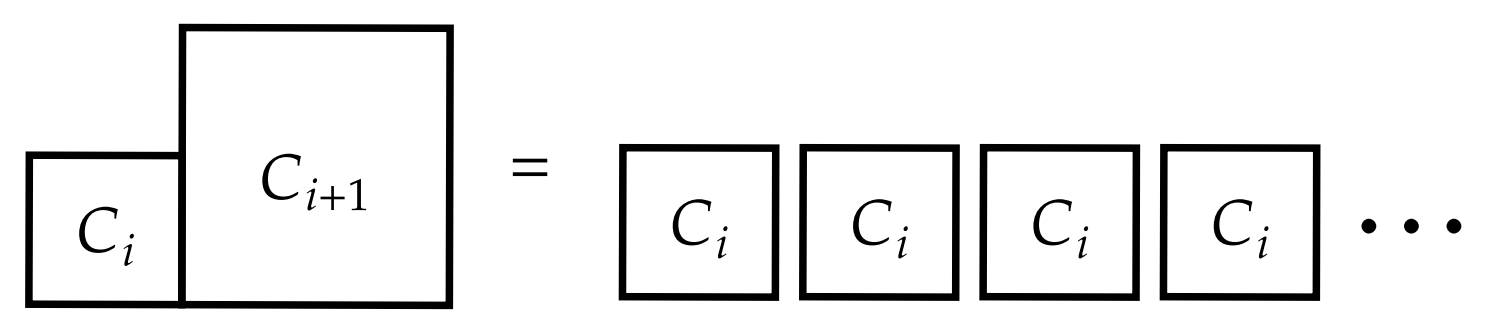}
    \caption{An adjacent pair of $C_i, C_{i+1}$ blocks.}
\end{subfigure}
\begin{subfigure}[b]{0.4\textwidth}
\centering
    \includegraphics[width = 0.6\linewidth]{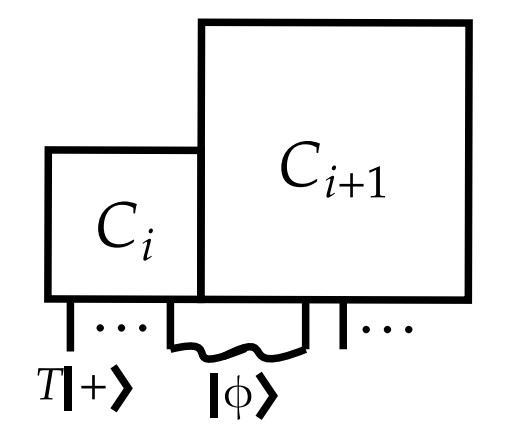}
    \caption{Teleporting $T\ket{+}$ from $C_i$ into $C_{i+1}$.}
\end{subfigure}
\caption{Phase 1 - A protocol to teleport noisy $T$ states into a top level code-block $C_r$.}
\label{fig:tgates_phase1}
\end{figure}

\begin{lemma}
    [Runtime of the Teleportation Phase]\label{lemma:runtime_teleportation} The runtime to teleport a single $T\ket{+}$ state encoded into $C_0$ at the bottom of the ladder, into the $C_r$ block at the top of the ladder, is $\exp (\Theta(r^3))$.
\end{lemma}

This is simply since the protocol consists of a constant number of $i$-$\meas$ and $i$-$\ec$ at each layer $0\leq i\leq r$.

\subsubsection{Phase 2: $T$ Gates via Distillation and Gate Injection}

We require the following protocol for magic state distillation. %Given the overhead in implementing operations within $C_r$, there is no need for an optimal protocol.

\begin{fact}[\cite{bravyi2005distillation}]\label{fact:sd_fact}
    Let $\rho$ be a single qubit mixed state with some constant fidelity $f \geq f^*$ with $T\ket{+}$. Then, there exists a stabilizer circuit which on input $\polylog(\delta^{-1})$ copies of $\rho$, outputs a single qubit with fidelity $1-\delta$ with $T\ket{+}$.
\end{fact}

Leveraging the teleportation protocol of Phase 1, we encode $\polylog(\delta^{-1})$ $T$ states into the top-level $C_r$ block, each with a small constant amount of independent noise. \footnote{The block-length of $C_r$ will later be chosen to be $\exp (\poly\log\log(n/\delta))$, which is asymptotically larger than $\polylog(\delta^{-1})$.} It remains to distill these noisy $T$ states into one high-fidelity $T$ state, and then via state injection apply a $T$ gate to the data-qubits in the computation.

\begin{enumerate}
    \item \textbf{$T$ State Distillation.} Using logical stabilizer operations and Pauli feed-forward as described in \cref{subsection:computation-outline}, we implement the magic state distillation protocol of \cref{fact:sd_fact} within $C_r$. 

    \item \textbf{Shuttling $C_r$ blocks.} The $r$ blocks containing the data-qubits, and the $T$ state, are not adjacent. We move the code-blocks sideways using the shuttling protocol of \cref{section:shuttling} until the relevant $r$ blocks are adjacent.

    \item \textbf{$T$ Gate Injection.} Via logical controlled $Z$ measurements and a Clifford correction within the $r$ blocks, implement the $T$ gate injection.
\end{enumerate}

The scheme above imparts a single $T$ gate to a single qubit in an adjacent data $C_r$ block. We repeat this protocol a number of times proportional to the number of logical qubits in $C_r$, to implement a single layer of $T$ gates.

\subsection{Time and Space Overhead}

\begin{lemma}\label{lemma:comp_time_overhead}
    The time overhead to perform a single layer of logical $T$ gates and nearest-neighbor $2$ qubit gates to per-gate-error $\delta$, over data qubits encoded into a linear arrangement of $C_r$ blocks, is $\exp(\Theta(r^3))\cdot \polylog\frac{1}{\delta}$.
\end{lemma}

\begin{proof}
    Each 2-qubit Clifford operation within each code-block must be performed sequentially, and their time-overhead is precisely that of the error-correction routine: $T_r =\exp(\Theta(r^3))$. The runtime of all the gates in the block is then $|C_r|\cdot T_r =\exp(\Theta(r^3))$. The runtime of qubit shuttling (\cref{lemma:shuttling_runtime}) is subsumed by that of the distillation step, discussed below.

     By \cref{fact:sd_fact}, to distill a single $T$ gate to error $\delta$, we need to teleport $\polylog\frac{1}{\delta}$ noisy $T$ states, and then run a $\polylog\frac{1}{\delta}$-sized logical circuit within $C_r$. We invoke this protocol $\leq |C_r|$ times, to perform a layer of $T$ gates within a data-block of $C_r$. The resulting time-overhead is $|C_r|\cdot \polylog\frac{1}{\delta}\cdot \exp(\Theta(r^3)) = \exp(\Theta(r^3))\cdot \polylog\frac{1}{\delta}$ by \cref{lemma:runtime_teleportation}.   
    
\end{proof}

\begin{lemma}\label{lemma:space_overhead}
    So long as the target logical error rate satisfies $d/\epsilon \leq \exp\exp O(\log^{1/3} m)$, the rate of the resulting computer is $\geq 1/20$.
\end{lemma}

\begin{proof}
    The outline for the computer architecture follows the memory construction of \cref{section:proof_of_main}. The $m$ logical qubits are divided into blocks of size $b$, and encoded separately into instances of $r$-blocks. $r$ is chosen such that the target logical error rate per gate is $\delta = \poly(\epsilon/md)$, i.e. $r = \log\log \delta^{-1}$, s.t. the block-length is $\exp \Theta(\log^2\log \delta^{-1})$. The constraint on the target logical error rate ensures this blocklength is $o(m)$.
    
    Within each block, an extra $\polylog\delta^{-1}$ logical qubits are reserved for magic state distillation, and an extra $\exp(O((r-1)^2))$ qubits are reserved for the $T$ state teleportation. The resulting rate is
    \begin{gather}
        > \frac{b}{m}\cdot  \bigg\lfloor\frac{m}{b}\bigg\rfloor \cdot \bigg( \text{Rate}(C_r)- \frac{\polylog(\delta^{-1})}{b}\bigg) \cdot \bigg(1-\frac{\exp(O((r-1)^2))}{b}\bigg) \\ \geq \text{Rate}(C_r) - \polylog(\delta^{-1})\cdot e^{-\Theta(r^2)}-e^{-\Theta(r)} = \text{Rate}(C_r) - o(1) \geq 1/20 
    \end{gather}

    \noindent where at last we leverage \cref{theorem:code_parameters}.    
\end{proof}

\subsection{Fault-tolerance}

We begin with a simple correctness lemma, which states that if at each level $i$ in the teleportation protocol, there are no $i\error$s, then the $T\ket{+}$ state is teleported correctly. 

\begin{lemma}
    [$i$-block teleportation]\label{lemma:teleportation-correctness} Assume \cref{properties:prop1}, \cref{properties:prop2}. Let $C_i$ be an $i$-block with an encoded $T\ket{+}$ state, and let $C_{i+1}$ be the adjacent $(i+1)$-block in the ladder. Assume that every $i$-block within the pair $(C_i, C_{i+1})$ contains no $i\adj$, and that the entire teleportation circuit contains no bad $i$-$\rec$s. Then, the output $i$-blocks contain no bad $i$-$\rec$s. 
\end{lemma}

This follows immediately from the error-propagation properties \cref{properties:prop1}, \cref{properties:prop2}. This correctness lemma implies a bound on the fidelity of a single $T\ket{+}$ state teleportation.

\begin{lemma}
    [$T\ket{+}$ state teleportation] Assume properties \cref{properties:prop1}, \cref{properties:prop2}. There exists a constant threshold noise rate $p^*\in (0, 1)$, such that for all $p\leq p^*$ and concatenation level $r\geq 0$, the teleportation protocol of \cref{section:tstate-distillation} encodes a $T\ket{+}$ state into $C_r$ with fidelity $1-\Theta(p)$. 
\end{lemma}

\begin{proof}
    By \cref{lemma:teleportation-correctness}, the $T\ket{+}$ is teleported correctly across all code-blocks in the ladder if at each step $i$ there does not exist any bad $i$-$\rec$s. The $i$th teleportation step consists of a $2^{\poly(i)}$ number of $i$-$\rec$s. This tells us $T\ket{+}$ is teleported correctly with probability all but
    \begin{gather}
        \mathbb{P}[\exists i: \exists \text{ Bad } i-\rec \text{ during }i-\text{block teleportation}]  \leq \\\leq  \sum_i 2^{\poly(i)} \cdot p_i^2  \leq  \sum_i  2^{\poly(i)} \cdot 2^{-2^{ci}} \leq \Theta(p).
    \end{gather}

    Where we used \cref{lemma:threshold} on the percolation of bad $r$-$\rec$s.
\end{proof}

\begin{corollary}
   [Magic State Distillation]\label{lemma:state-distillation} Assume properties \cref{properties:prop1}, \cref{properties:prop2}, and condition on the absence of any bad $(r-1)$-$\rec$s during the teleportation and distillation protocols. Then, for all noise rates $p\leq p^*$, the protocols distills a single $T\ket{+}$ state to fidelity $1-\delta$ within $C_r$.
\end{corollary}

\begin{proof}
    Conditioned on the absence of any bad $(r-1)$-$\rec$s within the various teleportation protocols, then $\polylog(\delta^{-1})$ $T\ket{+}$ states are teleported to within $C_r$ with fidelity $1-\Theta(p)$. In fact, since the existence of bad $\rec$s during the teleportation of each $T\ket{+}$ state is independent, we have the stronger guarantee that a $1-\Theta(p)$ fraction of the states teleported are exactly $T\ket{+}$ with probability all but $\exp(-\polylog(\delta^{-1}))$. 
    
    We invoke a slight strengthening of the $T\ket{+}$ distillation protocol of \cref{fact:sd_fact}, which ensures the resulting $T\ket{+}$ is produced with fidelity $1-\delta$ given sufficiently many noiseless $T\ket{+}$ states (but the remaining states may be adversarially chosen).

    \begin{fact}[\cite{bravyi2005distillation}, under adversarial inputs]\label{fact:sd_adversarial}  Given an $k=\polylog(\delta^{-1})$ qubit state, where an unknown $(1-f) > (1-f^*)$ fraction of the qubits are exact $T\ket{+}$ states, the protocol of \cref{fact:sd_fact} outputs a single qubit with fidelity $1-\delta$ with $T\ket{+}$.
\end{fact}
\end{proof}

We are now in a position to conclude the proof of our Fault-tolerance theorem. 

\begin{proof} 

[of \cref{theorem:main-ft}]
    We consider the circuit execution in layers. We first convert the logical circuit into an architecture of $d \cdot \polylog(\delta)$ layers of nearest neighbor 2-qubit stabilizer gates, and single qubit $T$ gates.
    
    By \cref{lemma:state-distillation}, the correctness guarantees of \cref{section:propagation_proofs}, and an appropriate choice of $\delta$, conditioned on the absence of any bad $(r-1)$-$\rec$s, each stabilizer gate is implemented perfectly, and each $T$ gate is implemented with fidelity $\delta$. The entire logical circuit is then implemented with error $\delta\cdot m\cdot d\cdot \polylog(\delta)$, again under the conditioning. Via the percolation argument \cref{section:percolation}, the existence of a bad $r$-$\rec$ in the entire logical circuit (including the code-block shuffling) is $$\leq d\cdot m\cdot \polylog(\delta) \cdot 2^{2^{-\Theta(r)}} \leq \delta d\cdot m\cdot \polylog(\delta) ,$$

    \noindent under the appropriate choice of $r = \log\log \delta^{-1}$. A choice of $\delta = \poly(\epsilon/nd)$ concludes the proof of correctness.

    The time and space overheads derived in \cref{lemma:comp_time_overhead}, \cref{lemma:space_overhead} conclude the parameters of the result. 
\end{proof}

\section{Contributions and Acknowledgments}
\label{section:contributions}

Craig created the initial construction and wrote some informal arguments for its correctness. Thiago formalized the proof of correctness and wrote the majority of the paper.

The authors thank Shankar Balasubramanian, Marharyta Davydova,  Louis Golowich, Matt McEwen, and Quynh Nguyen for valuable comments. 

\section{Conclusion}
\label{section:conclusion}

In this paper, we showed that geometrically local stabilizer circuits are fundamentally more powerful than geometrically local stabilizer codes.
In particular, we showed that stabilizer circuits maximally violate the Bravyi-Poulin-Terhal bound~\cite{bravyi2010stabilizerbound2}, by constructing a 1D-local circuit implementing a fault tolerant quantum memory with a constant coding rate. 

There are many ways that it might be possible to improve on our construction.
For example: the threshold could be improved \cite{Yoshida2024ConcatenateCS} and the coding rate could be improved. Most notably, the quasi-polylog time overhead in our construction, inherited from the tower of hamming codes of \cite{Yamasaki_2024}, could be particularly limiting. Typically, a fault tolerant construction is only considered efficient if it has polylogarithmic overheads. More generally, what are the true limits on the \emph{spacetime} overhead of fault tolerant stabilizer circuits in low dimensions?

A particularly interesting improvement would be to achieve fault tolerant quantum computation with constant spatial overhead under even stricter constraints, such as requiring the operations to be \textit{translationally invariant}, akin to the results of \cite{Gcs1983}.

Our results suggest that bounds proven for stabilizer codes do not necessarily generalize to stabilizer circuits.
For example, stabilizer codes need to be at least three dimensional to support constant depth non-Clifford gates~\cite{Bravyi2013}.
Perhaps a stabilizer circuit can loosen this bound in some way.
In general, we recommend caution when assuming a bound proven for stabilizer codes will limit the behavior of real world quantum computers.

\printbibliography

%\appendix   
%\include{sections/bfs}

\end{document}